\documentclass[a4paper,11pt]{article}
\pdfoutput=1

\usepackage{xcolor}
\definecolor{MyDarkBlue}{rgb}{0.15,0.25,0.45}
\usepackage{graphicx}
\usepackage[textwidth = 430 pt, textheight = 630 pt]{geometry}

\usepackage[utf8]{inputenc}

\usepackage{amsmath,amssymb,amsthm}
\usepackage{amsfonts}
\usepackage{mathrsfs}
\usepackage{bbm}
\usepackage[normalem]{ulem}
\usepackage{booktabs,longtable}
\usepackage{latexsym}
\usepackage{aliascnt}     

\usepackage{hyperref}
\hypersetup{
    hypertexnames=false,
    colorlinks=true,
    citecolor=MyDarkBlue,
    linkcolor=MyDarkBlue,
    urlcolor=MyDarkBlue,
    pdfauthor={},
    pdftitle={},
    pdfsubject={hep-th math-ph},
    breaklinks=true
}

\newcommand{\sharedtheorem}[2]{%
    \newaliascnt{#1}{theorem}
    \newtheorem{#1}[#1]{#2}
    \aliascntresetthe{#1}
}

\theoremstyle{plain}       
\newtheorem{theorem}{Theorem}[section]   
\sharedtheorem{lemma}{Lemma}
\sharedtheorem{proposition}{Proposition}
\sharedtheorem{corollary}{Corollary}
\sharedtheorem{conjecture}{Conjecture}

\theoremstyle{definition}  
\sharedtheorem{definition}{Definition}
\sharedtheorem{example}{Example}
\sharedtheorem{question}{Question}
\sharedtheorem{notation}{Notation}

\theoremstyle{remark}      
\sharedtheorem{remark}{Remark}

\numberwithin{equation}{section}

\usepackage[nameinlink,noabbrev]{cleveref}

\crefname{theorem}{theorem}{theorems}
\crefname{lemma}{lemma}{lemmata}
\crefname{proposition}{proposition}{propositions}
\crefname{corollary}{corollary}{corollaries}
\crefname{conjecture}{conjecture}{conjectures}
\crefname{definition}{definition}{definitions}
\crefname{example}{example}{examples}
\crefname{notation}{notation}{notations}
\crefname{question}{question}{questions}
\crefname{remark}{remark}{remarks}

\Crefname{theorem}{Theorem}{Theorems}
\Crefname{lemma}{Lemma}{Lemmata}
\Crefname{proposition}{Proposition}{Propositions}
\Crefname{corollary}{Corollary}{Corollaries}
\Crefname{conjecture}{Conjecture}{Conjectures}
\Crefname{definition}{Definition}{Definitions}
\Crefname{example}{Example}{Examples}
\Crefname{notation}{Notation}{Notations}
\Crefname{question}{Question}{Questions}
\Crefname{remark}{Remark}{Remarks}

\usepackage{tikz}
\usetikzlibrary{matrix,cd,arrows}
\usepackage{mathtools}

\tikzcdset{scale cd/.style={every label/.append style={scale=#1},
        cells={nodes={scale=#1}}}}

\let\fn\footnote
\renewcommand{\footnote}[1]{\linespread{1.1}\fn{#1}\linespread{1.29}}

\makeatletter
\renewcommand{\section}{\@startsection
    {section}{1}{\z@}{-3.5ex plus -1ex minus
        -.2ex}{2.3ex plus .2ex}{\mathversion{bold}\bf }}
\renewcommand{\subsection}{\@startsection{subsection}{2}{\z@}{-3.25ex
        plus -1ex minus
        -.2ex}{1.5ex plus .2ex}{\mathversion{bold}\bf }}
\renewcommand{\subsubsection}{\@startsection{subsubsection}{3}{-2.45ex}{-3.25ex
        plus -1ex minus -.2ex}{1.5ex plus .2ex}{\it }}
\renewcommand\paragraph{\@startsection{paragraph}{4}{\z@}%
    {3.25ex \@plus1ex \@minus.2ex}%
    {-1em}%
    {\normalfont\normalsize\bfseries\mathversion{bold}}}

\renewcommand{\@seccntformat}[1]{\@nameuse{the#1}.~~}

\makeatletter
\newcommand{\groupappendicesintoc}{%
    \addtocontents{toc}{\protect\addvspace{1em}}%
    \phantomsection\addcontentsline{toc}{section}{Appendices}%
    \let\oldaddcontentsline\addcontentsline
    \renewcommand\addcontentsline[3]{%
        \def\@tempa{##2}\def\@tempb{section}%
        \ifx\@tempa\@tempb
        \oldaddcontentsline{##1}{subsection}{##3}%
        \else
        \oldaddcontentsline{##1}{##2}{##3}%
        \fi
    }%
}
\makeatother

\newcommand{\appendices}{
    \section*{Appendix}\label{appendices}\setcounter{subsection}{0}
    \appendix
    \crefalias{section}{appendix}
    \crefalias{subsection}{appendix}
    \crefname{appendix}{appendix}{appendices}
    \Crefname{appendix}{Appendix}{Appendices}
    \groupappendicesintoc
}

\newcommand{\makecommand}[3]{%
    \foreach \i in #3 {%
        \expandafter\xdef\csname #1\i\endcsname{\noexpand#2{\unexpanded\expandafter{\i}}}%
    }%
}

\makecommand{fr}{\mathfrak}{{c,der,g,h,A,X,gl,sl,o,so,osp,u,su,sp,spin,string,Poisson,td,tb}}

\makecommand{sf}{\mathsf}{{id,String,Lie,Hom,SU,SO,Sp,Cat,Set,Spin,Pin,inv,CE,Diff,HSymp,TD,TB,pr,Aut,GL,SL,IO,Mat,Sym,Inn,Out,AUT,Vect}}

\makecommand{rm}{\mathrm}{{Ham,diag,im,Ad,dim}}

\newcommand{\latinalphabet}{A,a,B,b,C,c,d,D,E,e,F,f,G,g,H,h,I,i,J,j,K,k,L,l,M,m,N,n,O,o,P,p,Q,q,R,r,S,s,T,t,U,u,V,v,W,w,X,x,Y,y,Z,z}
\newcommand{\caplatinalphabet}{A,B,C,D,E,F,G,H,I,J,K,L,M,N,O,P,Q,R,S,T,U,V,W,X,Y,Z}
\makecommand{I}{\mathbbm}{\latinalphabet}
\makecommand{bf}{\mathbf}{\latinalphabet}
\makecommand{bm}{\bm}{\latinalphabet}
\makecommand{ca}{\mathcal}{\caplatinalphabet}
\makecommand{fr}{\mathfrak}{\latinalphabet}
\makecommand{rm}{\mathrm}{\latinalphabet}
\makecommand{sf}{\mathsf}{\latinalphabet}
\makecommand{sc}{\mathscr}{\latinalphabet}

\newcommand{\acton}{\triangleright}

\def\slasha#1{\setbox0=\hbox{$#1$}#1\hskip-\wd0\hbox to\wd0{\hss\sl/\/\hss}}

\def\periodb#1{\setbox0=\hbox{$#1$}#1\hskip-\wd0\hbox to\wd0{-}}

\newcommand{\dpar}{\partial}

\newcommand{\unit}{\mathbbm{1}}   			

\newcommand{\eand}{~~~\mbox{and}~~~}

\newenvironment{myenumerate}{
    \vspace{-2mm}\begin{enumerate}
        \setlength{\itemsep}{-1mm}
    }{\vspace{-2mm}\end{enumerate}}

\makeatletter
\AtBeginDocument{%
    \let\oldref\ref
    \renewcommand*{\ref}{\@ifstar{\oldref*}{\cref}}%
    \renewcommand{\eqref}[1]{{\rm (\oldref{#1})}}
}
\makeatother

\newcommand{\eps}{\varepsilon}

\newcommand{\sfTBLFT}{\mathsf{T}_2\mathsf{BL}^{F_2}}
\newcommand{\sfTBRFT}{\mathsf{T}_2\mathsf{BR}^{F_2}}
\newcommand{\sfTDFT}{\mathsf{T}_2\mathsf{D}^{F_2}}
\newcommand{\sfTBLFTT}{\mathsf{T}_2\mathsf{BL}^{F_3}}
\newcommand{\sfTBRFTT}{\mathsf{T}_2\mathsf{BR}^{F_3}}
\newcommand{\sfTDFTT}{\mathsf{T}_2\mathsf{D}^{F_3}}
\newcommand{\sfTBLFTA}{\mathsf{T}_2\mathsf{BL}^{F_2,\ltimes}}
\newcommand{\sfTBRFTA}{\mathsf{T}_2\mathsf{BR}^{F_2,\ltimes}}
\newcommand{\sfTDFTA}{\mathsf{T}_2\mathsf{D}^{F_2,\ltimes}}
\newcommand{\sfTBLFTR}{\mathsf{T}_2\mathsf{BL}^{F_2,{\rm red}}}
\newcommand{\sfTBRFTR}{\mathsf{T}_2\mathsf{BR}^{F_2,{\rm red}}}
\newcommand{\sfTDFTR}{\mathsf{T}_2\mathsf{D}^{F_2,{\rm red}}}
\newcommand{\sfTBLFTRA}{\mathsf{T}_2\mathsf{BL}^{F_2,{\rm red},\ltimes}}
\newcommand{\sfTBRFTRA}{\mathsf{T}_2\mathsf{BR}^{F_2,{\rm red},\ltimes}}
\newcommand{\sfTDFTRA}{\mathsf{T}_2\mathsf{D}^{F_2,{\rm red},\ltimes}}

\newcommand{\low}{\mathrm{low}}

\begin{document}
    
    \begin{titlepage}
        \setcounter{page}{0}
        \renewcommand{\thefootnote}{\fnsymbol{footnote}}
        \vspace*{2.0cm}
        \begin{center}
            {\LARGE \bf
                Differential $T_2$-Duality and Spans \\[0.5cm] of Principal 3-Bundles with Connections
            }
            \vskip1.5cm
            {\large Gianni Gagliardo$^{a}$, Christian Saemann$^{a}$, and Roberto Tellez-Dominguez$^{b}$}\footnote{{\it E-mail addresses:\/}
                \href{mailto:ggg2001@hw.ac.uk}{\ttfamily ggg2001@hw.ac.uk},
                \href{mailto:c.saemann@hw.ac.uk}{\ttfamily c.saemann@hw.ac.uk}, 
                \href{mailto:roberto.tellez@icmat.es}{\ttfamily roberto.tellez@icmat.es}
            }
            
            \setcounter{footnote}{0}
            \renewcommand{\thefootnote}{\arabic{thefootnote}}
            \vskip1cm
            $^a$ 
            Maxwell Institute for Mathematical Sciences\\
            Department of Mathematics,
            Heriot--Watt University\\
            Edinburgh EH14 4AS, United Kingdom\\[.3cm]
            
            $^b$
            Department of Mathematics and Data Science, CEU San Pablo University\\ 
            C. Juli\'an Romea, 23, 28003 Madrid, Spain
        \end{center}
        \vskip1.0cm
        \begin{center}
            {\bf Abstract}
        \end{center}
        \begin{quote}
            We present a generalization of T-duality, called $T_2$-duality, that captures certain aspects of U-duality in M-theory. In particular, it features backgrounds that contain a 2-gerbe over a principal torus bundle, structures familiar from eleven-dimensional supergravity backgrounds. In special cases, it also incorporates S-duality in a precise sense. Our main result is a theorem characterizing \(T_2\)-dual pairs through spans of categorified principal bundles with connections, giving explicit rules for constructing a \(T_2\)-dual for a given background. A number of explicit examples illustrate how the various characteristic classes are mapped between the left and right backgrounds. For two-dimensional torus fibers, we also relate our construction to the combination of dimensional reduction and T-duality that links eleven-dimensional supergravity to type IIB supergravity. Finally, we discuss how $T_2$-duality arises from a Buscher-like span of 3d sigma models.
        \end{quote}

    \end{titlepage}

    \tableofcontents
    
    \section{Introduction and results}\label{sec:intro}
    
    \paragraph{Introduction.} T-duality~\cite{Giveon:1994fu} is a prominent feature of string theory, distinguishing it from particle physics. In toroidal compactifications of M-theory, the corresponding non-perturbative duality is U-duality, which incorporates both T-duality and S-duality, the latter being a strong-weak coupling duality. T-duality teaches us that spacetime geometry may not be fundamental, since string theories in spacetimes with very different topologies can be fully equivalent. U-duality takes this considerably further by showing that apparently different kinds of extended objects can be related by a change of duality frame. For example, the basic S-duality transformation in type IIB string theory interchanges the fundamental string and the D1-brane.
    
    Studying U-duality is certainly a very powerful approach towards learning more about M-theory. For example, the active research area of exceptional field theory, see \cite{Berman:2020tqn,Hohm:2019bba} for reviews, underlines the importance of this approach. Concretely, U-duality combines degrees of freedom arising from the metric and from higher gauge fields into unified generalized geometric objects. It allows us to see the usually listed constituents of M-theory, such as M2- and M5-branes, as well as Kaluza--Klein momenta and winding modes as different components of a unified multiplet. Furthermore, U-duality, like T-duality, predicts the existence of non-geometric backgrounds that cannot be described in terms of ordinary Riemannian geometry.
    
    There is a well-established form of topological T-duality~\cite{Bouwknegt:2003vb,Bunke:2005sn,Bunke:2005um}, where toroidal string backgrounds, i.e.~gerbes $\scG$ over torus bundles $P$, are related over a correspondence space. Concretely, two backgrounds $\scG\rightarrow P\rightarrow M$ and $\hat \scG\rightarrow \hat P\rightarrow M$ are T-dual if the pullbacks of the gerbes $\scG$ and $\hat \scG$ to the correspondence space $P\times_M\hat P$ are isomorphic, with the restriction of the isomorphism to each fiber satisfying the Poincaré condition.
    
    A higher differential geometric description of topological T-duality was given in~\cite{Nikolaus:2018qop}. It was shown that backgrounds of type $F_2$, i.e.~those appearing in geometric T-dualities, can be encoded in a principal 2-bundle with a particular structure 2-group $\sfT\sfB^{F_2}_d$. This 2-group is now replaced by a second 2-group, $\sfT\sfD^{F_2}_d$, whose classifying space is homotopy equivalent to that of $\sfT\sfB^{F_2}_d$, so that there is a natural bijection between isomorphism classes of principal $\sfT\sfD^{F_2}_d$- and $\sfT\sfB^{F_2}_d$-bundles over the same base manifold. The advantage of the 2-group $\sfT\sfD^{F_2}_d$ is that it admits a natural action\footnote{There are minor subtleties, see~\cite{Kim:2022opr,Waldorf:2022lzs}.} of the T-duality group $\sfO(d,d;\IZ)$.\footnote{There is another, surprising advantage: the 2-group $\sfT\sfD^{F_2}_d$, contrary to the 2-group $\sfT\sfB^{F_2}_d$, admits an adjustment, data required to define non-flat connections on principal $\sfT\sfD^{F_2}_d$-bundles~\cite{Kim:2022opr}.} As explained in~\cite{Nikolaus:2018qop}, there are morphisms
    \begin{equation}
        \phi:\sfT\sfD^{F_2}_d\rightarrow \sfT\sfB^{F_2}_d
        \eand
        \hat \phi:\sfT\sfD^{F_2}_d\rightarrow \sfT\sfB^{F_2}_d~,
    \end{equation}
    where $\hat\phi$ is obtained by first applying the action of the flip element in $\sfO(d,d;\IZ)$ that implements T-duality and then applying $\phi$. Given a principal $\sfT\sfD^{F_2}_d$-bundle, these two morphisms define two principal $\sfT\sfB^{F_2}_d$-bundles, and, as shown in~\cite{Nikolaus:2018qop}, the corresponding backgrounds of type $F_2$ are T-dual.  
    
    A differential refinement of this description, including the gerbe and Kaluza--Klein connections on the torus bundle, was recently provided in~\cite{Kim:2022opr}.\footnote{See also~\cite{Waldorf:2022tib} for a full combination of topological T-duality with the local rules for T-duality, known as Buscher rules.} This captures the features of differential T-duality~\cite{Kahle:2009uta}, i.e.~T-duality without metrics but Kaluza--Klein connections.
    
    A comparably complete topological picture for U-duality is still largely missing. Exceptional generalized geometry and exceptional field theory do not capture the global and finite aspects of U-duality in the way that the Bunke--Schick formalism does for T-duality. Global aspects of U-duality have been studied in four-dimensional supergravity~\cite{Lazaroiu:2021vmb}, but only for the 1-form gauge fields, rather than the higher form fields of M-theory. More closely related to our setting is spherical T-duality~\cite{Bouwknegt:2014oka}, where the second Chern class of a principal $\sfSU(2)$-bundle $P$, which is the characteristic class of the Chern--Simons 2-gerbe of $P$, is exchanged with a topological class in degree 7. While all 2-gerbes in dimension 4 are isomorphic to Chern--Simons 2-gerbes of some principal $\sfSU(2)$-bundle, this no longer holds in higher dimensions. Dualities for the specific type of 2-gerbes that are expected to model the fields in M-theory under Hypothesis H have also been studied ~\cite{Fiorenza:2018ekd,Sati:2018tvj,Giotopoulos:2024jcr}. However, these are rational constructions (i.e.~these constructions are formulated at the level of rational homotopy theory, using $L_\infty$-algebraic models), and so they do not provide a global, finite description of the kind sought here.
    
    Our aim in this paper is to formulate a more direct duality, called $T_2$-duality, between general 2-gerbes over torus bundles and other topological data; namely, 2-gerbes over 1-gerbes over torus bundles. Our expectation is that a global, finite formulation of U-duality will incorporate an underlying $T_2$-duality. Using simplicial manifolds and higher differential-geometric generalizations of~\cite{Nikolaus:2018qop,Kim:2022opr}, we describe $T_2$-duality as a span of certain principal 3-bundles with connection. For the latter, we have to rely on adjusted connections~\cite{Saemann:2019dsl,Rist:2022hci,Fischer:2024vak} and in particular~\cite{Gagliardo:2025oio}.
    
    \subsection*{Results and discussion}
    
    \paragraph{$T_2$-duality.}
    We define topological $T_2$-duality in \ref{def:t2duality}, and \ref{def:t2dualitydiff} extends this to differential $T_2$-duality, in which all involved bundles and higher bundles carry connections. Contrary to T-duality, we have two different types of backgrounds: left and right backgrounds. While a left background is a 2-gerbe with band\footnote{That is, the analogue of a higher structure group.} $\sfU(1)$ over a principal $\IT^d$-bundle, a right background is a 2-gerbe with band $\sfU(1)$ over a 1-gerbe with band $\IT^{d,\ast}$ over a principal $\IT^{\binom{d}{2},\ast}$-torus bundle. Analogously to the Bunke--Schick picture, a left and a right background are $T_2$-dual if there exists an isomorphism between the 2-gerbes over the correspondence space of the fibrations, restricting over its fibers to a natural object that we call the \emph{Poincar\'e gerbe}.
    
    Let us now show that the asymmetry in the backgrounds is rather natural, from both from a mathematical and a physical perspective. From the mathematical perspective, we note first that the dual torus $\IT^{d,\ast}$ defining the fibers of a right background in T-duality is meant to parameterize flat line bundles over the fibers $\IT^d$ of the left background, as these give automorphisms of its gerbe with connection. When the gerbe on the left background is replaced by a 2-gerbe, the natural generalization is thus to consider right backgrounds with fibers isomorphic to $\caM(\IT^d)$, the moduli space of flat 1-gerbes over $\IT^d$. The space $\caM(\IT^d)$ is naturally identified with the groupoid $\Lambda^2 \IT^{d,\ast}\times B\IT^{d,\ast}$. This leads to a ``less symmetric'' but still natural correspondence space with fibers $\IT^d \times \caM(\IT^d)$, which carries the universal Poincar\'e gerbe, as we explain in \ref{ssec:poincare}.
    
    From a physical perspective, consider the Buscher procedure for a three-dimensional sigma model with target space a topologically trivial left $T_2$-duality background. As is familiar from ordinary T-duality, this procedure relates the model to a physically equivalent three-dimensional sigma model by gauging the torus directions, introducing a Lagrange multiplier to ensure flatness of the gauge connection, and integrating out the gauge potential (which also causes the torus directions to drop out of the action). Effectively, the torus directions are traded for a Lagrange multiplier. In the case of ordinary T-duality and 2d sigma models, the Lagrange multiplier is simply a vector of 0-forms on the worldsheet and interpreted as coordinates on a dual torus $\IT^{d,\ast}$. In our case of $T_2$-duality, the Lagrange multiplier is a vector of 1-forms, which is naturally seen as coordinates on the groupoid $B\IT^{d,\ast}$ which parameterizes the fibers of a gerbe over the target space base manifold. (Recall that sigma models are maps into a manifold and gauged sigma models are maps into groupoids, such as the total space of a gerbe.) Thus, in $T_2$-duality the $\IT^d$ directions of the left background are traded for $B\IT^{d,\ast}$ directions on the right background; the additional $\Lambda^2 \IT^{d,\ast}$ directions only appear when the topology of the left background is intricate enough, as explained in \ref{ssec:F3-T2-duality}. A detailed discussion for the case $d=2$ is found in \ref{ssec:sigmamodel}.
    
    \paragraph{Toroidal 3-groups.} As a first step towards a formulation of $T_2$-duality in terms of higher principal bundles, we have to identify the appropriate higher structure groups. In the case of T-duality, the categorified torus of~\cite{Ganter:2014zoa} was found to provide a key structure 2-group~\cite{Nikolaus:2018qop}. In~\ref{ssec:categorical_tori}, we introduce the appropriate notion of categorical 3-torus, a particular form of 3-group, and we denote the categorical 3-torus featuring in our description of $T_2$-duality by $\sfTDFT_d$.
    
    Unfortunately, this 3-group is generally not a semi-strict 3-group or, equivalently, a 2-crossed module of groups. Therefore, we have to introduce the new definition of an almost 2-crossed module in \ref{sec:almostcrossed}. This is a generalization of 2-crossed modules with a conceptually transparent set of axioms. We also derive the Čech cocycle relations for principal 3-bundles with these 3-groups as higher structure groups, which then gives us a simplicial model for the delooping, which, in turn, allows us to define weak homomorphisms between 3-groups of this type.
    
    It turns out that almost 2-crossed modules, while straightforward to define, are complicated to work with. We therefore restrict ourselves to the special case of toroidal 3-groups, and all the relevant 3-groups in this paper are of this type. We derive the Čech cocycle and coboundary relations of toroidal principal 3-bundles, i.e.~principal 3-bundles with toroidal higher structure group, giving us access to delooping and weak morphisms. We also discuss how adjusted connections are defined on these. The details for one particular case are given in \ref{app:adjustment}; the other case needed below is a truncation of the  data presented there.
    
    In \ref{ssec:T2-duality_3-groups}, we then introduce the other two important toroidal 3-groups: $\sfTBLFT_d$ and $\sfTBRFT_d$. \Cref{prop:leftbackgroundsf2,prop:rightbackgroundsf2} show that principal 3-bundles with these two structure groups define topological left and right $T_2$-duality backgrounds, respectively, and, conversely, that any such background can be obtained from such principal 3-bundles.
    
    Crucially, \cref{prop:span_3_groups} arranges these toroidal 3-groups into a span of 3-groups which induces spans of corresponding toroidal principal 3-bundles:
    \begin{equation}\label{eq:span_and_bundle_spans}
        \begin{tikzcd}
            & \sfTDFT_d \arrow[dl,"\phi"'] \arrow[dr,"\hat \phi"]
            \\
            \sfTBLFT_d & & \sfTBRFT_d 
        \end{tikzcd}
        ~~~\longrightarrow~~~
        \begin{tikzcd}
            & \scP_\rmC\arrow[dl,"p"'] \arrow[dr,"\hat p"]
            \\
            \scP & & \hat \scP
        \end{tikzcd}
    \end{equation}
    
    \paragraph{$T_2$-duality as spans of toroidal 3-bundles.} \Cref{sec:t2df2} then contains the main theorems of this paper. First, we show in \ref{thm:spans_are_t2_duality} that, for a span of principal 3-bundles as in \eqref{eq:span_and_bundle_spans}, the contained principal $\sfTBLFT_d$- and $\sfTBRFT_d$-bundles $\scP$ and $\hat \scP$ describe a $T_2$-dual pair of topological $T_2$-duality backgrounds; \cref{thm:spans_are_differential_t2_duality} gives its differential refinement. Second, \ref{cor:f2t2dual} shows that any topological left $T_2$-duality background admits a topological $T_2$-dual, and \ref{cor:difff2t2dual} shows that such a background admits a choice of connection for which a differential $T_2$-dual exists. We present an illustrative example in \ref{ssec:exs2}.
    
    \paragraph{Affine $T_2$-duality.} Extending $T_2$-duality to affine $T_2$-duality, which allows for affine torus bundles in the backgrounds, involves only technical difficulties. Recall that an affine torus bundle with $d$-dimensional fibers is (an associated bundle to) a principal $\sfGL(d,\IZ)\ltimes \IT^d$-bundle. The toroidal 3-group $\sfTDFT_d$ can be extended to a 3-group $\sfTDFTA_d$ that incorporates a natural $\sfGL(d,\IZ)$ action on all of its components. We explain the details in \ref{sec:affine}, where we also establish the expected generalizations of the theorems of~\ref{sec:t2df2} and give two illustrative examples.
    
    \paragraph{Reduced $T_2$-duality.} It is useful to introduce the concept of reduced $T_2$-duality by restricting the discrete lattices appearing in the definition of the toroidal 3-groups $\sfTDFT_d$ and $\sfTDFTA_d$. These restrictions induce corresponding restrictions on the toroidal 3-groups and toroidal 3-bundles involved in an (affine) $T_2$-duality. This allows us to see other transformations of geometric backgrounds as special cases of $T_2$-duality, for which we give some examples.
    
    Of particular interest is here the reduced $T_2$-duality discussed in \ref{ssec:embedding_T-duality}, as it shows how ordinary T-duality is embedded into $T_2$-duality. In particular, this form of reduced $T_2$-duality directly embeds the results of~\cite{Nikolaus:2018qop,Kim:2022opr,Waldorf:2022tib} into our constructions.
    
    \paragraph{Physical interpretation.} Let us begin by stressing that $T_2$-duality is certainly not a fully fledged form of $U$-duality, but it shares at least some features. In particular, it contains T-duality, and it relates higher-geometric backgrounds built from 2-gerbes, 1-gerbes, and torus fibrations. We also have a natural action of $\sfSL(2,\IZ)$ on $\sfTDFT_2$, which suggests a possible extension of S-duality to topologically non-trivial backgrounds.
    
    Restricting ourselves to the topologically trivial case for $d=2$, we can link $T_2$-duality to the transition from 11d supergravity (SUGRA) fields to type IIB SUGRA fields, as described, for example, in~\cite{Bergshoeff:1995as}. We make this link precise in \ref{sec:rrfields}. In particular, we can see how the connections on left and right $T_2$-duality backgrounds enter directly into the Kaluza--Klein metric and Ramond--Ramond and Kalb--Ramond form fields of 11d and type IIB supergravity, respectively. We also see how $\sfSL(2,\IZ)$-transformations on the two-dimensional torus fibrations in the left background induce $S$-duality transformations on fields on the $T_2$-dual right background.
    
    We can also make the higher analogue of the Buscher procedure explicit, constructing a 3d sigma model which has a (topologically trivial) left $T_2$-duality background as target space. This is done for $d=2$ in \ref{ssec:sigmamodel}. Gauging the toroidal symmetry and integrating out the gauge potential leads to a dual 3d sigma model in which the target space is that of a topological right-background.
    
    \subsection*{Outlook}
    
    As explained at the end of \ref{sec:rrfields}, our construction of $T_2$-duality is definitely not fully fledged U-duality, and several important gaps remain to be addressed. Most importantly, we would have to incorporate M5-branes (presumably by including some twisted form of 5-gerbes) into the picture to achieve a democratic formulation of the M-theory field content. 
    
    Our constructions are also restricted to a generalization of the differential T-duality of~\cite{Kahle:2009uta}, i.e. we do not fully incorporate the Riemannian structure beyond its Kaluza--Klein part.  Witten's $G_4$ flux quantization condition~\cite{Witten:1996md}, for example, already shows that the curvatures of the (higher) connections on the involved $p$-gerbes and the torus bundle need to mix non-trivially (in the form of more general adjustments in the definition of higher connections) with the Riemannian structure. Such a consistent inclusion of the frame bundle may also contribute to the further development of Hypothesis H, see e.g.~\cite{Fiorenza:2020dcp,Sati:2020nob}, which aims at a precise mathematical description of the fundamental constituents of M-theory.
    
    More immediately, we would also like to expand our discussion to $T_2$-dualities of type $F_1$, possibly by a direct generalization of the corresponding T-duality treatment of \cite{Kim:2022opr}. A first, significant step in this direction has already been made in~\cite{Tellez-Dominguez:2026krp}.
    
    \section{\texorpdfstring{$T_2$}{T2}-duality}\label{sec:simpalmost}
    
    In topological $T$-duality, string backgrounds with torus fibers over a common base are linked by the Poincar\'e line bundle over their fiber product. Our definition of $T_2$-duality now replaces this bundle with a natural higher analogue, the Poincar\'e gerbe. The resulting duality requires us to consider smooth 2-gerbes over smooth 1-gerbes. To define such objects, we use a framework based on simplicial manifolds\footnote{See \ref{app:simplicial} for a concise background on simplicial manifolds.}. We first introduce higher gerbes with connections over simplicial manifolds, then construct the Poincaré gerbe, and then define $T_2$-duality backgrounds and the $T_2$-duality relation. A final subsection analyzes $T_2$-duality using the Serre spectral sequences for the relevant fibrations to identify the correct notion of backgrounds of type $F_2$ and to show that $T_2$-duals always exist for these.
    
    \subsection{Higher gerbes with connections over simplicial manifolds}\label{ssec:higher_gerbes}
    
    We start with a motivating example. Let $M$ be a manifold and let $\sigma:Y\rightarrow M$ be a surjective submersion. For $n\geq 1$, define the $n$-fold fiber product by 
    \begin{equation}
        Y^{[n]}\coloneqq Y\times_M \ldots \times_M Y=\{(y_1,\ldots,y_n)\in Y^n\,|\,\sigma(y_1)=\ldots=\sigma(y_n)\}~.
    \end{equation}
    A standard example is obtained from an open cover $\mathcal{U} = \{U_a\}_{a \in A}$ of $M$ with $\sigma$ the embedding of $Y=\sqcup_{a \in A} U_a$ into $M$. The spaces $Y^{[n]}$ then consist of the $n$-fold intersections $U_{a_1\ldots a_n} \coloneqq  \cap_{i=1}^n U_{a_i}$. The \emph{\v{C}ech resolution} of $M$ with respect to $\sigma$ is the nerve $\check{C}(\sigma)_{\bullet}$ of the Čech groupoid with
    \begin{equation}\label{eq:defcech}
        \check{C}(\sigma)_n \coloneqq  Y^{[n+1]}~.
    \end{equation}
    For all non-decreasing functions $f:[n_1] \rightarrow [n_2]$, the corresponding simplicial pullback map is given by 
    \begin{equation}
        f^{\ast}(y_{a_0},\ldots,y_{a_{n_1}}) = (y_{a_{f(0)}},\ldots,y_{a_{f(n_1)}})~.
    \end{equation}
    The surjective submersion $\sigma$ induces a canonical morphism of simplicial manifolds $\check{C}(\sigma)_{\bullet} \rightarrow M_\bullet$, where we regard $M$ as the constant simplicial manifold. This morphism yields a homotopy equivalence between the geometric realizations of $\check{C}(\sigma)_{\bullet}$ and $M_\bullet$~\cite{Segal1968}. 
    
    Now let $\sfT$ be an abelian Lie group, and consider the simplicial model $B^{p+1} \sfT$ for its $(p+1)$-fold delooping, as defined in \ref{app:simplicial}. A morphism of simplicial manifolds
    \begin{equation}
        t:\check{C}(\sigma)_{\bullet} \rightarrow B^{p+1}\sfT_{\bullet}
    \end{equation}
    is determined by its component in simplicial degree $p+1$, i.e.~a smooth map\footnote{Here and in the following, a hat over an index denotes an omission of the corresponding entry.}
    \begin{equation}
        \begin{gathered}
            t_{p+1}:\check{C}(\sigma)_{p+1} \rightarrow (B^{p+1}\sfT)_{p+1}~,
            \\
            \sum_{j=0}^{p+2} (-1)^jt_{y_0,\ldots,\hat{y}_j,\ldots,y_{p+2}}=0~.
        \end{gathered}
    \end{equation}
    That is, such a simplicial map is equivalently a smooth $\sfT$-valued \v{C}ech $(p+1)$-cocycle subordinate to $\sigma$. Associated to such a cocycle, we have the simplicial manifold
    \begin{subequations}\label{eq:totalspacelgerbe}
        \begin{equation}
            (t^{\ast}EB^{p}\sfT)_{\bullet} \coloneqq  \check{C}(\sigma)_{\bullet} \times_{(B^{p+1}\sfT)_{\bullet}} E(B^p\sfT)_{\bullet}
        \end{equation}
        with $n$-simplices
        \begin{equation}
            \begin{aligned}
                (t^{\ast}EB^{p}\sfT)_{n} &\phantom{:}= \Big\{ ((y_0,\ldots,y_n),(t_{i_0,\ldots,i_{p}})_{i_0 < i_1 < \ldots< i_{p} \in [n]}) \in Y^{[n+1]} \times \sfT^{\binom{n+1}{p+1}} 
                \\
                &\hspace{1.5cm}
                \:\Big|~\forall i_0 < \ldots < i_{p+1} \in [n], \: \sum_{j} (-1)^j t_{i_0,\ldots,\hat{i}_j,\ldots,i_{p+1}} = t_{y_{i_0}\ldots y_{i_{p+1}}} \Big\}~.
            \end{aligned}
        \end{equation}
    \end{subequations}
    In each simplicial degree $n$, this is a principal $\sfT^{\binom{n}{p}}$-bundle over $Y^{[n+1]}$. The principal action is induced from the canonical action of $(B^p\sfT)_n \cong \sfT^{\binom{n}{p}}$ on $(EB^{p}\sfT)_n$. An alternative presentation is given by the diffeomorphisms
    \begin{align}
        \begin{split}
            \left(Y^{[n+2]}\times (B^{p}\sfT_{\bullet})_n \right) / \sim &\rightarrow (t^{\ast}EB^{p}\sfT)_n ~,
            \\
            [(y_0,\ldots,y_{n+1},(t_{i_0,\ldots,i_p}))] &\mapsto (y_0,\ldots,y_n,(t_{i_0,\ldots,i_p} + t_{y_{n+1}y_{i_0}\ldots y_{i_p}}))~,
        \end{split}
    \end{align}
    where the equivalence relation is
    \begin{equation}
        \begin{aligned}
            &(y_0,\ldots,y_{n+1},(t_{i_0,\ldots,i_p})) \sim
            \\
            &\hspace{3cm}(y_0,\ldots,y_n,y'_{n+1},(t_{i_0,\ldots,i_p} + t_{y_{n+1}y_{i_0}\ldots y_{i_p}} - t_{y'_{n+1}y_{i_0}\ldots y_{i_p}}))~.
        \end{aligned}
    \end{equation}
    A coboundary $(s_{i_0\ldots i_p})$ between two simplicial maps $t^1, \,t^2:\check{C}(\sigma)_{\bullet} \rightarrow B^{p+1}\sfT_{\bullet}$, determines a map $t^{1,\ast}EB^{p}\sfT_{\bullet} \rightarrow t^{2,\ast}EB^{p}\sfT_{\bullet}$ defined at each level $n$ by the following isomorphism of principal $\sfT^{\binom{n}{p}}$-bundles: 
    \begin{equation}
        (y_0,\ldots,y_n,(t_{i_0,\ldots,i_p})) \mapsto (y_0,\ldots,y_n,(t_{i_0,\ldots,i_p} + s_{y_{i_0}\ldots y_{i_p}} ))~.
    \end{equation}
    
    The above constructions motivate a more general definition of $p$-gerbes over a simplicial manifold $M_{\bullet}$.
    \begin{definition}[$p$-gerbe]\label{def:highgerbes}
        Let $M_{\bullet}$ be a simplicial manifold, let $\sfT$ be an abelian Lie group, and let $p\in\IN$. A \emph{$p$-gerbe $P$ with band $\sfT$} over $M_{\bullet}$ consists of the following data:
        \begin{myenumerate}
            \item a simplicial manifold $Y_{\bullet}$, called a \emph{resolution}, with a simplicial morphism $Y_{\bullet} \rightarrow M_{\bullet}$, that is degree-wise a surjective submersion and induces a homotopy equivalence on geometric realizations;
            \item a simplicial manifold $P_{\bullet}$ with a morphism $P_{\bullet} \rightarrow Y_{\bullet}$ for which there exists a simplicial morphism $t:Y_{\bullet} \rightarrow B^{p+1}\sfT_{\bullet}$ and an isomorphism $P_{\bullet} \cong Y_{\bullet} \times_{B^{p+1}\sfT_{\bullet}} E(B^p\sfT)_{\bullet}$.\\[-0.5cm]
        \end{myenumerate}
    \end{definition}
    \noindent Unless stated otherwise, the band will be $\sfT=\sfU(1)$. Note that a $p$-gerbe with band $\sfT$ is a simplicial manifold whose $n$-simplices are of the form
    \begin{equation}
        P_n\cong Y_n\times \sfT^{\binom{n}{p}}~.
    \end{equation}
    
    \begin{definition}[Isomorphisms of $p$-gerbes]
        Given two $p$-gerbes $(Y^1_{\bullet},P^1_{\bullet})$ and $(Y^2_{\bullet},P^2_{\bullet})$ with band $\sfT$ over $M_{\bullet}$, an \emph{isomorphism} between them consists of the following data:
        \begin{myenumerate}
            \item a morphism of simplicial manifolds $r:Y^1_{\bullet} \rightarrow Y^2_{\bullet}$;
            \item a morphism of simplicial manifolds $\psi:P^1_{\bullet} \rightarrow P^2_{\bullet}$ covering $r$ and such that each $\psi_n:P^1_n \rightarrow r^{\ast}_nP^2_n$ is an isomorphism of principal $\sfT^{\binom{n}{p}}$-bundles over $Y^1_n$.\\[-0.5cm]
        \end{myenumerate}
    \end{definition}
    The above definition is sufficient for our purposes. More generally, one may want to generalize $r$ to a common refinement of the covers.
    
    We now turn to connections. Consider again the motivating setting from above, i.e.~a $p$-gerbe with band $\sfT$ as a map $t:\check{C}(\sigma)_{\bullet} \rightarrow B^{p+1}\sfT_{\bullet}$. Write $\mathfrak{t}$ for the Lie algebra of $\sfT$. Generalizing the standard description of connections on 1-gerbes, cf.~\cite{Murray:9407015,Chatterjee:1998,Hitchin:1999fh}, a connection is specified by differential forms 
    \begin{equation}
        \begin{aligned}
            (A_{y_0},A_{y_0y_1},\ldots,A_{y_0\ldots y_p})\in \Omega^1(Y,\frt)\times \Omega^2(Y^{[2]},\frt)\times \ldots \times \Omega^{p+1}(Y^{[p+1]},\frt)~,
        \end{aligned}
    \end{equation}
    such that
    \begin{equation}\label{eq:gerbecon1}
        \begin{aligned}
            \sum_{j} (-1)^j A_{y_0,\ldots,\hat{y}_j,\ldots,y_{p+1}} &= t_{y_0\ldots y_{p+1}}^{\ast}\theta^\sfT~, 
            \\
            \sum_{j} (-1)^j A_{y_0,\ldots,\hat{y}_j,\ldots,y_{k}} &= \rmd A_{y_0,\ldots,y_{k}}~,
        \end{aligned}
    \end{equation}
    over all $Y^{[k]}$, $k=1,\ldots p$. Moreover, we may define a connection $A_1$ on the $\sfT$-bundle $t^{\ast}(EB^p\sfT)_p \rightarrow Y^{[p+1]}$ by
    \begin{equation}
        A_1 \coloneqq  t_{y_{i_0}\ldots y_{i_p}}^{\ast}\theta^\sfT-A_{y_0,\ldots,y_{p}}~,
    \end{equation}
    as well as $A_j \in \Omega^j(t^{\ast}(EB^p\sfT)_{p+1-j},\mathfrak{t})$ for $j = 2, \,\ldots,\,p+1$ by $A_j \coloneqq  A_{y_0,\ldots,y_{j-1}}$. Equations~\eqref{eq:gerbecon1} state that the total derivative of the connection cochain $(A_1,\ldots,A_{p+1})$ in the simplicial de~Rham complex of $t^{\ast}(EB^p\sfT)_{\bullet}$ is the pullback of a closed form $F_{p+2} \coloneqq  \rmd A_{p+1} \in \Omega^{p+2}(M,\mathfrak{t})$. We further generalize this as follows.
    
    \begin{definition}[Connection on a $p$-gerbe]\label{def:connection}
        Let $M_{\bullet}$ be a simplicial manifold, let $\sfT$ be an abelian Lie group with Lie algebra $\mathfrak{t}$, and let $(Y_{\bullet},P_{\bullet})$ be a $p$-gerbe with band $\sfT$ over $M_{\bullet}$. A \emph{connection} on $(Y_{\bullet},P_{\bullet})$ consists of the following data:
        \begin{myenumerate}
            \item a collection $F=\{F_j\}_{j=2}^{p+2}$ with $F_j \in \Omega^{j}(M_{p+2-j},\mathfrak{t})$, called \emph{curvature forms}, such that 
            \begin{equation}
                \begin{aligned}
                    \delta F_2 &= 0~,
                    \\
                    \delta F_j - \rmd F_{j-1} &= 0~, \quad j= 3, \, \ldots, \, p+2~,
                    \\
                    \rmd F_{p+2} &= 0~;
                \end{aligned}
            \end{equation}
            \item a collection $A=\{A_j\}_{j=1}^{p+1}$ with $A_j \in \Omega^{j}(P_{p+1-j},\mathfrak{t})$, called \emph{connection forms}, such that $A_1$ is a connection on the $\sfT$-bundle $P_p \rightarrow Y_p$ and 
            \begin{equation}
                \begin{aligned}\label{eq:A-relations}
                    \delta A_1 &= 0~,
                    \\
                    \delta A_j + \rmd A_{j-1} &= F_j~, \quad j= 2, \, \ldots, \, p+1~,
                    \\
                    \rmd A_{p+1} &= F_{p+2}~,
                \end{aligned}
            \end{equation}
            where the pullbacks along $P_{\bullet} \rightarrow M_{\bullet}$ are omitted from the notation. 
        \end{myenumerate}
    \end{definition}
    \noindent The curvature forms are included explicitly to require that the corresponding combinations of the connection forms descend from $P_\bullet$ to $M_\bullet$.
    
    \begin{definition}[Connections on an isomorphism of $p$-gerbes]
        Let $(Y^1_{\bullet},P^1_{\bullet})$ and $(Y^2_{\bullet},P^2_{\bullet})$ be two $p$-gerbes  with band $\sfT$ over $M_{\bullet}$ with connections $(A^1,F^1)$ and $(A^2,F^2)$, and let $(r,\psi)$ be an isomorphism between the $p$-gerbes. A \emph{connection on the isomorphism} $(r,\psi)$ consists of the following data:
        \begin{myenumerate}
            \item a collection $\{G_j\}_{j=2}^{p+1}$ with $G_j \in \Omega^{j}(M_{p+1-j},\mathfrak{t})$, called \emph{curvature forms of the isomorphism}, satisfying
            \begin{equation}
                \begin{aligned}
                    \delta G_2 &= F_2^1 - F_2^2~,
                    \\
                    \delta G_j + \rmd G_{j-1}&= F_j^1 - F_j^2~, \quad j= 3, \, \ldots, \, p+1~,
                    \\
                    \rmd G_{p+1} &= F_{p+2}^1 - F_{p+2}^2~;
                \end{aligned}
            \end{equation}
            \item a collection $\{R_j\}_{j=1}^{p}$ with $R_j \in \Omega^{j}(P_{p-j},\mathfrak{t})$, called \emph{connection forms on the isomorphism}, satisfying
            \begin{equation}
                \begin{aligned}
                    \delta R_1 &= A_1^1 - \psi^{\ast}A_1^2~,
                    \\
                    \delta R_j -\rmd R_{j-1} &= A_j^1 - r^{\ast}A_j^2 - G_{j}~, \quad j= 2, \, \ldots, \, p-1~,
                    \\
                    \rmd R_p &= -A_{p+1}^1 + r^{\ast}A_{p+1}^2 - G_{p+1}~.
                \end{aligned}
            \end{equation}
        \end{myenumerate}
    \end{definition}
    \noindent We stress in particular that an isomorphism equipped with connection need not preserve the curvature forms on the $p$-gerbes strictly; the failure of strict preservation is measured by the curvature forms $G_j$.
    
    \begin{remark}
        Upon geometric realization, a $p$-gerbe with band $\sfT$ over a simplicial manifold $M_{\bullet}$ determines a map $||M_{\bullet}||\rightarrow B^{p+1}\sfT=|B^{p+1}\sfT_\bullet|$. If $\sfT = \mathbb T^d$, then pulling back the canonical cohomology class in $H^{p+2}(B^{p+1}T,\IZ^d) \cong \IZ^{d,\ast} \otimes \IZ^d$, we obtain a class in $H^{p+2}(||M_{\bullet}||,\IZ^d)$, called the \emph{Dixmier--Douady class} of the $p$-gerbe. Its image in de~Rham cohomology of $M_{\bullet}$ under the isomorphism from \ref{th:simpderham} can be represented by the curvature of any connection on the $p$-gerbe.
    \end{remark}
    
    \begin{remark}\label{rk:isolgerbes}
        The trivial $p$-gerbe (with trivial connection) over $M_{\bullet}$ is the one given by $Y_{\bullet} = M_{\bullet}$ and $t=\unit$ (with $A=0$ and $F=0$). It is not hard to see that an automorphism (with connection) of the trivial $p$-gerbe (with trivial connection) can be identified with a $(p-1)$-gerbe (with connection). In particular, for $\sfT = \mathbb T^d$ the automorphism determines a class in $H^{p+1}(||M_{\bullet}||,\IZ^d)$ whose image in de~Rham cohomology can be represented by the forms $(G_2,\ldots,G_{p+1})$ if a connection has been chosen.
    \end{remark}
    
    \subsection{The Poincar\'e gerbe}\label{ssec:poincare}
    
    Consider the $d$-torus $\mathbb T^d$. The set of principal $\sfU(1)$-bundles\footnote{or, equivalently, $0$-gerbes with band $\sfU(1)$} with connection may be identified up to bundle isomorphism (and non-canonically) with
    \begin{equation}
        H^2(\mathbb T^d,\IZ) \times \Omega^1(\mathbb T^d) / \Omega^1_\IZ(\IT^d)~,
    \end{equation}
    where $H^2(\IT^d,\IZ)$ records the first Chern class, $\Omega^1(\mathbb T^d)$ parameterizes the change from a chosen reference connection, and $\Omega^1_\IZ(\IT^d)$ are the closed 1-forms with integral periods, describing the action of the remaining bundle automorphisms\footnote{These are often called ``large gauge transformations'' in the physics literature.} on the chosen connection. The subgroup of isomorphism classes represented by flat bundles is therefore
    \begin{equation}
        \Omega^1_{\rm cl}(\mathbb T^d) / \Omega^1_\IZ(\IT^d) \cong H^1(\mathbb T^d,\mathbb R) / H^1(\mathbb T^d,\IZ) \cong \mathbb T^{d,\ast}~.
    \end{equation} 
    
    Similarly, after choosing representatives of the topological classes, a flat 1-gerbe with band $\sfU(1)$ over $\mathbb T^d$ may be represented by a closed 2-form\footnote{Recall that the cohomology of the $d$-torus $\IT^d$ is free of torsion.}. Automorphisms of such a gerbe are described by principal $\sfU(1)$-bundles with connection. Consequently, their isomorphism classes are parameterized by
    \begin{equation}\label{eq:moduli_flat_1_gerbes}
        \frac{\Omega^2_{\rm cl}(\mathbb T^d,\mathbb R)}{H^2(\mathbb T^d,\IZ) \times \Omega^1(\mathbb T^d) / \Omega^1_\IZ(\IT^d)} \cong H^2(\mathbb T^d,\mathbb R) / H^2(\mathbb T^d,\IZ) = \Lambda^2 \mathbb T^{d,\ast}~.
    \end{equation}
    Here, we used \ref{not:torus}.
    
    In higher differential geometry, it is often desirable to retain the automorphisms and hence to replace this quotient by the corresponding \emph{action groupoid}. In our setting of simplicial manifolds, we thus have the following.
    \begin{proposition}
        For $d \in \mathbb N$, the moduli groupoid of flat $1$-gerbes over $\mathbb T^d$ with band $\sfU(1)$ and their isomorphisms is the simplicial manifold $\mathcal{M}(\mathbb T^d) \coloneqq  \Lambda^{2} \mathbb T^{d,\ast} \times B \mathbb T^{d,\ast}$. 
    \end{proposition}
    \noindent Here, the additional factor of $B \mathbb T^{d,\ast}$ accounts for the isotropy in the quotient~\eqref{eq:moduli_flat_1_gerbes}; we do not retain the higher isomorphisms, as they are not needed for our purposes.
    
    Using the simplicial model~\eqref{eq:ordinary_first_delooping}, the $n$-simplices of this moduli groupoid are explicitly given by 
    \begin{align}
        \begin{split}
            \mathcal{M}(\mathbb T^d)_n &=\Big\{ (\hat{v}^0,(w^1_{i_0i_1})_{i_0<i_1 \in [n]}) \in \Lambda^{2} \mathbb T^{d,\ast} \times (\mathbb T^{d,\ast})^{\binom{n+1}{2}} 
            \\
            &\hspace{1cm}\:\big|\:~ \forall \: i_0 < i_{1} < i_2 \in [n], \: \textstyle w^1_{i_1i_2} - w^1_{i_0i_2} + w^1_{i_0i_1} = 0\Big\}~.
        \end{split}
    \end{align}
    
    Since $\mathcal{M}(\mathbb T^d)$ parameterizes flat gerbes over $\mathbb T^d$, there is a universal gerbe with connection over $\mathbb T^d \times \mathcal{M}(\mathbb T^d)$ which enters the definition of $T_2$-duality in \ref{sec:tlduality}. This is the higher analogue of the Poincar\'e bundle, where the latter exists over $\mathbb T^d \times \mathbb T^{d,\ast}$ because $\mathbb T^{d,\ast}$ parameterizes flat principal $\sfU(1)$-bundles over $\mathbb T^d$. We use the following notation for the canonical evaluation pairings:
    \begin{align}
        \langle -, -\rangle: \mathbb R^d \otimes \mathbb R^{d,\ast} \rightarrow \mathbb R~,
        \\ 
        \langle -, \langle -, -\rangle \rangle: \mathbb R^d \otimes \mathbb R^d \otimes \Lambda^2 \mathbb R^{d,\ast} \rightarrow \mathbb R~.
    \end{align} 
    
    \begin{definition}[Poincaré gerbe]\label{def:poincare}
        Let $Z_\bullet$ denote the Čech resolution associated with the universal covering map $\IR^d\rightarrow \IT^d$. Its $n$-simplices read as 
        \begin{equation}
            \begin{aligned}
                Z_n &\coloneqq  \Big\{ (\{v^0_{i_0}\}_{i_0 \in [n]}, \{\lambda_{i_0i_1} \}_{i_0 < i_1 \in [n]})\in (\mathbb R^d)^{\binom{n+1}{1}} \times (\IZ^d)^{\binom{n+1}{2}}
                \\
                &\hspace{5cm}
                \:|\:
                \forall \,i_0 < i_{1} \in [n], \: \textstyle v^0_{i_1} - v^0_{i_0} = \lambda^0_{i_0i_1}\Big\}~.
            \end{aligned}
        \end{equation}
        The \emph{Poincar\'e gerbe} is the 1-gerbe $\mathfrak{P}_\bullet$ with band $\sfU(1)$ over the simplicial manifold $M_{\bullet} = \mathbb T^d \times \mathcal{M}(\mathbb T^d)$ given by the resolution
        \begin{subequations}\label{eq:poincaredef}
            \begin{equation}
                Y_\bullet=Z_\bullet\times \caM(\IT^d)_\bullet
            \end{equation}
            and the total space\footnote{See \ref{not:low} for an explanation of the map $\langle-,\langle-,-\rangle\rangle^\low$.}
            \begin{equation}
                \begin{aligned}
                    \mathfrak{P}_n &\coloneqq  \Big\{ ((v^0_{i_0})_{i_0 \in [n]}, (\lambda_{i_0i_1} )_{i_0 < i_1 \in [n]}, \hat{v}^0,(w^1_{i_0i_1})_{i_0<i_1 \in [n]},(v^2_{i_0i_1})_{i_0<i_1 \in [n]}) 
                    \\
                    &\hspace{3cm}
                    \in (\mathbb R^d)^{\binom{n+1}{1}} \times (\IZ^d)^{\binom{n+1}{2}} \times \Lambda^{2} \mathbb T^{d,\ast} \times (\mathbb T^{d,\ast})^{\binom{n+1}{2}} \times (\mathbb T^{1})^{\binom{n+1}{2}}
                    \\
                    &\hspace{1cm}
                    \:\Big|\: \forall \,i_0 < i_{1} \in [n], \: \textstyle v^0_{i_1} - v^0_{i_0} = \lambda^0_{i_0i_1}~,
                    \\
                    &\hspace{1cm}
                    \phantom{\:\Big|\:}\forall \,i_0 < i_{1} < i_2 \in [n], \: \textstyle w^1_{i_1i_2} - w^1_{i_0i_2} + w^1_{i_0i_1} = 0~, 
                    \\
                    &\hspace{1cm}
                    \phantom{\:\Big|\:}\forall \,i_0 < i_{1} < i_2 \in [n], \: \textstyle v^2_{i_1i_2} - v^2_{i_0i_2} + v^2_{i_0i_1} = \langle \lambda_{i_{0}i_1}, w^1_{i_1i_2} \rangle + \langle \lambda_{i_0i_1}, \langle \lambda_{i_1i_2}, \hat{v}^{0} \rangle \rangle^\low\Big\}~.
                \end{aligned}
            \end{equation}
            Its \emph{canonical connection} is given by the connection forms $A_1 \in \Omega^1(\mathfrak{P}_1,\mathbb R)$, $A_2 \in \Omega^2(\mathfrak{P}_0,\mathbb R)$
            \begin{equation}
                \begin{aligned}
                    A_1 &= \rmd v^2_{01} + \langle v^0_0, \rmd w_{01}^1 \rangle + \langle v^0_0, \langle \lambda_{01}, \rmd\hat{v}^{0} \rangle \rangle^\low,\\ 
                    A_2 &= -\langle \rmd v^0_0 , \langle v^0_0, \rmd\hat{v}^0 \rangle \rangle^\low,
                \end{aligned}
            \end{equation}
            and the curvature forms $F_2 \in \Omega^2(M_1,\mathbb R)$, $F_3 \in \Omega^3(M_0,\mathbb R)$
            \begin{equation}
                \begin{aligned}\label{eq:fpoincare}
                    F_2 &= \langle \rmd v_0^0 , \rmd w_{01}^1 \rangle~, 
                    \\ 
                    F_3 &= -\langle \rmd v^0_0 , \langle \rmd v^0_0 , \rmd\hat{v}^0 \rangle \rangle^\low~. 
                \end{aligned}
            \end{equation}
        \end{subequations}
    \end{definition}
    
    Equations~\eqref{eq:fpoincare}, together with the de~Rham representatives described in \ref{ex:bltsimpman}, show that the Poincaré gerbe represents the canonical class
    \begin{equation}\label{eq:canonical_class_Poincare}
        \Pi \in H^{3}(\mathbb T^d \times \mathcal{M}(\mathbb T^d), \IZ) = H^{3}(\mathbb T^d \times \Lambda^2 \mathbb T^{d,\ast} \times B\mathbb T^{d,\ast}, \IZ)~.
    \end{equation}
    More precisely, under the Künneth decomposition, $\Pi$ is the sum of the canonical pairing components
    \begin{equation}\label{eq:canonical_pairings_Poincare}
        \begin{aligned}
            \Lambda^2 \IZ^{d,\ast} \otimes \Lambda^2 \IZ^d &= H^2(\mathbb T^d,\IZ) \otimes H^1(\Lambda^2 \mathbb T^{d,\ast},\IZ) \subset H^{3}(\mathbb T^d \times \Lambda^2 \mathbb T^{d,\ast} \times B\mathbb T^{d,\ast}, \IZ)~, 
            \\
            \IZ^{d,\ast} \otimes \IZ^d &= H^1(\mathbb T^d,\IZ) \otimes H^2( B\mathbb T^{d,\ast}, \IZ) \subset H^{3}(\mathbb T^d \times \Lambda^2 \mathbb T^{d,\ast} \times B\mathbb T^{d,\ast}, \IZ)~.
        \end{aligned}
    \end{equation}

    \subsection{\texorpdfstring{$T_2$}{T2}-duality backgrounds and duality relation}\label{sec:tlduality}
    
    We now formulate $T_2$-duality as a higher analogue of topological T-duality, cf.~\cite{Bunke:2005um}, replacing the Poincaré bundle by the Poincar\'e gerbe from \ref{def:poincare}. We start with the definition of the relevant backgrounds entering the duality.
    \begin{definition}[Topological $T_2$-duality backgrounds]\label{def:t2-backgrounds}
        A \emph{topological left $T_2$-duality background} consists of a triple $(M,P,R)$, where
        \begin{myenumerate}
            \item $P \rightarrow M$ is a principal $\mathbb T^d$-bundle and
            \item $R \rightarrow P$ is a $2$-gerbe with band $\sfU(1)$.
        \end{myenumerate} 
        
        A \emph{topological right $T_2$-duality background} consists of a quadruple $(M,\hat P,\hat Q,\hat R)$, where
        \begin{myenumerate}
            \item $\hat{P} \rightarrow M$ is a principal $\Lambda^2 \mathbb T^{d,\ast}$-bundle\footnote{see \ref{not:torus}},
            \item $\hat{Q} \rightarrow \hat{P}$ is a 1-gerbe with band $\mathbb T^{d,\ast}$, whose restriction to any fiber $\hat{P}_{|x}$, $x \in M$, is trivial, and
            \item $\hat{R} \rightarrow \hat{Q}$ is a $2$-gerbe with band $\sfU(1)$.
        \end{myenumerate} 
    \end{definition}
    \noindent We note that contrary to T-duality, there is an asymmetry between the two types of backgrounds. Mathematically, this asymmetry arises from the asymmetry in the moduli space $\caM(\IT^d)$ of flat gerbes. Physically, it arises from dualizing torus directions in 3d sigma models, which yields 1-forms that should be regarded as coordinates on gerbes, cf.~\ref{sec:physics}.
    
    The duality relation is now straightforward:
    \begin{definition}[Topological $T_2$-duality]\label{def:t2duality}
        Consider a topological left $T_2$-duality background $(M,P,R)$ and a topological right $T_2$-duality background $(M,\hat{P},\hat Q,\hat R)$. We call these backgrounds \emph{$T_2$-dual} if the following data exists:
        \begin{myenumerate}
            \item an isomorphism $\psi$ of 2-gerbes over $P \times_M \hat Q$ between the pullback of $R$ and the pullback of $\hat R$;\label{itm:t2phi}
            \item for each $x \in M$, trivializations $\chi_x$ of $R|_{P_x}$ and $\hat{\chi}_x$ of $\hat R|_{\hat Q_x}$ such that $\chi_x \circ \psi_x \circ \hat \chi_x^{-1}$, as an automorphism of the trivial 2-gerbe over \label{itm:2:t2dualityrel}
            \begin{equation}
                P_x \times \hat Q_x \cong \mathbb T^d \times \Lambda^2 \mathbb T^{d,\ast} \times B(\mathbb T^{d,\ast}) = \mathbb T^d \times \mathcal{M}(\mathbb T^d)~,
            \end{equation}
            is isomorphic to the Poincar\'e gerbe (cf.~\ref{rk:isolgerbes}).\\[-0.5cm]
        \end{myenumerate}
    \end{definition}
    
    Note that \ref{def:t2duality} does not require that $\hat Q\rightarrow M$ be a principal 2-bundle for the natural Lie 2-group structure on $\mathcal{M}(\mathbb T^d) = \Lambda^2 \mathbb T^{d,\ast} \times B \mathbb T^{d,\ast}$. In fact, it follows from~\eqref{eq:lhat} below that  $\hat Q \rightarrow M$ can only be principal when $P \rightarrow M$ is topologically trivial. 
    
    The following definition is inspired by the notion of differential $T$-duality in~\cite{Waldorf:2022tib}, which is closely related to the differential T-duality pairs from~\cite{Kahle:2009uta}.
    \begin{definition}[Differential $T_2$-duality backgrounds]
        A \emph{differential left/right $T_2$-duality background} is a topological left/right $T_2$-duality background in which the principal bundles and $p$-gerbes are endowed with connections.
    \end{definition}
    \begin{definition}[Differential $T_2$-duality]\label{def:t2dualitydiff}
        Consider a differential left $T_2$-duality background 
        $(M,P,R)$ and a differential right $T_2$-duality background $(M,\hat P,\hat Q,\hat R)$. We say that these backgrounds are \emph{$T_2$-dual} if the following holds:
        \begin{myenumerate}
            \item the backgrounds are $T_2$-dual as topological backgrounds with isomorphism $\psi$;
            \item the isomorphism $\psi$ comes with a connection whose curvature forms are\footnote{See \ref{not:low} for the notation $\langle-,\langle-,-\rangle\rangle^\low$.}
            \begin{equation}\label{eq:defdifft2dual}
                \begin{aligned}
                    G_2&=\langle A , \hat{\Lambda} \rangle~~~\in~~\Omega^2(P \times_M \hat Q_1,\mathbb R)~,
                    \\
                    G_3&=\langle A , \hat{B} \rangle - \langle A , \langle A , \hat{A} \rangle \rangle^\low~~~\in~~\Omega^{3}(P \times_M \hat Q_0,\mathbb R)~,
                \end{aligned}
            \end{equation}
            where $A$ and $\hat A$ are the connections on $P$ and $\hat P$, and $(\hat B,\hat \Lambda)$ is the connection on $\hat Q$.
        \end{myenumerate}
    \end{definition}
    
    \begin{remark}
        On every fiber of $P \times_M \hat Q \rightarrow M$, the differential forms~\eqref{eq:defdifft2dual} restrict to the curvature of the Poincar\'e gerbe given in~\eqref{eq:fpoincare}.
    \end{remark}
    
    \subsection{Characterization of topological \texorpdfstring{$T_2$}{T2}-backgrounds}\label{ssec:characterization}   
    
    Let us now further characterize the backgrounds entering into a $T_2$-duality by introducing the notion of left  and right $T_2$-duality backgrounds of type $F_2$. These will be the backgrounds for which we can show that $T_2$-duals always exist. The arguments here are direct generalizations of standard arguments with spectral sequences in topological T-duality~\cite{Bunke:2005um,Baraglia:1105.0290}, although some technical difficulties arise due to the complicated fibration structure of right $T_2$-duality backgrounds. We start by recalling some standard definitions and facts, then prove some technical lemmas, and conclude with the final result in \cref{th:existencet2duals}. For further details on the Serre spectral sequence, see  e.g.~\cite[Sec.~5.1]{Hatcher:2002bk}.
    
    Let 
    \begin{equation}
        F \xrightarrow{~~~} P \xrightarrow{~\pi_P~} M
    \end{equation}
    be a fibration over a connected manifold $M$. We denote by $\caH^t(\pi_P)$ the local system on $M$ given by $x \mapsto H^t(\pi_P^{-1}(x),\IZ)$, with the monodromy action of the fundamental group $\pi_1(M)$. In most cases of interest to us, $\caH^t(\pi_P)$ is just the constant local system $H^t(F,\IZ)$. 
    
    We now choose a CW structure on $M$ and write $M_j$ for the $j$-skeleton of $M$ and $P_j := \pi_P^{-1}(M_j)$. This yields fibrations
    \begin{equation}
        M\supset \ldots \supset M_1\supset M_0\supset M_{-1}=\varnothing
        \eand 
        P\supset \ldots \supset P_1\supset P_0\supset P_{-1}=\varnothing~.
    \end{equation}
    Applying the cohomology functor to the latter filtration, we obtain the decreasing filtration
    \begin{equation}\label{eq:cohomology_filtration}
        \begin{aligned}
            &H^q(P,\IZ) = F_0 H^q(P,\IZ) \supset F_1 H^q(P,\IZ) \supseteq \ldots 
            \\
            &\hspace{2cm}\supseteq F_q H^q(P,\IZ) = \pi_P^{\ast}H^q(M,\IZ) \supseteq F_{q+1}H^q(P,\IZ) = 0~, 
        \end{aligned}
    \end{equation}
    with 
    \begin{equation}\label{eq:skeletalfiltration}
        F_jH^q(P,\IZ) = \ker\big(H^q(P,\IZ) \rightarrow H^q(P_{j-1},\IZ)\big)~.
    \end{equation}
    
    By the standard theory of spectral sequences for filtered complexes, the successive quotients are given by
    \begin{equation}
        F_jH^q(P,\IZ)~/~F_{j+1}H^q(P,\IZ) ~\cong~ E_{\infty}^{j,q-j}~,
    \end{equation}
    where $E_\infty^{j,q-j}$ is the $\infty$-page of a first-quadrant cohomology spectral sequence $E_r^{s,t}$ converging to $H^\bullet(P,\IZ)$,
    \begin{equation}
        \begin{gathered}
            E_2^{s,t} \Longrightarrow H^{s+t}(P,\IZ)
            ~,~~~
            d_r: E_r^{s,t} \rightarrow E_r^{s+r,t-r+1}~,
            \\
            E_{r+1}^{s,t} = \frac{\ker d_r: E_{r}^{s,t}\rightarrow E_r^{s+r,t-r+1}}{\operatorname{im} d_r: E_r^{s-r,t+r-1}\rightarrow E_r^{s,t}}~.
        \end{gathered}
    \end{equation}
    
    Concretely, the $E_1$-page of this spectral sequence is given by\footnote{Here, $H^{s+t}(P_s,P_{s-1},\IZ)$ denotes the relative cohomology $H^{s+t}((P_s,P_{s-1});\IZ)$, and $C^s(M,\caH^t(\pi_P))$ are the $s$-dimensional cellular cochains with coefficients in $\caH^t(\pi_P)$.}
    \begin{equation}
        E_1^{s,t} \coloneqq H^{s+t}(P_s,P_{s-1},\IZ) = C^s(M,\caH^t(\pi_P))
    \end{equation}
    as cellular cochains and $d_1$ is the cellular differential. Recall that relative cohomology fits into a long exact sequence, which contains a connecting homomorphism
    \begin{equation}
        \delta: H^t(P_1,\IZ)\rightarrow H^{t+1}(P_2,P_1,\IZ) = E_1^{2,t-1}
    \end{equation}
    
    The $E_2$-page\footnote{which is independent of the choice of CW structure} reads as 
    \begin{equation}
        E_2^{s,t} \coloneqq H^s(M,\caH^t(\pi_P))~,
    \end{equation}
    and we need the explicit form of the differential $d_2$ acting on $u \in E_2^{0,t}$, which is defined as follows. Representing $u$ on the $E_1$-page, the condition $u \in E_2^{0,t}$ means that
    \begin{equation}
        u\in \ker(d_1:H^t(P_0, P_{-1}, \mathbb{Z})=H^t(P_0, \mathbb{Z}) \xrightarrow{~\delta~} H^{t+1}(P_1, P_0, \mathbb{Z}))~,
	\end{equation}
    and we have the left exact sequence
    \begin{equation}
        H^t(P_1,\IZ) \xrightarrow{~~~} H^{t}(P_0,\IZ) \xrightarrow{~\delta~} H^{t+1}(P_1,P_0,\IZ)~.
	\end{equation}
    Hence, there is a $\tilde{u} \in H^t(P_1,\IZ)$ whose restriction to $H^{t}(P_0,\IZ)$ is $u$. Then we have $\delta \tilde u\in H^{t+1}(P_2,P_1,\IZ)$ and $\delta \tilde u\in \ker(d_1)$, and we can set 
    \begin{equation}\label{eq:d2obstruction}
        d_2 u\coloneqq [\delta \tilde u]~.
    \end{equation}
    
    We now return to the filtration~\eqref{eq:cohomology_filtration}. For $c \in F_jH^q(P,\IZ)$ we define its \emph{$j$-th symbol} as
    \begin{equation}
        \sigma_j(c) := [c] \in F_jH^q(P,\IZ)~/~F_{j+1}H^q(P,\IZ)~\cong~E_{\infty}^{j,q-j}~.
    \end{equation}
    Note that the above-mentioned long exact sequence contains in particular maps $e^*:H^q(P_j,P_{j-1},\IZ)\rightarrow H^q(P_j,\IZ)$, and $\sigma_j(c)$ is represented in $E_1^{j,q-j}$ by the image of any lift of $c|_{P_j}$ along $e^*$. In particular, suppose $c = c(\mathcal{G})$ is the class of a $(q-2)$-gerbe $\mathcal{G} \rightarrow P$. By~\eqref{eq:skeletalfiltration}, $c(\mathcal{G})$ lies in $F_jH^{q}(P,\IZ)$ precisely if $\mathcal{G}$ is trivializable over $P_{j-1}$, and any such trivialization $\psi$ determines a representative of $\sigma_j(c(\mathcal{G}))$ in $E_1^{j,q-j}$; namely the obstruction to extending $\psi$ over $P_j$, which is indeed an element of $H^q(P_j,P_{j-1},\IZ)$ projecting to $c(\mathcal{G})|_{P_j}$. 
    
    The first lemma will allow us to relate the two conditions in the \ref{def:t2duality} of topological $T_2$-duality to properties of the Dixmier--Douady classes of the $2$-gerbes.
    \begin{lemma}\label{lem:correspondence}
        Let 
        \begin{equation}
            F \xrightarrow{~~~} P \xrightarrow{~\pi_P~} M
            \eand
            \hat{F} \xrightarrow{~~~} \hat{P} \xrightarrow{~\pi_{\hat{P}}~} M
        \end{equation}
        be fibrations over a connected manifold $M$ with connected fibers such that $H^{\bullet}(F,\IZ)$ and $H^{\bullet}(\hat{F},\IZ)$ are torsion-free. Write 
        \begin{equation}
            C := P \times_M \hat{P}~, \qquad p: C \rightarrow P~, \qquad \hat{p}: C \rightarrow \hat{P}
        \end{equation}
        for the correspondence space and its projections, and let $q \geq 2$. If $c \in H^q(P,\IZ)$ and $\hat{c} \in H^q(\hat{P},\IZ)$ satisfy 
        \begin{equation}
            p^{\ast}c = \hat{p}^{\ast}\hat{c} \in H^q(C,\IZ)~,
        \end{equation}
        then the following hold.
        \begin{myenumerate}
            \item\label{itm:corrf2} $c \in F_2H^q(P,\IZ)$ and $\hat{c} \in F_2H^q(\hat{P},\IZ)$.
            \item\label{itm:corrsymbols} Denote by 
            \begin{equation}\label{eq:mixedpart}
                \mathcal{K}^{t} := \bigoplus_{\substack{i+j = t \\ i,j \geq 1}} \caH^i(\pi_P) \otimes \caH^j(\pi_{\hat P}) \subseteq \caH^t(\pi_P\times \pi_{\hat P})
            \end{equation}
            the mixed part of the K\"unneth decomposition of the fibers of $\pi_C:C \rightarrow M$. Then there are representatives $\zeta \in E_2^{2,q-2}(\pi_P)$ of $\sigma_2(c) \in E_{\infty}^{2,q-2}(\pi_P)$ and $\hat{\zeta} \in E_2^{2,q-2}(\pi_{\hat{P}})$ of $\sigma_2(\hat{c}) \in E_{\infty}^{2,q-2}(\pi_{\hat{P}})$, and a class 
            \begin{equation}
                \theta \in H^0(M,\mathcal{K}^{q-1}) \subseteq E_2^{0,q-1}(\pi_C)~,
            \end{equation}
            such that 
            \begin{equation}\label{eq:correspondencesymbols}
                d_2\theta = p^{\ast}\zeta - \hat{p}^{\ast}\hat{\zeta} \quad \in \quad E_2^{2,q-2}(\pi_C)~.
            \end{equation}
        \end{myenumerate}
    \end{lemma}
    
    \begin{proof}
        Since the fibers are connected and their cohomology is torsion-free, we have a decomposition of local systems $\caH^t(\pi_P\times \pi_{\hat P}) = \caH^t(\pi_P) \oplus \mathcal{K}^t \oplus \caH^t(\pi_{\hat P})$. The maps induced by $p$ and $\hat{p}$ on second pages of the corresponding spectral sequences are the inclusions of the first and third summands, which for $t \geq 1$ are injective with images intersecting trivially. Note that no differential $d_r$ with $r\geq 2$ has image in bidegrees $(0,q)$ and $(1,q-1)$, hence
        \begin{equation}
            E^{j,q-j}_\infty \subset E^{j,q-j}_2
            ~~~\mbox{for}~~~
            j=0,1~.
        \end{equation}
        For $j=0,1$, the symbols $\sigma_j$ are therefore elements of the second page, and naturality of the spectral sequence gives, again for $j=0,1$,
        \begin{equation}
            p^{\ast}\sigma_j(c) = \sigma_j(p^{\ast}c) = \sigma_j(\hat{p}^{\ast}\hat{c}) = \hat{p}^{\ast}\sigma_j(\hat{c}) \quad \in \quad E_{2}^{j,q-j}(\pi_C)
        \end{equation}
        for $c \in F_jH^q(P,\IZ)$ and $\hat{c} \in F_jH^q(\hat{P},\IZ)$. But $p^{\ast}\sigma_j(c)$ lies in the summand of $E_{2}^{j,q-j}(\pi_C)$ taking values in $\caH^{q-j}(\pi_P)$, while $\hat{p}^{\ast}\sigma_j(\hat{c})$ lies in the summand taking values in $\caH^{q-j}(\pi_{\hat P})$. Therefore, both must vanish for $j=0,1$, which shows~\cref{itm:corrf2}. For $j=2$, however, $\sigma_2$ takes values in the third page. The argument then shows that for representatives $\zeta$ and $\hat{\zeta}$ of $\sigma_2(c)$ and $\sigma_2(\hat{c})$ in the second page, $p^{\ast}\zeta$ and $\hat{p}^{\ast}\hat{\zeta}$ represent the same class in the third page, and so $p^{\ast}\zeta - \hat{p}^{\ast}\hat{\zeta} = d_2\eta$ for some $\eta \in E_2^{0,q-1}(\pi_C)$. Writing $\eta = p^{\ast}\eta_1 + \theta + \hat{p}^{\ast}\eta_2$ according to the K\"unneth decomposition, naturality of $d_2$ lets us replace $\zeta$ by $\zeta - d_2\eta_1$ and $\hat{\zeta}$ by $\hat{\zeta} + d_2\eta_2$, which are again representatives of the two symbols and yield~\eqref{eq:correspondencesymbols}.
    \end{proof}
    
    The next lemma will let us derive properties of the Dixmier--Douady class of the gerbe $\hat{Q} \rightarrow \hat{P}$ in a $T_2$-dualizable right background, as well as express elements of $\caH^2(\pi_{\hat P}\circ \pi_{\hat Q})$ in a practical way.
    
    \begin{lemma}\label{lem:d2gerbe}
        Let $(M,\hat{P},\hat{Q},\hat{R})$ be a right $T_2$-duality background and $\pi_{\hat Q}: \hat{Q} \rightarrow \hat{P}$ the gerbe projection. Then $c(\hat{Q}) \in F_1H^3(\hat{P},\IZ^{d,\ast})$ for the fibration $\pi_{\hat{P}}$, the fibration $\pi_{\hat{Q}}$ has constant $\caH^1(\pi_{\hat P}\circ \pi_{\hat Q}) = \Lambda^2\IZ^d$, and there is a short exact sequence of local systems on $M$,
        \begin{equation}\label{eq:seshtwo}
            0 \longrightarrow \Lambda^2(\Lambda^2\IZ^d) \xrightarrow{~\pi_{\hat Q}^{\ast}~} \caH^2(\pi_{\hat P}\circ \pi_{\hat Q}) \longrightarrow \IZ^d \longrightarrow 0~,
        \end{equation}
        with extension class in $\operatorname{Ext}^1(\IZ^d,\Lambda^2(\Lambda^2\IZ^d)) = E_2^{1,2}(\pi_{\hat{P}};\IZ^{d,\ast})$ equal to $\sigma_1(c(\hat{Q}))$. In particular,~\eqref{eq:seshtwo} splits if and only if $c(\hat{Q}) \in F_2H^3(\hat{P},\IZ^{d,\ast})$. Moreover, in this case every trivialization $\tau$ of $\pi_{\hat Q}$ over $\hat{P}_1$ determines a splitting $s_{\tau}: \IZ^d \rightarrow \caH^2(\pi_{\hat P}\circ \pi_{\hat Q})$ of~\eqref{eq:seshtwo} and a representative
        \begin{equation}\label{eq:nutau}
            \nu_{\tau} \in E_2^{2,1}(\pi_{\hat{P}};\IZ^{d,\ast}) = H^2(M,\Lambda^2\IZ^{d} \otimes \IZ^{d,\ast})
        \end{equation}
        of $\sigma_2(c(\hat{Q})) \in E_3^{2,1}(\pi_{\hat{P}};\IZ^{d,\ast})$.
    \end{lemma}
    
    \begin{proof}
        A right background is by definition trivializable over each fiber $\hat{P}_x$, which is precisely $c(\hat Q) \in F_1H^3(\hat{P},\IZ^{d,\ast})$. Since\footnote{Recall that $BS^1\cong \IC P^\infty$ and $H^k(\IC P^\infty,\IZ)$ is $\IZ$ for $k$ even and trivial else.} $H^1(B\mathbb T^{d,\ast},\IZ) = 0$, we see that $\caH^1(\pi_{\hat P}\circ \pi_{\hat Q}) \cong \caH^1(\pi_{\hat P})$ and, because $\hat P$ is a principal $\Lambda^2 \IT^{d,\ast}$-bundle, we have $\caH^1(\pi_{\hat P}\circ \pi_{\hat Q}) \cong \Lambda^2\IZ^d$. Furthermore, we have $H^2(B\mathbb T^{d,\ast},\IZ) \cong \IZ^d$, which makes~\eqref{eq:seshtwo} immediate. Now a trivialization of $\hat Q$ over a subspace $Z \subseteq \hat{P}$ corresponds to a class in $H^2(\pi_{\hat Q}^{-1}(Z),\IZ^{d,\ast})$ restricting to $\unit \in H^2(B\mathbb T^{d,\ast},\IZ^{d,\ast}) = \operatorname{End}(\IZ^{d,\ast})$ on the fibers of $\pi_{\hat Q}$. Then choose a trivialization $\tau_0$ over $\hat{P}_0$, which exists as $c(\hat{Q}) \in F_1H^3(\hat{P},\IZ^{d,\ast})$, and let $\tilde{\iota}_0$ be the corresponding class. Restricting it to the fibers of $\pi_{\hat{Q}}$ gives a cochain in $E_1^{0,2}(\pi_{\hat{Q}};\IZ^{d,\ast}) = C^0(M,\caH^2(\pi_{\hat P}\circ \pi_{\hat Q}) \otimes \IZ^{d,\ast})$ lifting $\unit$ along~\eqref{eq:seshtwo}. Its differential $d_1\tilde{\iota}_0$ is, on the one hand, the image of $\unit$ under the connecting homomorphism of~\eqref{eq:seshtwo}, that is, the extension class; on the other hand it is $\pi_{\hat Q}^{\ast}$ of the obstruction to extending $\tau_0$ over $\hat{P}_1$, which represents $\sigma_1(c(\hat{Q}))$ by the description of the symbols recalled above. Finally, if $c(\hat{Q}) \in F_2H^3(\hat{P},\IZ^{d,\ast})$, a trivialization $\tau$ over $\hat{P}_1$ exists. Moreover, the restriction to the fibers of $\pi_{\hat{Q}}$ of its class $\tilde{\iota}_{\tau}$ is now $d_1$-closed, and hence an element of $H^0(M,\caH^2(\pi_{\hat P}\circ \pi_{\hat Q}) \otimes \IZ^{d,\ast}) = H^0(M,\sfHom(\IZ^d,\caH^2(\pi_{\hat P}\circ \pi_{\hat Q})))$ lifting $\unit$ along~\eqref{eq:seshtwo}, i.e. a splitting $s_{\tau}$. The obstruction to extending $\tau$ over $\hat{P}_2$ is $d_1$-closed and represents $\sigma_2(c(\hat{Q}))$, giving~\eqref{eq:nutau}.
    \end{proof}
    
    The final lemma we require yields a closed-form expression for the differential $d_2$ on the spectral sequence of the correspondence space of a pair of $T_2$-duality backgrounds. Recall the definition of the canonical class $\Pi$ of the Poincar\'e gerbe from~\eqref{eq:canonical_class_Poincare} as well as the canonical pairing components~\eqref{eq:canonical_pairings_Poincare}. We also note that
    \begin{equation}\label{eq:cohomologyfiber}
        H^{\bullet}(\mathbb T^d \times \Lambda^2\mathbb T^{d,\ast} \times B\mathbb T^{d,\ast},\IZ) = \Lambda^{\bullet}\big(\IZ^{d,\ast} \oplus \Lambda^2 \IZ^{d}\big) \otimes \operatorname{Sym}^{\bullet}\big(\IZ^d\big)~,
    \end{equation}
    where $\IZ^{d,\ast} = H^1(\mathbb T^d,\IZ)$ and $\Lambda^2\IZ^d = H^1(\Lambda^2\mathbb T^{d,\ast},\IZ)$ sit in degree $1$ and $\IZ^d = H^2(B\mathbb T^{d,\ast},\IZ)$ sits in degree $2$    
    
    \begin{lemma}\label{lem:d2correspondence}
        Let $(M,P,R)$ and $(M,\hat{P},\hat{Q},\hat{R})$ be topological left and right $T_2$-duality backgrounds with $c(\hat{Q}) \in F_2H^3(\hat{P},\IZ^{d,\ast})$, and write $C := P \times_M \hat{Q}$ and $\pi_C=C\rightarrow M$. Choose a trivialization $\tau$ of $\pi_{\hat Q}$ over $\hat{P}_1$, use it to write $\caH^2(\pi_{\hat P}\circ \pi_{\hat Q}) = \Lambda^2(\Lambda^2 \IZ^d) \oplus \IZ^d$ under a splitting $s_{\tau}$, and let $\nu := \nu_{\tau}$ be the associated representative of $\sigma_2(c(\hat{Q}))$ as in \cref{lem:d2gerbe}. Then the differential $d_2$ of the fibration $\pi_C$ is the unique $H^{\bullet}(M,\IZ)$-linear antiderivation with
        \begin{equation}\label{eq:d2generators}
            \begin{aligned}
                d_2 \alpha &= \langle \alpha, c(P)\rangle \in H^2(M,\IZ)~, &&\alpha \in \IZ^{d,\ast}~,\\
                d_2 \beta &= \langle \beta, c(\hat{P})\rangle \in H^2(M,\IZ)~, &&\beta \in \Lambda^2 \IZ^{d}~,\\
                d_2 s_{\tau}(n) &= \langle n, \nu \rangle \in H^2(M,\Lambda^2\IZ^{d})~, &&n \in \IZ^{d}~,
            \end{aligned}
        \end{equation}
        In particular, decomposing
        \begin{equation}\label{eq:h2fiberdecomposition}
            \caH^2(\pi_C) = \underbrace{\Lambda^2\IZ^{d,\ast}}_{= \,\caH^2(\pi_P)} \oplus \underbrace{\big(\IZ^{d,\ast} \otimes \Lambda^2\IZ^d\big)}_{= \,\mathcal{K}^2} \oplus \underbrace{\Lambda^2(\Lambda^2 \IZ^d) \oplus \IZ^d}_{= \,\caH^2(\pi_{\hat P}\circ \pi_{\hat Q})}~,
        \end{equation}
        the class $\Pi$ of the Poincar\'e gerbe satisfies
        \begin{equation}\label{eq:d2poincare}
            d_2 \Pi = \Big(c(\hat{P}),\; \delta c(P) - \nu,\; 0,\; c(P)\Big) \quad \in \quad E_2^{2,2}(\pi_C)~,
        \end{equation}
        where $\delta: \IZ^d \rightarrow \operatorname{Hom}(\IZ^d,\Lambda^2\IZ^d) = \Lambda^2 \IZ^d \otimes \IZ^{d,\ast}$ is the canonical map $\delta(k)(l) := k \wedge l$.
    \end{lemma}
    
    \begin{proof}
        By~\eqref{eq:cohomologyfiber}, the ring $\caH^{\bullet}(\pi_C)$ is generated by $\IZ^{d,\ast} \oplus \Lambda^2\IZ^d$ in degree $1$ together with the image of $s_{\tau}$ in degree $2$; hence~\eqref{eq:d2generators} determines it. Denote by $p: C \rightarrow P$ and $\hat{p}: C \rightarrow \hat{Q}$ the projections. The first two formulas follow by naturality from the corresponding statement for a principal torus bundle, for which $d_2$ is contraction with the Chern class~\cite{1969728}. For the third one, recall that $d_2s_{\tau}$ is the obstruction to extending over $\hat{Q}_2$ the class $\tilde{\iota}_{\tau} \in H^2(\hat{Q}_1,\IZ^{d,\ast})$ lifting $s_{\tau}$ as in~\eqref{eq:d2obstruction}. But this is the same as extending $\tau$ over $\hat{P}_2$, whose obstruction is $\nu$, because $\caH^1(\pi_{\hat P}\circ \pi_{\hat Q}) = \caH^1(\pi_P)$ yields $E^{2,1}_1(\pi_{\hat{P}};\IZ^{d,\ast}) = E^{2,1}_1(\pi_{\hat{Q}};\IZ^{d,\ast})$. Thus $d_2s_{\tau} = \nu$, as desired. Formula~\eqref{eq:d2poincare} then follows directly from these relations, writing the canonical pairings in a basis.
    \end{proof}
    
    
    With the above lemmata, we are now ready to prove the main result of this section.
    
    \begin{theorem}\label{th:existencet2duals}
        Let $(M,P,R)$ and $(M,\hat{P},\hat{Q},\hat{R})$ be topological left and right $T_2$-duality backgrounds, respectively, over the same manifold $M$. If they are $T_2$-dual, then their characteristic classes
        \begin{equation}
            \begin{gathered}
                c(P) \in H^2(M,\IZ^d)~, \qquad c(R) \in H^4(P,\IZ)~, \\
                c(\hat{P}) \in H^2(M,\Lambda^2\IZ^{d,\ast})~, \qquad c(\hat{Q}) \in H^3(\hat{P},\IZ^{d,\ast})~, \qquad c(\hat{R}) \in H^4(\hat{Q},\IZ)~,
            \end{gathered}
        \end{equation}
        satisfy the following.
        \begin{myenumerate}
            \item\label{itm:thleft} $c(R) \in F_2H^4(P,\IZ)$ for the fibration $\pi_P$;
            \item\label{itm:thright} $c(\hat{R}) \in F_2H^4(\hat{Q},\IZ)$ for the fibration $\pi_{\hat{Q}}$;
            \item\label{itm:thqhat} $c(\hat{Q}) \in F_2H^3(\hat{P},\IZ^{d,\ast})$ for the fibration $\pi_{\hat{P}}$;
            \item\label{itm:thmomentum} the identities
            \begin{align}
                \begin{split}\label{eq:momentumexchange}
                    \sigma_2(c(R)) &= \big[c(\hat{P})\big] \in E_3^{2,2}(\pi_P)~, \\ 
                    \sigma_2(c(\hat{Q})) &= \big[\delta c(P)\big] \in E_3^{2,1}(\pi_{\hat{P}};\IZ^{d,\ast})~, \\ 
                    \sigma_2(c(\hat{R})) &= -\big[c(P)\big] \in E_3^{2,2}(\pi_{\hat{Q}})~,
                \end{split}
            \end{align}
            where $\delta$ is as in \cref{lem:d2correspondence} and $[\,\cdot\,]$ denotes the class of a $d_2$-cocycle in $E_3 = \ker d_2/\operatorname{im}d_2$, the relevant second pages being
            \begin{align}
                \begin{split}\label{eq:momentumpages}
                    E_2^{2,2}(\pi_P) &= H^2(M,\Lambda^2\IZ^{d,\ast})~, \\
                    E_2^{2,1}(\pi_{\hat{P}};\IZ^{d,\ast}) &= H^2(M,\Lambda^2\IZ^{d} \otimes \IZ^{d,\ast})~, \\
                    E_2^{2,2}(\pi_{\hat{Q}}) &= H^2\big(M,\Lambda^2(\Lambda^2\IZ^d) \oplus \IZ^d\big)~,
                \end{split}
            \end{align}
            and $c(P)$ being regarded in the summand $H^2(M,\IZ^d)$ of the last one.
        \end{myenumerate}
    \end{theorem}
    \begin{proof}
        Write $C := P \times_M \hat{Q}$ with projections $p$ and $\hat{p}$. The isomorphism $\psi$ of~\cref{itm:t2phi} in \cref{itm:2:t2dualityrel} of \ref{def:t2duality} gives $p^{\ast}c(R) = \hat{p}^{\ast}c(\hat{R})$, so \cref{itm:corrf2} of \cref{lem:correspondence} proves~\cref{itm:thleft,itm:thright}. In particular $c(R)$ and $c(\hat{R})$ restrict trivially to $P_1$ and $\hat{Q}_1$, so we may choose trivializations $\Psi$ of $R|_{P_1}$ and $\hat{\Psi}$ of $\hat{R}|_{\hat{Q}_1}$. The obstructions to extending them over $P_2$ and $\hat{Q}_2$ are representatives $\zeta$ and $\hat{\zeta}$ of $\sigma_2(c(R))$ and $\sigma_2(c(\hat{R}))$, and, as in the proof of \cref{itm:corrsymbols} of \cref{lem:correspondence},
        \begin{equation}\label{eq:thetaphi}
            \theta := \big[\Psi \circ \psi \circ \hat{\Psi}^{-1}\big] \in H^3(C_1,\IZ) \qquad \text{satisfies} \qquad d_2 \theta = p^{\ast}\zeta - \hat{p}^{\ast}\hat{\zeta}~,
        \end{equation}
        where we also write $\theta$ for the induced section of $\mathcal{K}^3 = (\Lambda^2\IZ^{d,\ast} \otimes \Lambda^2 \IZ^d) \oplus (\IZ^{d,\ast} \otimes \caH^2(\pi_{\hat P}\circ \pi_{\hat Q}))$.
        
        By~\cref{itm:2:t2dualityrel} of \cref{def:t2duality}, the image of $\theta$ under the projection induced by~\eqref{eq:seshtwo} is the class $\Pi$ of~\eqref{eq:canonical_pairings_Poincare}. Its second component is thus a section of $\sfHom(\IZ^d,\caH^2(\pi_{\hat P}\circ \pi_{\hat Q}))$ splitting~\eqref{eq:seshtwo}, so \cref{lem:d2gerbe} gives~\cref{itm:thqhat} and identifies $\theta$ with $\Pi$. Substituting~\eqref{eq:d2poincare} into~\eqref{eq:thetaphi} and comparing the four summands of~\eqref{eq:h2fiberdecomposition} gives~\cref{itm:thmomentum}.
    \end{proof}
    
    Following~\cite{Bunke:2005um}, we single out the backgrounds satisfying the conditions of \cref{th:existencet2duals}.
    \begin{definition}[Backgrounds of type $F_2$]\label{def:typef2}
        A topological left $T_2$-duality background $(M,P,R)$ is \emph{of type $F_2$} if $c(R) \in F_2H^4(P,\IZ)$. A topological right $T_2$-duality background $(M,\hat{P},\hat{Q},\hat{R})$ is \emph{of type $F_2$} if
        \begin{equation}
            c(\hat{R}) \in F_2H^4(\hat{Q},\IZ)~,~~~c(\hat{Q}) \in F_2H^3(\hat{P},\IZ^{d,\ast})~,
		\end{equation}
        and $\sigma_2(c(\hat{R}))$ admits a representative $\hat{y}$ in the summand $H^2(M,\IZ^d)$ of~\eqref{eq:momentumpages} such that
        \begin{equation}\label{eq:typef2right}
            \sigma_2(c(\hat{Q})) = -\big[\delta \hat{y}\big] \in E_3^{2,1}(\pi_{\hat{P}};\IZ^{d,\ast})~,
        \end{equation}
        with the notation of~\eqref{eq:momentumexchange}.
    \end{definition}
    
    \Cref{th:existencet2duals} implies that $T_2$-duality backgrounds admitting $T_2$-duals must be of type $F_2$. The next two sections use Lie 3-groups to show that this is also a sufficient condition, that backgrounds of type $F_2$ can always be equipped with connections admitting a differential $T_2$-dual, and that both topological and differential $T_2$-duals can be constructed by choosing some specific cocycle data representing the background.

    \section{Almost 2-crossed modules of Lie groups and toroidal 3-groups}\label{sec:almostcrossed}
    
    We now introduce the general class of structure 3-groups that will allow us to describe $T_2$-duality as spans of principal 3-bundles. Lie 2-crossed modules~\cite{Conduche:1984:155,Conduche:2003} are commonly used as algebraic models for 3-groups. These provide a convenient description of Lie Gray $3$-groups~\cite{Kamps:2002aa}. For our purposes, however, we need to be slightly more general, and we introduce a new notion of almost 2-crossed modules of groups. 
    
    \subsection{Almost 2-crossed modules}
    
    Almost 2-crossed modules of groups are 2-crossed modules of groups in which one of the strict group actions has been lifted to a weak one.
    \begin{definition}[Almost 2-crossed module of groups]\label{def:almostcrossed}
        An almost 2-crossed module of groups is a normal complex of groups
        \begin{subequations}
            \begin{equation}
                \sfL \xrightarrow{~\sft~}\sfH \xrightarrow{~\sft~}\sfG~,~~~\sft\circ \sft=\unit~,
            \end{equation}
            together with the following structure maps and coherence conditions:
            \begin{myenumerate}
                \item a (strict, unital) action of $\sfG$ on $\sfL$ by automorphisms, denoted by $\acton$;
                \item a unital weak action $\triangleright : \sfG \times \sfH \rightarrow \sfH$ for which $\sft$ is equivariant in the sense that
                \begin{equation}
                    \begin{aligned}
                        g \triangleright \unit &= \unit
                        ~,~~~&
                        \unit \triangleright h &= h~,
                        \\
                        \sft(g \triangleright \ell) &= g \triangleright \sft(\ell)
                        ~,~~~&
                        \sft(g \triangleright h) &= g\sft(h)g^{-1}~,
                    \end{aligned}
                \end{equation}
                where the failure of the action to be strictly associative and multiplicative is measured by maps 
                \begin{equation}
                    \alpha: \sfG \times \sfG \times \sfH \rightarrow \sfL
                    \eand
                    \beta: \sfG \times \sfH \times \sfH \rightarrow \sfL
                \end{equation}
                as follows:
                \begin{equation}\label{eq:weak_action}
                    \begin{aligned}
                        (g_1g_2) \triangleright h &= g_1 \triangleright (g_2 \triangleright h)\, \sft(\alpha(g_1,g_2,h))~,
                        \\
                        g \triangleright (h_1h_2) &= (g \triangleright h_1) (g \triangleright h_2)\, \sft(\beta(g,h_1,h_2))~,
                    \end{aligned}
                \end{equation}
                with the maps satisfying the unitality conditions
                \begin{equation}\label{eq:unitality}
                    \begin{gathered}
                        \alpha(\unit,g,h) = \alpha(g,\unit,h) = \alpha(g_1,g_2,\unit)= \alpha(g_1,g_2,\sft(\ell))= \unit ~,
                        \\
                        \beta(\unit,h_1,h_2) = \beta(g,\unit,h) = \beta(g,h,\unit)=\beta(g,\sft(\ell_1),\sft(\ell_2))=\unit~,
                    \end{gathered}
                \end{equation}
                and the coherence conditions
                \begin{equation}\label{eq:associativity}
                    \begin{aligned}
                        \alpha(g_1g_2,g_3,h)&=\alpha(g_1,g_2,g_3 \triangleright h)^{-1}(g_1 \triangleright \alpha(g_2,g_3,h))
                        \\
                        &\hspace{0.5cm}\times\beta(g_1,g_2 \triangleright (g_3 \triangleright h),\sft(\alpha(g_2,g_3,h)))\, \alpha(g_1,g_2g_3,h)
                        ~,
                        \\
                        \beta(g,h_1,h_2h_3) &= \beta(g,h_2,h_3)^{-1}((g \triangleright h_3)^{-1}\triangleright  \beta(g,h_1,h_2))\, \beta(g,h_1h_2,h_3)~,
                        \\
                        \beta(g_1g_2,h_1,h_2) 
                        &= \alpha(g_1,g_2,h_2)^{-1}((g_1 \triangleright (g_2 \triangleright h_2)^{-1}) \triangleright \alpha(g_1,g_2,h_1))^{-1} 
                        \\
                        &\hspace{0.5cm}\times
                        \beta(g_1,g_2\acton h_1,g_2\acton h_2)(g_1 \triangleright \beta(g_2,h_1,h_2))
                        \\
                        &\hspace{0.5cm}\times \beta(g_1,(g_2\acton h_1)(g_2\acton h_2),\sft(\beta(g_2,h_1,h_2)))\,\alpha(g_1,g_2,h_1h_2)~,
                    \end{aligned}
                \end{equation}
                where we use the shorthand notation
                \begin{equation}\label{eq:defhactl}
                    h \triangleright \ell \coloneqq \ell\{\sft(\ell^{-1}),h\}~;
                \end{equation} 
                \item a map $\{-,-\}: \sfH \times \sfH \rightarrow \sfL$, called the \emph{Pfeiffer lifting}, satisfying
                \begin{equation}\label{eq:peiff_axioms_def}
                    \begin{gathered}
                        \sft\{h_1,h_2\} = h_1h_2h^{-1}_1(\sft(h_1) \triangleright h_2^{-1})~,
                        \\
                        \{\sft(\ell_1),\sft(\ell_2)\} = \ell_1\ell_2\ell^{-1}_1\ell^{-1}_2~,
                        \\
                        \{\sft(\ell),h\}\{h,\sft(\ell)\} = \ell (\sft(h) \triangleright \ell^{-1})~, 
                        \\
                    \end{gathered}
                \end{equation}
                and
                \begin{equation}\label{eq:peiff_ax}
                    \begin{aligned}
                        \{h_1h_2,h_3\} &= 
                        \{h_1,h_2h_3h_2^{-1}\}(\sft(h_1)\acton\{h_2,h_3\})\,\beta(\sft(h_1),h_2h_3^{-1}h_2^{-1},\sft(\{h_2,h_3\}))
                        \\
                        &\hspace{0.5cm}\times\alpha(\sft(h_1),\sft(h_2),h_3^{-1})~,
                        \\
                        \{h_1,h_2h_3\} &= \{h_1,h_2\}\beta(\sft(h_1),h_2,h_2^{-1})\{h_1,h_3\}\{\sft(\{h_1,h_3\}^{-1}),\sft(h_1)\acton h_2\}
                        \\
                        &\hspace{0.5cm}\times\beta(\sft(h_1),h_2,h_2^{-1})^{-1}\beta(\sft(h_1),h_3^{-1},h_2^{-1})~,
                        \\
                        g \triangleright \{h_1,h_2\} &= 
                        R_1
                        \{g \triangleright h_1, g \triangleright h_2\} 
                        R_2~,
                    \end{aligned}
                \end{equation}
                where
                \begin{equation}
                    \begin{aligned}
                        R_1&=(g\acton (h_1h_2))\acton (\beta(g,h_1^{-1},h_1)^{-1}\beta(g,h_1,h_2))~,
                        \\
                        R_2&=\beta(\sft(g\acton h_1),g\acton(h_2^{-1}),(g\acton(h_2^{-1}))^{-1}(g\acton h_2)^{-1})^{-1}
                        \\&\hspace{1cm}
                        \times (\sft(g\acton h_1)\acton((g\acton h_2^{-1})^{-1}\acton \beta(g,h_2^{-1},h_2)))^{-1}
                        \\&\hspace{1cm}
                        \times \alpha(\sft(g\acton h_1),g,h_2^{-1})~\alpha(g,\sft(h_1),h_2^{-1})^{-1}
                        \\&\hspace{1cm}
                        \times \beta(g,h_1^{-1},\sft(h_1)\acton h_2^{-1})~\beta(g,h_1h_2,h_1^{-1}(\sft(h_1)\acton h_2^{-1}))~.
                    \end{aligned}
                \end{equation}
            \end{myenumerate}
        \end{subequations}
    \end{definition}
    \noindent A 2-crossed module is the same as an almost 2-crossed module for which $\alpha$ and $\beta$ are trivial. 
    
    The above axioms arise from those of a 2-crossed module by weakening the action of $\sfG$ on $\sfH$: relations~\eqref{eq:unitality} are implied by unitality of the actions,~\eqref{eq:associativity} are implied by the associativity of the group products and the two ways in which $(g_1g_2)\acton (h_1h_2)$ can be resolved, and~\eqref{eq:peiff_axioms_def} follow from compatibility with the application of $\sft$ on both sides. Full coherence of the axioms can then be verified by direct symbolic computation.
    
    As for ordinary 2-crossed modules, the axioms imply that the complex
    \begin{subequations}
        \begin{equation}
            \sfL\xrightarrow{~\sft~}\sfH
        \end{equation} 
        is a crossed module of groups with action~\eqref{eq:defhactl}. Indeed, the two relevant identities are satisfied:
        \begin{equation}
            \sft(h \triangleright \ell_1) = h \sft(\ell_1) h^{-1}
            \eand
            \sft(\ell_1) \triangleright \ell_2 = \ell_1\ell_2\ell_1^{-1}~.
        \end{equation}
    \end{subequations}
    
    The $\sfH$-action on $\sfL$ fails to be $\sfG$-equivariant in the following sense:
    \begin{equation}
        g \triangleright (h \triangleright \ell) = \big((g \triangleright h) \triangleright (g \triangleright \ell) \big)\beta(g,\sft(\ell),h^{-1})\beta(g,h,\sft(\ell)h^{-1})\beta(g,\sft(\ell^{-1}),\sft(h\acton \ell))~.
    \end{equation}
    
    We are exclusively interested in almost 2-crossed modules in which $\sfG$, $\sfH$, and $\sfL$ are Lie groups. Applying the tangent functor to these, we obtain a notion of almost 2-crossed modules of Lie algebras.
    
    \subsection{Principal 3-bundles and homomorphisms of almost 2-crossed modules}
    
    Let $\caG=(\sfL\rightarrow \sfH\rightarrow \sfG)$ be an almost 2-crossed module of Lie groups. Principal $\caG$-bundles, i.e.~principal 3-bundles with $\caG$ as structure 3-group and their isomorphisms, admit several equivalent descriptions. Here, we use Čech cocycles and (higher) Čech coboundaries. The cocycles will always be subordinate to a surjective submersion $\sigma:Y\rightarrow M$. As before, we write $Y^{[n]}=Y\times_M\ldots\times_M Y$ for the $n$-fold fiber product.
    
    A Čech cocycle for a principal $\caG$-bundle subordinate to $\sigma$ consists of smooth maps
    \begin{equation}
        (g_{ij})\in C^\infty(Y^{[2]},\sfG)~,~~~(h_{ijk})\in C^\infty(Y^{[3]},\sfH)~,~~~(\ell_{ijkl})\in C^\infty(Y^{[4]},\sfL)
    \end{equation}
    such that
    \begin{subequations}\label{eq:cocycle3pbs}
        \begin{equation}\label{eq:cocycle-gh-relation}
            \begin{aligned}
                \sft(h_{ijk})g_{ij}g_{jk} &= g_{ik}~,
                \\
                h_{ikl}h_{ijk}\sft(\ell_{ijkl}) &= h_{ijl} \, g_{ij} \triangleright h_{jkl}~,
                \\
            \end{aligned}
        \end{equation}
        and
        \begin{equation}\label{eq:cocycle-l-relation}
            \begin{aligned}
                &\ell_{ijkl}\,((g_{ij} \triangleright h_{jkl})^{-1} \triangleright \ell_{ijlm} ) \,\beta(g_{ij},h_{jlm},h_{jkl}) \,g_{ij} \triangleright \ell_{jklm}
                \\& \hspace{1cm}
                \times \: \beta(g_{ij},h_{jlm}h_{jkl},\sft(\ell_{jklm})) \, \beta^{-1}(g_{ij},h_{jkm},g_{jk} \triangleright h_{klm})
                \\& \hspace{1cm}
                = (h_{ijk}^{-1} \triangleright \ell_{iklm}) \, \{h_{ijk}^{-1}, (g_{ik} \triangleright h_{klm})^{-1}\} \,  \alpha(\sft(h_{ijk}^{-1}),g_{ik},h_{klm})
                \\& \hspace{2cm}
                \times \: \alpha^{-1}(g_{ij},g_{jk},h_{klm}) \,\big((g_{ij} \triangleright (g_{jk} \triangleright h_{klm}))^{-1} \triangleright \ell_{ijkm}\big)
            \end{aligned}
        \end{equation}
        for all $ij\in Y^{[2]}$, $ijk\in Y^{[3]}$, etc. The first relation in~\eqref{eq:cocycle-gh-relation} is a lift of the cocycle condition of a principal $\sfG$-bundle, and the second relation describes a lift of the cocycle condition in a principal 2-bundle. These two relations agree with those for a 2-crossed module of Lie groups~\cite{Jurco:2009px,Saemann:2013pca}.  Relation~\eqref{eq:cocycle-l-relation} is then obtained by deforming the cocycle relations for a 2-crossed module of Lie groups such that applying $\sft$ to both sides yields an identity. One then verifies coherence of the relations over $Y^{[6]}$.
    \end{subequations}
    
    These Čech cocycles allow us to give a simplicial model for the delooping $B\caG$ of $\caG$, which is obtained by applying the coskeleton functor~\cite{Zhu:0801.2057} to the complex of groups underlying $\caG$, leading to the $2$-tower of $\caG$, cf.~\cite{Li:2309.00901}. 
    
    Concretely, we have
    \begin{equation}
        \begin{aligned}
            (B\mathcal{G})_n \coloneqq  \Big\{ &\big((g_{ij})_{i<j \in [n]},(h_{ijk})_{i<j<k \in [n]},(\ell_{ijkl})_{i<j<k<l \in [n]}\big) \in \sfG^{\binom{n+1}{2}} \times \sfH^{\binom{n+1}{3}} \times \sfL^{\binom{n+1}{4}} 
            \\
            &\Big| \: \mbox{all $(g_{ij})$, $(h_{ijk})$, and $(\ell_{ijkl})$ satisfy~\eqref{eq:cocycle3pbs}} \Big\}~.
        \end{aligned}
    \end{equation} 
    For each non-decreasing function $f:[n_1] \rightarrow [n_2]$, we define $f^{\ast}: (B\mathcal{G})_{n_2} \rightarrow (B\mathcal{G})_{n_1}$ by
    \begin{align}
        \begin{split}\label{eq:simpmapspullback2}
            ((g_{ij}),(h_{ijk}), (\ell_{ijkl})) &\mapsto ((f^{\ast}g_{ij}),(f^{\ast}h_{ijk}), (f^{\ast}\ell_{ijkl}))~,
            \\
            f^{\ast}g_{ij} \coloneqq  g_{f(i)f(j)}
            ~,~~~ 
            f^{\ast}h_{ijk} &\coloneqq  h_{f(i)f(j)f(k)}
            ~,~~~
            f^{\ast}\ell_{ijkl} \coloneqq  \ell_{f(i)f(j)f(k)f(l)}
        \end{split}
    \end{align}
    with the understanding that $g_{ij}$, $h_{ijk}$ and $\ell_{ijkl}$ equal $\unit$ whenever there are repeated indices. The face and degeneracy maps on $B\caG_\bullet$ are then defined as $d^n_j = (\delta^n_j)^{\ast}$ and $s^n_j = (\sigma^n_j)^{\ast}$, as usual, cf.~\ref{app:simplicial}.
    
    \begin{proposition}
        The geometric realization $|B\mathcal{G}_\bullet|$ of this simplicial manifold is the \emph{classifying space for principal $\mathcal{G}$-bundles}. 
    \end{proposition}
    \begin{proof}
        A Čech cocycle as in~\eqref{eq:cocycle3pbs} for a principal $\caG$-bundle is a simplicial map from the nerve of the Čech groupoid of a surjective submersion $Y_\bullet\rightarrow M$ to $B\caG_\bullet$. Such a map, in turn, determines a continuous map $M \rightarrow |B\mathcal{G}|$. Moreover, any such continuous map is homotopic to a map obtained in this way. Such cocycle data determine a smooth principal $\caG$-bundle over $M$ subordinate to $Y\rightarrow M$ in terms of Čech cocycles, generalizing the principal 3-bundles of~\cite{Jurco:2009px,Saemann:2013pca}.
    \end{proof}
    
    Having the classifying spaces at our disposal allows us, in turn, to define (weak) homomorphisms of almost 2-crossed modules.
    \begin{definition}[Homomorphism, isomorphism, homotopy equivalence]\label{def:weak_morphisms}
        Consider almost 2-crossed modules of groups $\caG^1$ and $\caG^2$. A \emph{(weak) homomorphism} $\psi:\caG^1\rightarrow \caG^2$ is a morphism of simplicial manifolds between $B\caG^1_\bullet$ and $B\caG^2_\bullet$. We sometimes refer to the latter morphism as $B\psi$. 
        
        The \emph{kernel} of a weak homomorphism $\psi$ is the simplicial manifold $\psi^{-1}(\unit)$ whose $n$-simplices are given by the sub-manifolds $\psi_n^{-1}(\unit_n) \subset B\mathcal{G}_n$, where the simplicial unit $\unit_\bullet$ has $n$-simplices
        \begin{equation}
            \unit_n=((\unit,\ldots),(\unit,\ldots),(\unit,\ldots)) \in B\mathcal{G}_n~.
        \end{equation}
        
        An \emph{isomorphism} is a surjective homomorphism whose kernel is trivial, i.e.~the simplicial unit $\unit_\bullet$.
        
        A \emph{(weak) homotopy equivalence} is a homomorphism whose kernel has contractible geometric realization. 
    \end{definition}
    \noindent In particular, the geometric realization of a weak homotopy equivalence, $|\psi|:|B\mathcal{G}_\bullet| \rightarrow |B\mathcal{G}'_\bullet|$, is a weak homotopy equivalence of topological spaces.
    
    We also note that a strict homomorphism between two almost 2-crossed modules is simply a triple of group homomorphisms between the underlying groups that preserve all structure maps. Such a map trivially lifts to a map $B\caG^1_\bullet\rightarrow B\caG^2_\bullet$.
    
    From these definitions it is clear that a homomorphism $\psi:\mathcal{G} \rightarrow \mathcal{G}'$ sends the Čech cocycles of a principal $\mathcal{G}$-bundle over a manifold $M$ to a Čech cocycle of a principal $\mathcal{G}'$-bundle over $M$ by post-composition with $B\psi$. Moreover, every principal $\mathcal{G}'$-bundle arises from a principal $\mathcal{G}$-bundle in this way, up to isomorphism, if $\psi$ is a weak homotopy equivalence.
    
    \subsection{Toroidal 3-groups}
    
    Although the preceding definition is conceptually direct, general almost 2-crossed modules are cumbersome to use, particularly due to the involved axioms of the Peiffer lifting $\{-,-\}$. We therefore restrict ourselves to the following special class tailored to our applications in the context of $T_2$-duality.
    \begin{definition}[Toroidal 3-group]\label{def:toroidal_3-group}
        A \emph{toroidal 3-group} is an almost 2-crossed module of Lie groups $(\sfL\xrightarrow{\sft}\sfH\xrightarrow{\sft}\sfG)$ in which the groups $\sfH$ and $\sfL$ are abelian\footnote{We always use additive notation for abelian groups and subgroups.}, the $\sfG$-actions are strict\footnote{Although the action of $\sfG$ on $\sfH$ is strict in this specialization, the equivariance of the Peiffer lifting remains weak. This residual weakening is precisely what is needed in the description of $T_2$-duality.}, i.e.
        \begin{equation}
            \sft(\alpha(g_1,g_2,h))=\sft(\beta(g,h_1,h_2))=0
        \end{equation}
        for all $g,g_{1,2}\in \sfG$ and $h,h_{1,2}\in \sfH$, and the structure maps satisfy the properties
        \begin{equation}
            \begin{gathered}
                \sft(h)\acton \ell=\ell~,
                \\
                \{\sft(\ell_1),\sft(\ell_2)\}=\{\sft(\{h_1,h_2\}),h_3\}=0~,
            \end{gathered}
        \end{equation}
        and
        \begin{equation}
            \begin{gathered}
                \alpha(\sft(h_1),\sft(h_2),h_3)=0~,
                \\
                \beta(\sft(h_1),h_2,h_3)=\beta(g,\sft(\ell),h)=\beta(g,h,\sft(\ell))=0~,
                \\
                g_1\acton \alpha(g_2,g_3,h)=\alpha(g_2,g_3,h)~,
                \\
                g_1\acton \beta(g_2,h_1,h_2)=\beta(g_2,h_1,h_2)
            \end{gathered}
        \end{equation}
        for all $g,g_{1,2,3}\in \sfG$, $h,h_{1,2,3}\in \sfH$, and $\ell,\ell_{1,2}\in \sfL$. Moreover, $\alpha$ and $\beta$ are additive in the $\sfH$-valued arguments, as is the Peiffer lifting $\{-,-\}$.
        
        For simplicity, we also demand that the Lie algebra of $\sfG$ is abelian. 
    \end{definition}
    
    The axioms of an almost 2-crossed module further simplify as follows for toroidal 3-groups.
    \begin{lemma}
        Consider a toroidal 3-group as above. Then we have 
        \begin{subequations}
            \begin{equation}
                \begin{aligned}
                    \{\sft(\ell),h\}&=-\{h,\sft(\ell)\}~, 
                    \\
                    h_1\acton \{h_2,h_3\}&=\{h_2,h_3\}~,
                \end{aligned}
            \end{equation}
            \begin{equation}
                \begin{aligned}
                    h_1\acton \alpha(g_1,g_2,h_2)&=\alpha(g_1,g_2,h_2)~,
                    \\
                    h_1\acton \beta(g,h_2,h_3)&=\beta(g,h_2,h_3)~,
                \end{aligned}
            \end{equation}
            \begin{equation}
                \begin{aligned}
                    \alpha(g_1g_2,g_3,h)&=\alpha(g_2,g_3,h)+ \alpha(g_1,g_2g_3,h)-\alpha(g_1,g_2,g_3 \triangleright h)~,
                    \\
                    \beta(g,h_1,h_2h_3) &= \beta(g,h_1,h_2)+ \beta(g,h_1h_2,h_3)-\beta(g,h_2,h_3)~,
                    \\
                    \beta(g_1g_2,h_1,h_2) 
                    &= \beta(g_1,g_2\acton h_1,g_2\acton h_2)+\beta(g_2,h_1,h_2)
                    \\
                    &\hspace{0.5cm}
                    +\alpha(g_1,g_2,h_1h_2)
                    -\alpha(g_1,g_2,h_2)-\alpha(g_1,g_2,h_1)~.
                \end{aligned}
            \end{equation}
            The axioms for the Peiffer lifting in~\eqref{eq:peiff_ax} reduce to linearity as well as the weak equivariance relation
            \begin{equation}
                \begin{aligned}
                    g \triangleright \{h_1,h_2\} &=\{g\acton h_1,g\acton h_2\}+\beta(g,h_1,h_2)-\beta(g,h_2,h_1)
                    \\
                    &\hspace{1.5cm}+\alpha(g,\sft(h_1),h_2)-\alpha(\sft(h_1),g,h_2)
                \end{aligned}
            \end{equation}
        \end{subequations}
        for all $g,g_{1,2,3}\in \sfG$, $h,h_{1,2,3}\in \sfH$, and $\ell\in\sfL$.
    \end{lemma}
    \begin{proof}
        These relations are immediate consequences of the axioms of toroidal 3-groups.
    \end{proof}
    
    Explicit examples of toroidal 3-groups will be introduced in \ref{sec:toroidal-3-group-examples,sec:affine}. The former will be relevant to a description of $T_2$-duality in terms of spans of principal 3-bundles; the latter will play the same role in affine $T_2$-duality.
    
    \subsection{Toroidal principal 3-bundles with connection}\label{ssec:toroidal-3-bundles}
    
    Let $\caG=(\sfL\rightarrow \sfH\rightarrow \sfG)$ be a toroidal 3-group. A principal $\caG$-bundle subordinate to a surjective submersion $\sigma:Y\rightarrow M$ is represented by smooth maps
    \begin{subequations}\label{eq:toroidal_cocycles}
        \begin{equation}
            (g_{ij})\in C^\infty(Y^{[2]},\sfG)~,~~~(h_{ijk})\in C^\infty(Y^{[3]},\sfH)~,~~~(\ell_{ijkl})\in C^\infty(Y^{[4]},\sfL)
        \end{equation}
        such that
        \begin{equation}
            \begin{aligned}
                \sft(h_{ijk})g_{ij}g_{jk} &= g_{ik}~,
                \\
                h_{ikl}+h_{ijk}+\sft(\ell_{ijkl}) &= h_{ijl}+ g_{ij} \triangleright h_{jkl}~,
            \end{aligned}
        \end{equation}
        and
        \begin{equation}
            \begin{aligned}
                &\ell_{ijkl} + ((-g_{ij} \triangleright h_{jkl}) \triangleright \ell_{ijlm} ) +\beta(g_{ij},h_{jlm},h_{jkl})+g_{ij} \triangleright \ell_{jklm}-\beta(g_{ij},h_{jkm},g_{jk}\acton h_{klm})
                \\
                &\hspace{2cm}=(-h_{ijk})\triangleright \ell_{iklm} +\{h_{ijk} , g_{ik} \triangleright h_{klm}\} -\alpha(\sft(h_{ijk}),g_{ik},h_{klm})
                \\
                &\hspace{4cm}-\alpha(g_{ij},g_{jk},h_{klm}) +(-(g_{ij}+g_{jk}) \triangleright h_{klm}) \triangleright \ell_{ijkm}
            \end{aligned}
        \end{equation}
        for all $ijk\in Y^{[3]}$, $ijkl\in Y^{[4]}$, and $ijklm\in Y^{[5]}$. These are the specializations of~\eqref{eq:cocycle3pbs} to a toroidal 3-group. We refer to a triple $((g_{ij}),(h_{ijk}),(\ell_{ijkl}))$ as a Čech cocycle. We call a principal 3-bundle with toroidal structure group a \emph{toroidal principal 3-bundle} for short.
    \end{subequations}
    
    We also develop the (higher) bundle isomorphisms or (higher) Čech coboundaries. We consider two such cocycles $((g_{ij}),(h_{ijk}),(\ell_{ijkl}))$ and $((\tilde g_{ij}),(\tilde h_{ijk}),(\tilde \ell_{ijkl}))$ equivalent if they are related by a 1-coboundary\footnote{We use the same letters for cocycles and coboundaries; they are distinguished by the number of indices.}
    \begin{subequations}
        \begin{equation}
            (g_i)\in C^\infty(Y,\sfG)~,~~~(h_{ij})\in C^\infty(Y^{[2]},\sfH)~,~~~(\ell_{ijk})\in C^\infty(Y^{[3]},\sfL)
        \end{equation}
        such that
        \begin{equation}
            \begin{aligned}
                \tilde g_{ij}&=\sft(h_{ij})-g_i+g_{ij}+g_j~,
                \\
                \tilde h_{ijk}&=(-g_i)\acton (h_{ik}+h_{ijk}+\sft(\ell_{ijk})-g_{ij}\acton h_{jk}-h_{ij})~,
            \end{aligned}
        \end{equation}
        and
        \begin{equation}
            \begin{aligned}
                \tilde \ell_{ijkl}&=\{\sft(\ell_{jkl}),h_{ij}\}+(-g_i)\acton\Big(-(g_{ik}\acton h_{kl})\acton \ell_{ijk}-(-g_{ij}\acton h_{jkl})\acton \ell_{ijl}
                \\
                &\hspace{6cm}-(-h_{ijk})\acton\ell_{ikl}+g_{ij}\acton \ell_{jkl}+R_{ijkl}\Big)
            \end{aligned}
        \end{equation}
        for all $ij\in Y^{[2]}$, $ijk\in Y^{[3]}$, and $ijkl\in Y^{[4]}$, where $R_{ijkl}$ is a complicated expression that vanishes if $\alpha$ and $\beta$ are trivial. We refrain from presenting this lengthy expression here, and only give $R_{ijkl}$ for the concrete relevant toroidal 3-groups. 
    \end{subequations}
    
    We consider two such 1-coboundaries $((g_{i}),(h_{ij}),(\ell_{ijk}))$ and $((\tilde g_{i}),(\tilde h_{ij}),(\tilde \ell_{ijk}))$ equivalent, if they are related by a 2-coboundary
    \begin{subequations}
        \begin{equation}
            (h_{i})\in C^\infty(Y,\sfH)~,~~~(\ell_{ij})\in C^\infty(Y^{[2]},\sfL)
        \end{equation}
        such that
        \begin{equation}
            \begin{aligned}
                \tilde g_{i}&=\sft(h_{i})g_i~,
                \\
                \tilde h_{ij}&=h_i+h_{ij}+\sft(\ell_{ij})-g_{ij}\acton h_j~,
            \end{aligned}
        \end{equation}
        and
        \begin{equation}
            \begin{aligned}
                \tilde \ell_{ijk}&=(g_{ik}\acton h_k)\acton\Big(-(-h_{ijk})\acton\ell_{ik}+\ell_{ijk}+(-g_{ij}\acton h_{jk})\acton \ell_{ij}+g_{ij}\acton \ell_{jk}+R_{ijk}\Big)
            \end{aligned}
        \end{equation}
        for all $i\in Y$, $ij\in Y^{[2]}$, and $ijk\in Y^{[3]}$, where $R_{ijk}$ is again a somewhat lengthy expression that vanishes when $\alpha$ and $\beta$ are trivial.
    \end{subequations}      
    
    Finally, we consider two such 2-coboundaries $((h_{i}),(\ell_{ij}))$ and $((\tilde h_{i}),(\tilde \ell_{ij}))$ equivalent if they are related by a 3-coboundary
    \begin{subequations}
        \begin{equation}
            (\ell_{i})\in C^\infty(Y^{[2]},\sfL)
        \end{equation}
        such that
        \begin{equation}
            \begin{aligned}
                \tilde h_{i}&=\sft(\ell_{i})+h_i~,
                \\
                \tilde \ell_{ij}&=-(-h_i-h_{ij})\acton \ell_i+\ell_{ij}+(-g_{ij}\acton h_j)\acton g_{ij}\acton \ell_j
            \end{aligned}
        \end{equation}
        for all $i\in Y$ and $ij\in Y^{[2]}$.
    \end{subequations}      
    
    Flat connections on higher principal bundles can be described using the higher inner action groupoid $\scA(B\caG)$ of the higher structure group $\caG$. For Lie groupoids, this action groupoid was defined in~\cite{Fischer:2024vak}, and this definition extends in an evident manner to deloopings of higher groups. In the case of a toroidal 3-group $\caG$, we obtain a simplicial differential graded (dg-)manifold. The differential cocycles describing flat principal $\caG$-bundles with connections are then morphisms of simplicial dg-manifolds 
    \begin{equation}
        T[1]\check \scC(\sigma)_\bullet \xrightarrow{~~~} \scA(\caG)_\bullet~.
    \end{equation}
    Bundle isomorphisms and higher isomorphisms are represented by simplicial homotopies and higher simplicial homotopies between these maps.
    
    Our applications, however, require non-flat connections. Contrary to ordinary principal bundles, the gauge transformations\footnote{i.e.~the action of bundle automorphisms} of flat connections do not extend unchanged to the non-flat case. Instead, one has to deform these by expressions proportional to the curvature forms, and these deformations are called an \emph{adjustment}, see~\cite{Saemann:2019dsl,Kim:2019owc,Rist:2022hci}, and especially~\cite{Gagliardo:2025oio}. 
    
    For a toroidal 3-group $(\sfL\xrightarrow{\sft}\sfH\xrightarrow{\sft}\sfG)$ with corresponding Lie 3-algebra\footnote{i.e.~the almost 2-crossed module of Lie algebras obtained from the tangent functor at the identity} $(\frl\xrightarrow{\sft}\frh\xrightarrow{\sft}\frg)$, an adjustment is encoded by smooth maps
    \begin{equation}
        \kappa_1:\sfG\times \frg\rightarrow \frh~,~~~\kappa_2:\sfG\times \frh\rightarrow \frl~,~~~\kappa_3:\sfH\times \frg\rightarrow \frl~,~~~\kappa_4:\sfG\times \sfG\times \frg\rightarrow \frl~,
    \end{equation}
    which then appear in the gauge transformations of the local connection forms.
    
    A local connection on a principal $\caG$-bundle subordinate to a surjective submersion $\sigma:Y\rightarrow M$ is described by forms
    \begin{equation}
        A\in \Omega^1(Y,\frg)
        ~,~~~
        B\in \Omega^2(Y,\frh)
        ~,~~~
        C\in \Omega^3(Y,\frl)~.
    \end{equation}
    Gauge transformations will be of the general form
    \begin{equation}
        g\in \Omega^0(Y,\sfG)~,~~~\Lambda\in \Omega^1(Y,\frh)~,~~~\Sigma\in \Omega^2(Y,\frl)~,
    \end{equation}
    and act according to
    \begin{equation}
        \begin{tikzcd}[column sep=3.5cm]
            (\tilde A,\tilde B,\tilde C) & \arrow[l,"{(A,B,C,g,\Lambda,\Sigma)}"'] (A,B,C)
        \end{tikzcd}
    \end{equation}
    with\footnote{Although $\sfG$ may be non-abelian, its Lie algebra $\frg$ is abelian by assumption.}
    \begin{equation}
        \begin{aligned}
            \tilde A&=g^{-1} Ag+g^{-1}\rmd g~,
            \\
            \tilde B&=g^{-1}\acton B+\rmd \Lambda+ \tilde A\acton \Lambda-\kappa_1(g^{-1},F)~,
            \\ 
            \tilde C&=g^{-1}\acton C-\rmd\Sigma+\tilde A\acton \Sigma-\{\tilde B,\Lambda\}-\{\Lambda,\tilde B\}+\{\Lambda,\rmd \Lambda+\tilde A\acton \Lambda\}
            \\
            &\hspace{2cm}+\kappa_3(\Lambda,F)+\kappa_2(g^{-1},H)-\kappa_4(g,A,F)~.
        \end{aligned}
    \end{equation}
    Similarly, one introduces 2- and 3-gauge transformations. Their actions in the non-flat case also acquire additional terms involving the adjustment maps $\kappa_{1,2,3,4}$.
    
    The two conditions (1) that all morphisms compose consistently and (2) that higher morphisms link lower morphisms with the same source and target yield conditions on the maps $\kappa_{1,2,3,4}$, the \emph{adjustment conditions}.
    
    The derivation of the cocycles and higher coboundaries for adjusted connections is involved even for general toroidal 3-groups. Since we will only need this data for the toroidal 3-group $\sfTDFTA_d$ and its truncation, the toroidal 3-group $\sfTDFT_d$, we refrain from presenting the general discussion. Details in this special case are found in~\ref{app:adjustment}.
    
    We stress that we will work with a shifted parameterization of higher gauge transformations. To explain this convention, consider the rewriting of the transformation\footnote{In this example, as for the group $\sfTDFTA_d$, $A\acton \Lambda=0$.}
    \begin{equation}
        \begin{aligned}
            B'&=g^{-1}\acton B+\rmd\Lambda-\kappa_1(g^{-1},\rmd A)
            \\
            &=g^{-1}\acton B+\rmd(\Lambda-\kappa_1(g^{-1},A))+\kappa_1(\rmd g^{-1},A)
            \\
            &=g^{-1}\acton B+\rmd \Lambda'+\kappa_1(\rmd g^{-1},A)
        \end{aligned}
    \end{equation}
    for bilinear $\kappa_1$ and $\Lambda'=\Lambda-\kappa_1(g^{-1},A)$. Throughout this paper, all higher gauge parameters are understood in this shifted convention.
    
    \section{Toroidal 3-groups for \texorpdfstring{$T_2$}{T2}-duality}\label{sec:toroidal-3-group-examples}
    
    Having prepared the general theory, we can now introduce the toroidal 3-groups that will allow us (a) to describe $T_2$-duality backgrounds as toroidal principal 3-bundles and (b)~to capture $T_2$-duality relations by spans of toroidal principal 3-bundles.
    
    \subsection{Categorical 3-tori and the toroidal 3-group \texorpdfstring{$\sfTDFT_d$}{T2DF2d}}\label{ssec:categorical_tori}
    
    We begin with categorical 3-tori, a class of toroidal 3-groups containing the structure 3-group $\sfTDFT_d$ required for $T_2$-duality. Categorical tori were introduced in~\cite{Ganter:2014zoa} as central extensions of the torus by the 2-group $\sfB\sfU(1)$. We generalize this by one categorical degree as follows.    
    \begin{definition}[Categorical 3-torus]\label{def:categorical-3-torus}
        Consider three lattices inside vector spaces, $\Lambda_i \subset V_i$, $i=0,1,2$, together with maps 
        \begin{equation}
            \langle -,-\rangle_0: \Lambda_0 \otimes \Lambda_0 \rightarrow \Lambda_1
            \eand 
            \langle -,-\rangle_1: \Lambda_0 \otimes \Lambda_1 \rightarrow \Lambda_2
        \end{equation}
        such that 
        \begin{equation}\label{eq:jacobi}
            \langle u_0, \langle v_0, w_0 \rangle_0 \rangle_1 + \langle v_0, \langle u_0, w_0 \rangle_0 \rangle_1 = 0~,
        \end{equation}
        and linearly extend these maps to the full vector spaces.
        
        Then the corresponding categorical 3-torus $\caT$ is a toroidal 3-group with underlying complex
        \begin{equation}
            \begin{tikzcd}[column sep=0.8cm,row sep=0.2cm]
                \Lambda_1 \times V_2/\Lambda_2 \arrow[r,"\sft"] 
                & \Lambda_0 \times V_1 \arrow[r,"\sft"] 
                & V_0 
                \\
                (\lambda^1,[v^2])\arrow[r,mapsto] & (0,\lambda^1)~,
                \\
                &(\lambda^0,v^1) \arrow[r,mapsto] & \lambda^0~,
            \end{tikzcd}
        \end{equation}
        and structure maps\footnote{See \ref{not:low} for the definition of $\langle-,\langle-,-\rangle\rangle^\low$.} 
        \begin{equation}
            \begin{aligned}
                v^0 \triangleright (\lambda^0,v^1) &= (\lambda^0, v^1 + \langle v^0, \lambda^0 \rangle_0)~, 
                \\
                v^0 \triangleright (\lambda^1,[v^2]) &= (\lambda^1, [v^2 + \langle v^0, \lambda^1 \rangle_1])~,
                \\
                (\lambda^0,v^1) \triangleright (\lambda^1,[v^2]) &= (\lambda^1,[v^2])~,
                \\
                \{(\lambda^0_1,v^1_1),(\lambda^0_2,v^1_2)\} &= (-\langle \lambda^0_1,\lambda^0_2 \rangle_0,0)~,
                \\
                \alpha(v^0_1,v^0_2,(\lambda^0,v^1)) &= (0,[-\langle v^0_1, \langle v^0_2, \lambda^0 \rangle_0 \rangle_1^{\low}])~,
                \\
                \beta(v^0,(\lambda^0_1,v^1_1),(\lambda^0_2,v^1_2)) &= 0~.
            \end{aligned}
        \end{equation}
    \end{definition}
    
    Substituting the above structure maps into the cocycle equations~\eqref{eq:toroidal_cocycles}, we obtain the following explicit model for the delooping $B\mathcal{T}$ of a categorical torus $\caT$:
    \begin{align}\label{eq:bhightori}
        \begin{split}
            (B\mathcal{T})_n &\coloneqq  \Big\{ ((v^0_{ij}),(\lambda^0_{ijk},v^1_{ijk}),(\lambda^1_{ijkl},[v^2_{ijkl}]))\in V_0^{\binom{n+1}{2}} \!\times (\Lambda_0 \times V_1)^{\binom{n+1}{3}} \times (\Lambda_1 \times V_2/\Lambda_2)^{\binom{n+1}{4}} 
            \\
            &\hspace{1.5cm}\Big| \: v^0_{ij} - v^0_{ik} + v^0_{jk} = -\lambda^0_{ijk}~,
            \\
            &\hspace{1.5cm}\phantom{\Big|}\:v^1_{ijk} - v^1_{ijl} + v^1_{ikl} - v^1_{jkl} = \langle v^0_{ij}, \lambda^{0}_{jkl}  \rangle_0 -\lambda^1_{ijkl}~,
            \\
            &\hspace{1.5cm}\phantom{\Big|}\:   [v^2_{ijkl} - v^2_{ijkm} + v^2_{ijlm} - v^2_{iklm} + v^2_{jklm}]
            \\
            &\hspace{2cm}=[-\langle v^0_{ij}, \lambda^1_{jklm} \rangle_1 + \langle \lambda^0_{ijk}, \langle v^0_{ik}, \lambda^0_{klm} \rangle_0 \rangle_1^{\low} + \langle v^0_{ij}, \langle v^0_{jk}, \lambda^0_{klm} \rangle_0 \rangle_1^{\low}] \Big\}~.
        \end{split}
    \end{align}    
    As before, we assume here normalized cocycles in which all components with repeated indices are trivial.
    
    We note that for trivial $V_2$ and $\Lambda_2$, this toroidal 3-group is equivalent to an ordinary crossed module of groups $\mathcal{T}_{\leq 1}$: we can replace the 3-term complex $\Lambda_1\rightarrow \Lambda_0\times V_1\rightarrow V_0$ by the 2-term complex $\Lambda_0\times V_1/\Lambda_1\rightarrow V_0$. In this way, we recover the categorical 2-tori constructed from $\langle-,-\rangle_0$ in~\cite{Ganter:2014zoa}. 
    
    \begin{proposition}\label{prop:charclass}
        Consider a categorical 3-torus $\caT$ given by the data 
        \begin{equation}
            (\Lambda_0,\Lambda_1,\Lambda_2,\langle -,-\rangle_0,\langle -,-\rangle_1)~,
        \end{equation}
        as above, let $\mathcal{T}_{\leq 1}$ be the crossed module of groups defined by $(\Lambda_0,\Lambda_1,\langle-,-\rangle_0)$, and let $M$ be a manifold. Then:
        \begin{myenumerate}
            \item\label{it:charclass1} a principal $V_0/\Lambda_0$-bundle $P \rightarrow M$ determines a class in $H^4(M,\Lambda_1)$ which vanishes precisely when $P$ can be lifted to a principal $\mathcal{T}_{\leq 1}$-bundle;
            \item\label{it:charclass2} a principal $\mathcal{T}_{\leq 1}$-bundle $\scP \rightarrow M$ determines a class in $H^5(M,\Lambda_2)$ which vanishes precisely when $\scP$ can be lifted to a principal $\mathcal{T}$-bundle.
        \end{myenumerate}
    \end{proposition}
    \begin{proof}
        The first statement is well known, e.g.~\cite{Stevenson:0106018}. Let $[v^0_{ij}]: Y^{[2]}\rightarrow V_0/\Lambda_0$ be the cocycle data for $P$, subordinate to some surjective submersion $\sigma:Y\rightarrow M$. We can choose lifts of these functions to $V_0$-valued functions $v^0_{ij}:Y^{[2]}\rightarrow V_0$ and write 
        \begin{equation}
            \lambda^0_{ijk}:Y^{[3]}\rightarrow \Lambda_0~,~~~\lambda^0_{ijk} \coloneqq  -v^0_{ij} + v^0_{ik} - v^0_{jk}~.
        \end{equation}
        A direct calculations shows that
        \begin{equation}
            \xi_{ijkl}:Y^{[4]}\rightarrow \Lambda_1~,~~~
            \xi_{ijkl}\coloneqq -\langle \lambda^0_{ijk},\lambda^0_{jkl} \rangle_0 
        \end{equation}
        determines a \v{C}ech cocycle and thus a class $[\langle \lambda^0 , \lambda^0 \rangle_0] \in H^4(M,\Lambda_1)$. A lift of $P$ to a principal $\mathcal{T}_{\leq 1}$-bundle is determined by functions $v^1_{ijk}: Y^{[3]}\rightarrow V_1$ and $\lambda^1_{ijkl}:Y^{[4]}\rightarrow \Lambda_1$ as above, satisfying 
        \begin{equation}
            v^1_{ijk} - v^1_{ijl} + v^1_{ikl} - v^1_{jkl} = \langle v^0_{ij}, \lambda^{0}_{jkl}  \rangle_0 -\lambda^1_{ijkl}~,
        \end{equation}
        cf.~\eqref{eq:bhightori}, which exist precisely when the cocycle $\xi_{ijkl}$ is a coboundary.
        
        To prove the second statement, consider a lift of $P$ to a principal $\mathcal{T}_{\leq 1}$-bundle $\scP$ given by one such choice of functions $v^1_{ijk}$, $\lambda^1_{ijkl}$. The corresponding class in $H^5(M,\Lambda_2)$ is given by the $\Lambda_2$-valued \v{C}ech cocycle
        \begin{equation}
            - \langle \lambda^0_{ijk}, \lambda^1_{klmn} \rangle + \langle \lambda^0_{ijk}, \langle \lambda^0_{ikl}, \lambda^0_{lmn} \rangle \rangle^{\low} - \langle \lambda^0_{jkl}, \langle \lambda^0_{ijl}, \lambda^0_{lmn} \rangle \rangle^{\low}~.
        \end{equation}
        From~\eqref{eq:bhightori}, it is again clear that this class vanishes precisely when $\scP$ lifts to a principal $\mathcal{T}$-bundle.
    \end{proof}
    
    A categorical 3-torus $\caT$ can be seen as a central extension of 3-groups,
    \begin{equation}
        B^2(V_2/\Lambda_2) \rightarrow \mathcal{T} \rightarrow \mathcal{T}_{\leq 1}~.
    \end{equation}
    As such, the following proposition says that $\mathcal{T}$ as defined from the data
    \begin{equation}
        (\Lambda_0,\Lambda_1,\Lambda_2,\langle -,-\rangle_0,\langle -,-\rangle_1)
    \end{equation}
    is a square root of the Lie 2-crossed module defined from the same data in~\cite{Gagliardo:2025oio}.
    \begin{proposition}\label{prop:stricttorus}
        Assume that $\tfrac{1}{2}\langle -,-\rangle_1$ is $\Lambda_2$-valued on $\Lambda_0 \otimes \Lambda_1$. Then $\mathcal{T}$ is isomorphic to the Lie $2$-crossed module $\tilde{\mathcal{T}}$ defined by $(\Lambda_0,\Lambda_1,\Lambda_2,\langle -,-\rangle_0,\tfrac{1}{2}\langle -,-\rangle_1)$ in~\cite{Gagliardo:2025oio}. 
    \end{proposition}
    \begin{proof} 
        The toroidal 3-group $\tilde{\mathcal{T}}$ has underlying complex
        \begin{equation}
            \begin{tikzcd}[column sep=0.8cm,row sep=0.2cm]
                \Lambda_1 \times V_2/\Lambda_2 \arrow[r,"\sft"] 
                & \Lambda_0 \times V_1 \arrow[r,"\sft"] 
                & V_0 
                \\
                (\lambda_1,[u_2])\arrow[r,mapsto] & (0,\lambda_1)~,
                \\
                &(\lambda_0,u_1) \arrow[r,mapsto] & \lambda_0~,
            \end{tikzcd}
        \end{equation}
        and structure maps
        \begin{equation}
            \begin{aligned}
                u_0 \triangleright (\lambda_1,[u_2]) &= (\lambda_1, [u_2 + \tfrac{1}{2}\langle u_0, \lambda_1 \rangle_1])~,
                \\
                (\lambda_0,u_1) \triangleright (\lambda_1,[u_2]) &= (\lambda_1,[u_2])~,
                \\
                \{(\lambda_0,u_1),(\mu_0,v_1)\} &= (-\langle \lambda_0,\mu_0 \rangle_0,[\tfrac{1}{2}\langle \lambda_0,v_1 \rangle])~,
                \\
                \alpha(u_0,v_0,(\lambda_0,u_1)) &= (0,[0])~,
                \\
                \beta(u_0,(\lambda_0,u_1),(\mu_0,v_1)) &= 0~.
            \end{aligned}
        \end{equation}
        The corresponding delooping is 
        \begin{align}
            \begin{split}
                (B\mathcal{\tilde T})_n &\coloneqq  \Big\{ ((\tilde v^0_{ij}),(\tilde \lambda^0_{ijk},\tilde v^1_{ijk}),(\tilde \lambda^1_{ijkl},[\tilde v^2_{ijkl}])) \in V_0^{\binom{n+1}{2}}\!\times\! (\Lambda_0 \times V_1)^{\binom{n+1}{3}} \times (\Lambda_1 \times V_2/\Lambda_2)^{\binom{n+1}{4}} 
                \\
                &\hspace{1.5cm}\Big| \: \tilde v^0_{ij} - \tilde v^0_{ik} + \tilde v^0_{jk} = -\tilde \lambda^0_{ijk}~,
                \\
                &\hspace{1.5cm}\phantom{\Big|}\:\tilde v^1_{ijk} - \tilde v^1_{ijl} + \tilde v^1_{ikl} - \tilde v^1_{jkl} = \langle \tilde v^0_{ij}, \tilde \lambda^{0}_{jkl}  \rangle_0 -\tilde \lambda^1_{ijkl}~,
                \\
                &\hspace{1.5cm}\phantom{\Big|}\:   [\tilde v^2_{ijkl} - \tilde v^2_{ijkm} + \tilde v^2_{ijlm} - \tilde v^2_{iklm} + \tilde v^2_{jklm}]
                \\
                &\hspace{2cm}=[-\tfrac{1}{2}\langle \tilde v^0_{ij}, \tilde \lambda^1_{jklm} \rangle_1 + \tfrac{1}{2} \langle \tilde \lambda^0_{ijk}, \tilde v^1_{klm} + \langle \tilde v^0_{ik}, \tilde \lambda^0_{klm} \rangle_0 \rangle_1] \Big\}~,
            \end{split}
        \end{align}
        and the isomorphism of 3-groups $\tilde{\mathcal{T}} \rightarrow \mathcal{T}$ is given by the simplicial map $B\tilde{\mathcal{T}} \rightarrow B\mathcal{T}$ which is the identity except for 
        \begin{align}\label{eq:isostrict}
            v^2_{ijkl} \coloneqq  \tilde{v}^2_{ijkl} + \tfrac{1}{2}\langle \tilde v^0_{ij},\tilde v^1_{jkl} \rangle - \tfrac{1}{2}\langle \tilde v^0_{ij}, \langle \tilde v^0_{ij}, \tilde \lambda^0_{jkl} \rangle \rangle^{\low}.  
        \end{align}
    \end{proof}
    
    We now specialize \ref{def:categorical-3-torus} to the categorical 3-torus relevant to $T_2$-duality.
    \begin{definition}[3-group $\sfTDFT_d$]
        The toroidal 3-group $\sfTDFT_d$ is the categorical 3-torus constructed from the following data:
        \begin{subequations}
            \begin{equation}
                \begin{aligned}
                    V_0 &= \mathbb R^d \oplus \Lambda^2 \mathbb R^{d,\ast}~, 
                    \quad &V_1 &= \mathbb R^{d,\ast}~, 
                    \quad &V_2 &= \mathbb R~, 
                    \\
                    \Lambda_0 &= \IZ^d \oplus \Lambda^2 \IZ^{d,\ast}~, 
                    \quad &\Lambda_1 &= \IZ^{d,\ast}~, 
                    \quad &\Lambda_2 &= \IZ~, 
                \end{aligned}
            \end{equation}
            \begin{equation}\label{eq:pairingst2dnf2}
                \langle v^0_1+\hat v^0_1,v^0_2+\hat v^0_2\rangle_0 \coloneqq \iota_{v^0_1}\hat v^0_2
                ~,~~~
                \langle v^0_1 + \hat v^0_1,\lambda^1 \rangle_1 \coloneqq \iota_{v^0_1} \lambda^1
            \end{equation}
            for $v^0_{1,2}\in \mathbb R^d$, $\hat v^0_{1,2}\in \Lambda^2 \mathbb R^{d,\ast}$, and $\lambda^1\in \IZ^{d,\ast}$.
        \end{subequations}
    \end{definition}
    \noindent To understand the above choice, note that 
    \begin{equation}
        V_0/\Lambda_0\cong \IT^d\times \Lambda^2\IT^{d,\ast}
        ~,~~~
        V_1/\Lambda_1\cong \IT^{d,\ast}
        \eand
        V_2/\Lambda_2\cong \sfU(1)~,
    \end{equation}
    so that the toroidal 3-group $\sfTDFT_d$ can be seen as a 2-gerbe with band $\sfU(1)$ over the space
    \begin{equation}
        \IT^d\times \caM(\IT^d)=\mathbb T^d \times \Lambda^2 \mathbb T^{d,\ast} \times B\mathbb T^{d,\ast}~.
    \end{equation}
    
    The Čech cocycles~\eqref{eq:toroidal_cocycles} for toroidal principal 3-bundles with structure 3-group $\sfTDFT_d$ then consist of the following data:
    \begin{subequations}\label{eq:cocycledatat2d}
        \begin{equation}
            \begin{aligned}
                (v^0_{ij},\hat{v}^0_{ij}): Y^{[2]} &\rightarrow \mathbb R^d \times \Lambda^2\mathbb R^{d,\ast}~,
                \\
                (\lambda^0_{ijk},\hat{\lambda}^0_{ijk},v^1_{ijk}): Y^{[3]}&\rightarrow \IZ^d \times \Lambda^2\IZ^{d,\ast} \times \mathbb R^{d,\ast}~,
                \\
                (\lambda^1_{ijkl},[v^2_{ijkl}]): Y^{[4]} &\rightarrow \IZ^{d,\ast} \times \mathbb R/\IZ~,
            \end{aligned}
        \end{equation}
        satisfying
        \begin{equation}\label{eq:cocycleT2D}
            \begin{aligned}
                v^0_{ij} - v^0_{ik} + v^0_{jk} &= -\lambda^0_{ijk}~,
                \\
                \hat{v}^0_{ij} - \hat{v}^0_{ik} + \hat{v}^0_{jk} &= -\hat{\lambda}^0_{ijk}~,
                \\
                v^1_{ijk} - v^1_{ijl} + v^1_{ikl} - v^1_{jkl} &= \langle v^0_{ij}, \hat{\lambda}^{0}_{jkl}  \rangle_0 -\lambda^1_{ijkl}~,
                \\
                [v^2_{ijkl} - v^2_{ijkm} + v^2_{ijlm} - v^2_{iklm} + v^2_{jklm}]& 
                \\
                =[-\langle v^0_{ij}, \lambda^1_{jklm} \rangle_1 + &\langle \lambda^0_{ijk}, \langle v^0_{ik}, \hat{\lambda}^0_{klm} \rangle_0 \rangle_1^\low + \langle v^0_{ij}, \langle v^0_{jk}, \hat{\lambda}^0_{klm} \rangle_0 \rangle_1^\low]
            \end{aligned}
        \end{equation}
        for all $ijk\in Y^{[3]}$, $ijkl\in Y^{[4]}$, and $ijklm\in Y^{[5]}$. 
    \end{subequations}
    
    The simplicial model $(B\sfTDFT_d)\bullet$ for the delooping of $\sfTDFT_d$ is then defined by 
    \begin{equation}
        \begin{aligned}
            (B\sfTDFT_d)_n &\coloneqq  \Big\{ ((v^0_{ij},\hat{v}^0_{ij}),(\lambda^0_{ijk},\hat{\lambda}^0_{ijk},v^1_{ijk}),(\lambda^1_{ijkl},[v^2_{ijkl}])) 
            \\
            &\hspace{1cm}
            \in (\mathbb R^d \times \Lambda^2\mathbb R^{d,\ast})^{\binom{n+1}{2}} 
            \times (\IZ^d \times \Lambda^2\IZ^{d,\ast})^{\binom{n+1}{3}} 
            \times (\IZ^{d,\ast} \times \mathbb R/\IZ)^{\binom{n+1}{4}} 
            \\
            &\hspace{1cm}
            \Big| \: \mbox{data satisfying~\eqref{eq:cocycleT2D}}~\Big\}~.
        \end{aligned}
    \end{equation}
    
    \subsection{Toroidal 3-groups for \texorpdfstring{$T_2$}{T2}-duality backgrounds}\label{ssec:T2-duality_3-groups}
    
    We require two further toroidal 3-groups, encoding left and right $T_2$-duality backgrounds, respectively. For left $T_2$-duality backgrounds, we will use the following toroidal 3-group.
    \begin{definition}[Toroidal 3-group $\sfTBLFT_d$]\label{def:TBLFT}
        The toroidal 3-group
        $\sfTBLFT_d$ has underlying complex
        \begin{subequations}\label{eq:TBLFT}
            \begin{equation}
                \begin{tikzcd}[column sep=0.8cm,row sep=0.2cm]
                    \IZ^{d,\ast} \times \mathbb R/\IZ \arrow[r,"\sft"] & \IZ^d \times \Lambda^2 \IZ^{d,\ast} \arrow[r,"\sft"] & \mathbb R^d
                    \\
                    (\lambda^1,[v^2]) \arrow[r,mapsto] & (0,0)~,
                    \\
                    &(\lambda^0,\hat \lambda^0) \arrow[r,mapsto] & \lambda^0~,
                \end{tikzcd}
            \end{equation}
            and structure maps\footnote{See \ref{not:low} in \ref{app:notation} for details on the notation $B^\low$.}
            \begin{equation}
                \begin{aligned}
                    v^0 \triangleright (\lambda^0,\hat \lambda^0) &= (\lambda^0,\hat \lambda^0)~,
                    \\
                    v^0 \triangleright (\lambda^1,[v^2]) &= (\lambda^1, [v^2 + \iota_{v^0}\lambda^1])~,
                    \\
                    (\lambda^0,\hat \lambda^0) \triangleright (\lambda^1,[v^2]) &= (\lambda^1,[v^2])~,
                \end{aligned}
            \end{equation}
            \begin{equation}
                \begin{aligned}
                    \{(\lambda^0_1,\hat \lambda^0_1),(\lambda^0_2,\hat \lambda^0_2)\} &= (-\iota_{\lambda^0_1}\hat \lambda^0_2,[0])~,
                    \\
                    \alpha(v^0_1,v^0_2,(\lambda^0,\hat \lambda^0)) &= (0,[-\hat \lambda^{0,\low}(v^0_2,v^0_1)])~,
                    \\
                    \beta(v^0,(\lambda^0_1,\hat \lambda^0_1),(\lambda^0_2,\hat \lambda^0_2)) &= (0,[0])~.
                \end{aligned}
            \end{equation}
        \end{subequations}
    \end{definition}
    \noindent The above structure maps satisfy \ref{def:toroidal_3-group}, as one readily verifies.
    
    A principal $\sfTBLFT_d$-bundle over $M$ subordinate to a surjective submersion $\sigma:Y\rightarrow M$ is described by the following cocycle data:
    \begin{subequations}
        \begin{equation}\label{eq:cocycledatat2b}
            \begin{aligned}
                (v^0_{ij}): Y^{[2]} &\rightarrow \mathbb R^d~,
                \\
                (\lambda^0_{ijk},\hat{\lambda}^0_{ijk}): Y^{[3]} &\rightarrow \IZ^d \times \Lambda^2\IZ^{d,\ast}~,
                \\
                (\lambda^1_{ijkl},[v^2_{ijkl}]): Y^{[4]}&\rightarrow \IZ^{d,\ast} \times \mathbb R/\IZ
            \end{aligned}
        \end{equation}
        satisfying 
        \begin{equation}
            \begin{aligned}
                v^0_{ij} - v^0_{ik} + v^0_{jk} &= -\lambda^0_{ijk}~,
                \\
                \lambda^0_{ijk} - \lambda^0_{ijl} + \lambda^0_{ikl} - \lambda^0_{jkl} &= 0~,
                \\
                \hat{\lambda}^0_{ijk} - \hat{\lambda}^0_{ijl} + \hat{\lambda}^0_{ikl} - \hat{\lambda}^0_{jkl} &= 0~,
                \\
                \lambda^1_{ijkl} - \lambda^1_{ijkm} + \lambda^1_{ijlm} - \lambda^1_{iklm} + \lambda^1_{jklm} &= - \langle \lambda^0_{ijk}, \hat{\lambda}^{0}_{klm} \rangle_0~,
                \\
                [v^2_{ijkl} - v^2_{ijkm} + v^2_{ijlm} - v^2_{iklm} + v^2_{jklm}] 
                \\
                =[-\langle v^0_{ij}, \lambda^1_{jklm} \rangle_1 &+ \langle \lambda^0_{ijk}, \langle v^0_{ik}, \hat{\lambda}^0_{klm} \rangle_0 \rangle_1^\low + \langle v^0_{ij}, \langle v^0_{jk}, \hat{\lambda}^0_{klm} \rangle_0 \rangle_1^\low]
            \end{aligned}
        \end{equation}
    \end{subequations}
    for all $ijk\in Y^{[3]}$, $ijkl\in Y^{[4]}$ and $ijklm\in Y^{[5]}$.
    
    The following 3-group is appropriate for right $T_2$-duality backgrounds.
    \begin{definition}[Toroidal 3-group $\sfTBRFT_d$]\label{def:TBRFT}
        The toroidal 3-group $\sfTBRFT_d$ has underlying complex
        \begin{subequations}\label{eq:TBRFT}
            \begin{equation}
                \begin{tikzcd}[column sep=0.8cm,row sep=0.2cm]
                    \mathbb R/\IZ  \arrow[r,"\sft"] & \IZ^d \times \Lambda^2 \IZ^{d,\ast} \times \mathbb T^{d,\ast}  \arrow[r,"\sft"] & \Lambda^2 \mathbb R^{d,\ast} 
                    \\
                    {[w^2]}  \arrow[r,mapsto] & (0,0,[0])~,
                    \\
                    &(\lambda^0,\hat \lambda^0,[w^1]) \arrow[r,mapsto] & \hat \lambda^0~,
                \end{tikzcd}
            \end{equation}
            and structure maps
            \begin{equation}
                \begin{aligned}
                    \hat v^0 \triangleright (\lambda^0,\hat \lambda^0,[w^1]) &= (\lambda^0,\hat \lambda^0,[w^1+\iota_{\lambda^0} \hat v^0])~,
                    \\
                    \hat v^0 \triangleright [w^2] &= [w^2]~,
                    \\
                    (\lambda^0,\hat \lambda^0,[w^1]) \triangleright [w^2] &= [w^2]~,
                \end{aligned}
            \end{equation}
            \begin{equation}
                \begin{aligned}
                    \{(\lambda^0_1,\hat \lambda^0_1,[w^1_1]),(\lambda^0_2,\hat \lambda^0_2,[w^1_2])\} &= [\iota_{\lambda^0_1}w^1_2]~,
                    \\
                    \alpha(\hat v^0_1,\hat v^0_2,(\lambda^0,\hat \lambda^0,[w^1])) &= [0]~,
                    \\
                    \beta(\hat v^0,(\lambda^0_1,\hat \lambda^0_1,[w^1_1]),(\lambda^0_2,\hat \lambda^0_2,[w^1_2])) &= [\hat v^{0,\low}(\lambda^0_1,\lambda^0_2)]~.
                \end{aligned}
            \end{equation}
        \end{subequations}
    \end{definition}
    
    A principal $\sfTBRFT_d$-bundle subordinate to a surjective submersion $\sigma:Y\rightarrow M$ is then determined by the following data:
    \begin{subequations}\label{eq:cocycledatat2rb}
        \begin{equation}
            \begin{aligned}
                \hat{v}^0_{ij}: Y^{[2]} &\rightarrow \Lambda^2 \mathbb R^{d,\ast}~,
                \\
                (\lambda^0_{ijk},\hat{\lambda}^0_{ijk},[w^1_{ijk}]): Y^{[3]} &\rightarrow \IZ^d \times \Lambda^2\IZ^{d,\ast} \times \mathbb T^{d,\ast}~,
                \\
                [w^2_{ijkl}]: Y^{[4]} &\rightarrow  \mathbb R/\IZ
            \end{aligned}
        \end{equation}
        satisfying 
        \begin{equation}
            \begin{aligned}
                \hat{v}^0_{ij} - \hat{v}^0_{ik} + \hat{v}^0_{jk} &= -\hat{\lambda}^0_{ijk}~,
                \\
                \lambda^0_{ijk} - \lambda^0_{ijl} + \lambda^0_{ikl} - \lambda^0_{jkl} &= 0~,
                \\
                [w^1_{ijk} - w^1_{ijl} + w^1_{ikl} - w^1_{jkl}] &= [\langle \lambda^0_{jkl}, \hat{v}^0_{ij} \rangle]~,
                \\
                [w^2_{ijkl} - w^2_{ijkm} + w^2_{ijlm} - w^2_{iklm} + w^2_{jklm}] & 
                \\
                =[\langle \lambda^0_{klm}, \langle \lambda^0_{jkm}, \hat{v}^0_{ij} \rangle \rangle^\low - \langle \lambda^0_{jkl}, \langle &\lambda^0_{jlm}, \hat{v}^0_{ij} \rangle\rangle^\low + \langle \lambda^0_{ijk}, w^1_{klm} + \langle \lambda^0_{klm}, \hat{v}^0_{ik} \rangle \rangle ]
            \end{aligned}
        \end{equation}
    \end{subequations}
    for all $ijk\in Y^{[3]}$, $ijkl\in Y^{[4]}$ and $ijklm\in Y^{[5]}$.
    
    The three toroidal 3-groups defined so far can be arranged into a span.
    \begin{proposition}\label{prop:span_3_groups}
        There is a span (of weak homomorphisms) of toroidal 3-groups
        \begin{equation}\label{eq:span_3_groups}
            \begin{tikzcd}
                & \sfTDFT_d \arrow[dl,"\phi"'] \arrow[dr,"\hat \phi"]
                \\
                \sfTBLFT_d & & \sfTBRFT_d 
            \end{tikzcd}
        \end{equation}
        and the induced maps $|B\phi|$ and $|B\hat \phi|$ on geometric realizations of the deloopings are weak homotopy equivalences.
    \end{proposition}
    \begin{proof}
        The complexes underlying $\sfTDFT_d$, $\sfTBLFT_d$, and $\sfTBRFT_d$ read as
        \begin{equation}
            \begin{tikzcd}[column sep=0.8cm,row sep=0.2cm]
                \IZ^{d,\ast} \times \IR/\IZ \arrow[r,"\sft"] 
                & \IZ^d \times \Lambda^2 \IZ^{d,\ast} \times \mathbb R^{d,\ast} \arrow[r,"\sft"] 
                & \mathbb R^d \times \Lambda^2 \mathbb R^{d,\ast}~,
                \\
                \IZ^{d,\ast} \times \mathbb R/\IZ \arrow[r,"\sft"] & \IZ^d \times \Lambda^2 \IZ^{d,\ast} \arrow[r,"\sft"] & \mathbb R^d~,
                \\
                \mathbb R/\IZ  \arrow[r,"\sft"] & \IZ^d \times \Lambda^2 \IZ^{d,\ast} \times \mathbb T^{d,\ast}  \arrow[r,"\sft"] & \Lambda^2 \mathbb R^{d,\ast}~,
            \end{tikzcd}
        \end{equation}
        respectively.
        
        The homomorphism $\phi$ is strict: explicitly, it is given by forgetting the components in $\Lambda^2\mathbb R^{d,\ast}$ and $\mathbb R^{d,\ast}$. Note that the fiber of the map $B\phi:B\sfTDFT_d\rightarrow B \sfTBLFT_d$ over the unit simplicial point is the simplicial manifold $B(\Lambda^2 \mathbb R^{d,\ast}) \times B^2(\mathbb R^{d,\ast})$, whose geometric realization is contractible. Since the map $B\phi$ is a degree-wise surjective submersion, it follows that the induced map $|B \sfTDFT_d|\rightarrow |B \sfTBLFT_d| $ is a weak homotopy equivalence.
        
        The map $\hat \phi$ is a weak homomorphism in the sense of \ref{def:weak_morphisms}. That is, it is encoded in a simplicial map $B\hat\phi:B\sfTDFT_d\rightarrow B\sfTBRFT_d$. It is given by the identity on the common components and 
        \begin{align}
            \begin{split}\label{eq:homr}
                [w^1_{ijk}] &\coloneqq  [v^1_{ijk} + \langle v^0_{ij},\hat{v}^0_{jk} \rangle + \langle \lambda^0_{ijk},v^0_{ik} \rangle]~,
                \\
                [w^2_{ijkl}] &\coloneqq  [v^2_{ijkl} - \langle v^0_{ij}, v^1_{jkl} \rangle + \langle  v^0_{jk}, \langle v^0_{ij}, \hat{v}^0_{kl} \rangle \rangle^\low + \langle \lambda^0_{ijk}, \langle v^0_{ik},  \hat{v}^0_{kl} \rangle \rangle^\low  
                \\
                &\quad + \langle \lambda^0_{ijk}, \langle \lambda^0_{ikl}, \hat{v}^0_{il} \rangle \rangle^\low -  \langle \lambda^0_{jkl}, \langle \lambda^0_{ijl}, \hat{v}^0_{il} \rangle \rangle^\low ]
            \end{split}
        \end{align}
        on the remaining ones. One can check that the fiber of the map $B\hat \phi$ over the simplicial unit is the simplicial manifold $B\mathbb R^d \times E(B^2 \IZ^{d,\ast})$, whose geometric realization is contractible. Since $B\hat \phi$ is again a surjective submersion, it follows also here that $|B\hat \phi|:|B\sfTDFT_d| \rightarrow |B\sfTBRFT_d|$ is a weak homotopy equivalence. 
    \end{proof}
    
    The weak homotopy equivalence of the classifying spaces then yields the following, important corollary.
    \begin{corollary}\label{cor:any_bundle_from_T2D}
        Up to bundle isomorphisms, every principal $\sfTBLFT_d$-bundle and every principal $\sfTBRFT_d$-bundle can be obtained from a principal $\sfTDFT_d$-bundle. 
    \end{corollary}
    \noindent Concretely, the Čech cocycle data for the principal $\sfTBLFT_d$- and $\sfTBRFT_d$-bundles are obtained from the Čech data for a principal $\sfTDFT_d$-bundle by direct application of the morphisms $B\phi$ and $B\hat \phi$.

    \section{\texorpdfstring{$T_2$}{T2}-duality as spans of principal 3-bundles}\label{sec:t2df2}
    
    In this section, we show that left and right $T_2$-duality backgrounds can be represented by principal $\sfTBLFT_d$- and $\sfTBRFT_d$-bundles for the toroidal 3-groups introduced in \ref{def:TBLFT} and \ref{def:TBRFT}. As we saw above, the span of toroidal 3-groups~\eqref{eq:span_3_groups} directly induces spans of principal 3-bundles with these three toroidal 3-groups as structure groups:
    \begin{equation}
        \begin{tikzcd}
            & \sfTDFT_d \arrow[dl,"\phi"'] \arrow[dr,"\hat \phi"]
            \\
            \sfTBLFT_d & & \sfTBRFT_d 
        \end{tikzcd}
        ~~~\longrightarrow~~~
        \begin{tikzcd}
            & \scP_\rmC\arrow[dl,"p"'] \arrow[dr,"\hat p"]
            \\
            \scP & & \hat \scP
        \end{tikzcd}
    \end{equation}
    We also show that a $T_2$-duality between the backgrounds described by $\scP$ and $\hat \scP$ is precisely such a span of toroidal principal 3-bundles, which is fully encoded in the principal $\sfTDFT_d$-bundle $\scP_\rmC$.
    
    \subsection{Topological \texorpdfstring{$T_2$}{T2}-duality backgrounds as higher bundles}
    
    We first interpret left $T_2$-duality backgrounds as toroidal principal 3-bundles.
    \begin{proposition}\label{prop:leftbackgroundsf2}
        A principal $\sfTBLFT_d$-bundle over $M$ determines a topological left $T_2$-duality background of type $F_2$, and any such background arises from a principal $\sfTBLFT_d$-bundle up to equivalence.
    \end{proposition}
    \begin{proof}
        Consider the cocycle data~\eqref{eq:cocycledatat2b} for a principal $\sfTBLFT_d$-bundle subordinate to $\sigma:Y\rightarrow M$, from which we construct the simplicial manifold $R_\bullet$ with
        \begin{align}
            \begin{split}
                R_n &\coloneqq  \Big\{ ((v^0_{i}),([v^2_{ijk}])) \in Y^{[n+1]}\times (\mathbb R^d)^{n+1} \times (\mathbb R/\IZ)^{\binom{n+1}{3}} \\
                &\hspace{1cm}\Big| \: \: [v^0_{i} - v^0_{j}] = [v^0_{ij}]~,
                \\
                &\hspace{1cm}\:\: \phantom{\Big|}[v^2_{ijk} - v^2_{ijl} + v^2_{ikl} - v^2_{jkl}]= [v^2_{ijkl} - \langle v^0_i,\lambda^1_{ijkl} \rangle + \langle v^0_i, \langle v^0_{ij},  \hat{\lambda}^0_{jkl} \rangle \rangle^\low 
                \\
                &\hspace{1.5cm} \quad \quad \quad \quad \quad \quad \quad \quad \quad \quad \quad \quad \quad \quad  + \langle v^0_i - v^0_j - v^0_{ij}, \langle v^0_j,\hat{\lambda}^0_{jkl} \rangle \rangle^\low ] \Big\}~,
            \end{split}
        \end{align}
        where the equivalence classes are in the quotients $\IR^d/\IZ^d$ and $\IR/\IZ$. Recall from \ref{ssec:higher_gerbes} that a principal $\IT^d$-bundle\footnote{i.e.~a $0$-gerbe with band $\IT^d$} subordinate to a surjective submersion $\sigma:Y\rightarrow M$ is a simplicial manifold $P_\bullet$ with $P_n\cong Y^{[n+1]}\times \IT^{d\times(n+1)}$. The remaining cocycle data define a 2-gerbe $R$ with band $\sfU(1)\cong \IR/\IZ$ over the simplicial manifold $P_\bullet$ for the resolution induced by the projection $\IR^d\rightarrow \IT^d$. We have thus obtained the desired topological left $T_2$-duality background.
        
        The $2$-gerbe $R\rightarrow P$ can also be seen as a gluing of the trivial $2$-gerbes $\sigma^\ast R$ over $Y\times \mathbb T^d \cong \sigma^\ast P$ along the trivial isomorphisms of $2$-gerbes $\psi_{ij}: \sigma_i^\ast R \rightarrow \sigma_j^\ast R$ over $Y^{[2]}\times \mathbb T^d$, the $2$-isomorphisms of $2$-gerbes $\psi_{ijk}:\psi_{jk} \circ \psi_{ij} \Rightarrow \psi_{ik}$ over $Y^{[3]}\times \mathbb T^d$ given by the pullbacks from $ \mathbb T^d$ of the $\mathbb R/\IZ$-bundles with Chern class $\hat{\lambda}^0_{ijk} \in \Lambda^2 \IZ^{d,\ast} = H^2(\mathbb T^d, \IZ)$, and $3$-isomorphisms of $2$-gerbes $\psi_{ikl} \circ \psi_{ijk} \Rightarrow \psi_{jkl} \circ \psi_{ijl}$ over $Y^{[4]}\times \mathbb T^d$ defined by $(\lambda^1_{ijkl},[v^2_{ijkl}]) \in C^{\infty}(Y^{[4]}, \IZ^{d,\ast} \times \mathbb R/\IZ) \subset C^{\infty}(Y^{[4]}\times\mathbb T^d, \mathbb R/\IZ)$. This implies, by an argument similar to the one in~\cite[Lemma 2.1.6]{Nikolaus:2018qop}, which we refrain from repeating, that its Dixmier--Douady class lies in the $F_2$ step of the Serre filtration, and that every such gerbe arises from this construction.
    \end{proof}
    
    We then turn to right $T_2$-duality backgrounds. 
    \begin{proposition}\label{prop:rightbackgroundsf2}
        A principal $\sfTBRFT_d$-bundle over $M$ determines a topological right $T_2$-duality background $(M,\hat P, \hat Q, \hat R)$ of type $F_2$, and any such background arises from a principal $\sfTBRFT_d$-bundle up to equivalence.
    \end{proposition}
    \begin{proof}
        Consider cocycle data~\eqref{eq:cocycledatat2rb} for a principal $\sfTBRFT_d$-bundle, from which we construct the simplicial manifold $\hat R_\bullet$ whose successive coordinates describe the total space of $\hat P$, the 1-gerbe $\hat Q$ with band $\IT^{d,\ast}$ and the 2-gerbe $\hat R$. The simplices of $\hat R_\bullet$ read as 
        \begin{equation}\label{eq:lhat}
            \begin{aligned}
                \hat R_n &\coloneqq  \Big\{ ((\hat{v}^0_{i}),(w^1_{ij}),([w^2_{ijk}])) \in Y^{[n+1]} \times (\Lambda^2 \mathbb R^{d,\ast})^{n+1} \times (\mathbb R^{d,\ast})^{\binom{n+1}{2}} \times (\mathbb R/\IZ)^{\binom{n+1}{3}} 
                \\
                &\hspace{1cm}\Big|\: \: [\hat{v}^0_{i} - \hat{v}^0_{j}] \eqqcolon [\hat{v}^0_{ij}]~,
                \\
                &\hspace{1cm}\phantom{\Big|}\: \:[w^1_{ij} - w^1_{ik} + w^1_{jk}] = [w^1_{ijk} - \langle \lambda^0_{ijk}, \hat{v}^0_i \rangle]~,
                \\
                &\hspace{1cm}\phantom{\Big|}\:\: [w^2_{ijk} - w^2_{ijl} + w^2_{ikl} - w^2_{jkl}]= [w^2_{ijkl} + \langle \hat v^0_{ij},w^1_{jk} - w^1_{jl} + w^1_{kl} \rangle 
                \\
                &\hspace{1cm} \quad \quad \quad \quad \quad \quad \quad \quad \quad   + \langle \lambda^0_{jkl}, \langle \lambda^0_{ijl}, \hat{v}^0_i \rangle \rangle^\low - \langle \lambda^0_{ijk}, \langle \lambda^0_{ikl}, \hat{v}^0_i \rangle \rangle^\low ] \Big\}~,
            \end{aligned}
        \end{equation}
        where the equivalence classes are taken in the appropriate quotients $(\IR/\IZ)^i$, $i\in\IN$. We now compare with the description of $p$-gerbes from \ref{ssec:higher_gerbes}. The first relation is the Čech description of the principal $\Lambda^2\IT^{d,\ast}$-bundle $\hat P$. The next relation describes a $\IT^{d,\ast}$-gerbe $\hat Q$ over the simplicial total space of $\hat P$, and the last relation is the cocycle equation for the 2-gerbe $\hat R$ over $\hat Q$. Altogether, we obtain precisely a topological right $T_2$-duality background. 
        
        Note that $\hat{Q} \rightarrow \hat{P}$ is obtained by gluing trivial gerbes $\sigma^\ast \hat{Q}$ with band $\mathbb T^{d,\ast}$ over $Y\times \Lambda^2 \IT^{d,\ast} \cong \sigma^\ast \hat{P}$ along trivial isomorphisms $\psi_{ij}: \sigma_i^\ast \hat{Q} \rightarrow \sigma_j^\ast \hat{Q}$ over $Y^{[2]}\times \Lambda^2 \IT^{d,\ast}$, and $2$-isomorphisms $\psi_{ijk}:\psi_{jk} \circ \psi_{ij} \Rightarrow \psi_{ik}$ over $Y^{[3]}\times \Lambda^2 \IT^{d,\ast}$ given by the $\mathbb T^{d,\ast}$-valued functions obtained from $w^1_{ijk}$ and the pullbacks from $\Lambda^2 \IT^{d,\ast}$ of the contraction maps with $\lambda_{ijk} \in \IZ^d$. Similarly, $\hat{R} \rightarrow \hat{Q}$ is a gluing of trivial 2-gerbes, along trivial isomorphisms, 2-isomorphisms parameterized by the pull-back of $\lambda^0_{ijk} \in \IZ^{d} = H^2(B\mathbb T^{d,\ast}, \IZ)$ from the $B\mathbb T^{d,\ast}$ fibers, and $3$-isomorphisms parameterized by $[w^2_{ijkl}] \in C^{\infty}(Y^{[4]},\mathbb R/\IZ)$. Again, this implies as in~\cite[Lemma 2.1.6]{Nikolaus:2018qop} that this right background satisfies the conditions of \ref{def:typef2} to be of type $F_2$, and that every background satisfying those conditions arises from this construction.
    \end{proof}
    
    \subsection{\texorpdfstring{$T_2$}{T2}-duality as spans}
    
    We now come to the description of $T_2$-duality as spans of toroidal principal 3-bundles of the form
    \begin{equation}
        \begin{tikzcd}
            & \scP_\rmC\arrow[dl,"p"'] \arrow[dr,"\hat p"]
            \\
            \scP & & \hat \scP
        \end{tikzcd}
    \end{equation}
    Here, $\scP$ and $\hat \scP$ are principal $\sfTBLFT_d$- and $\sfTBRFT_d$-bundles, defining a left and a right $T_2$-duality background, respectively. Furthermore, $\scP_\rmC$ is a principal $\sfTDFT_d$-bundle containing the correspondence data for the topological $T_2$-duality relation, \ref{def:t2duality}, and, once all bundles are endowed with connections, for the differential $T_2$-duality relation, \ref{def:t2dualitydiff}.
    
    Let us start with topological $T_2$-duality. 
    \begin{proposition}\label{prop:lift_left}
        Let $\scP \rightarrow M$ be a principal $\sfTBLFT_d$-bundle. A lift of $\scP$ to a principal $\sfTDFT_d$-bundle $\scP_\rmC \rightarrow M$ along $\phi$ from~\eqref{eq:span_3_groups} determines a 1-gerbe $Q\rightarrow P$ with band $\mathbb T^{d,\ast}$. 
    \end{proposition}
    \begin{proof}
        If the cocycle data of $\scP$ arises from cocycle data~\eqref{eq:cocycledatat2d} for a principal $\sfTDFT_d$-bundle, then the corresponding 1-gerbe $Q$ is given by the simplicial manifold $Q_\bullet$ with
        \begin{align}
            \begin{split}
                Q_n \coloneqq  \Big\{ &((v^0_{i}),(v^1_{ij}])) \in Y^{[n+1]}\times (\mathbb R^d)^{n+1} \times (\mathbb T^{d,\ast})^{\binom{n+1}{2}} 
                \\
                &\Big|\: \: [v^0_{i} - v^0_{j}] = [v^0_{ij}]~,
                \\
                &\phantom{\Big|}\:\: [v^1_{ij} - v^1_{ik} + v^1_{jk}] = [v^1_{ijk} - \langle v^0_i,\hat{\lambda}^0_{ijk} \rangle]\Big\}~.
            \end{split}
        \end{align}
    \end{proof}
    \begin{remark}\label{rk:explainQ}
        Choose a splitting $\mathbb T^{d} = \mathbb T^1 \times \mathbb T^{d-1}$, and use the corresponding splittings $\mathbb T^{d,\ast} = \mathbb T^{\ast} \times \mathbb T^{d-1,\ast}$, $\Lambda^2 \mathbb T^{d,\ast} = \Lambda^2 \mathbb T^{d-1,\ast} \times \mathbb T^{\ast} \otimes \mathbb T^{d-1,\ast}$. The projection of $\langle v_i, \hat{\lambda}^0_{ijk} \rangle \in \mathbb T^{d,\ast}$ onto $\mathbb T^{\ast}$ is $\langle v_i^{\rm red}, \hat{\lambda}^{0,red}_{ijk} \rangle$, for $v^{\rm red}_i$ the projection of $v_i$ onto $\mathbb T^{d-1,\ast}$, and $\hat{\lambda}^{0,red}_{ijk}$ the projection of $\hat{\lambda}^0_{ijk}$ onto $\IZ^{\ast} \otimes \IZ^{d-1,\ast}$. In other words, the $\mathbb T^{\ast}$-component of the 1-gerbe $Q$ descends to the $\mathbb T^{d-1}$ bundle associated to $P \rightarrow M$, and is constructed from the gluing 2-isomorphisms of $R$ that are non-trivial along those $\mathbb T^{d-1}$ fibers. This indicates that the 1-gerbe $Q$ can be regarded as the collection of all fiber integrations of the 2-gerbe $R$, along the $d$ different torus fibers of $P \rightarrow M$; the fact that a lift to a $\sfTDFT_d$-bundle is necessary to present it is just a consequence of the non-trivial topology of $P \rightarrow M$.
    \end{remark}
    
    We can now show that $T_2$-duality can indeed be encoded as spans of toroidal principal 3-bundles.
    \begin{theorem}\label{thm:spans_are_t2_duality}
        \label{it:rightbackgrounds4} Let $\scP_\rmC\rightarrow M$ be a principal $\sfTDFT_d$-bundle, and let $\scP$ and $\hat \scP$ be the principal $\sfTBLFT_d$- and $\sfTBRFT_d$-bundles induced by the maps $\phi$ and $\hat \phi$ from~\eqref{eq:span_3_groups}. Then the left and right $T_2$-duality backgrounds defined by $\scP$ and $\hat \scP$ via \ref{prop:leftbackgroundsf2} and \ref{prop:rightbackgroundsf2} are topologically $T_2$-dual.
    \end{theorem}
    \begin{proof}
        Let $R_\bullet$ and $\hat R_\bullet$ be the two simplicial 2-gerbes obtained from $\scP$ and $\hat \scP$ through \ref{prop:leftbackgroundsf2,prop:rightbackgroundsf2}. If $\hat \scP$ is obtained from $\scP_\rmC$ through the weak homomorphism defined in~\eqref{eq:homr}, then we have
        \begin{equation}\label{eq:hatR-n}
            \begin{aligned}
                \hat R_n \coloneqq  \Big\{ &(\{\hat{v}^0_{i}\},\{w^1_{ij}\},\{[w^2_{ijk}]\}) \in Y^{[n+1]}\times (\Lambda^2 \mathbb R^{d,\ast})^{n+1} \times (\mathbb R^{d,\ast})^{\binom{n+1}{2}} \times (\mathbb R/\IZ)^{\binom{n+1}{3}} 
                \\
                &\Big| \:\: [\hat{v}^0_{i} - \hat{v}^0_{j}] = [\hat{v}^0_{ij}]~,
                \\
                &\phantom{\Big|}\: \:[w^1_{ij} - w^1_{ik} + w^1_{jk}] = [v^1_{ijk} + \langle v^0_{ij},\hat{v}^0_{jk} \rangle - \langle \lambda^0_{ijk}, \hat{v}^0_k \rangle]~,
                \\
                &\phantom{\Big|}\:\: [w^2_{ijk} - w^2_{ijl} + w^2_{ikl} - w^2_{jkl}]= [v^2_{ijkl} + \langle v^0_{ij},w^1_{jk} - w^1_{jl} + w^1_{kl}-v^1_{jkl} \rangle
                \\
                & \quad \quad \quad \quad \quad \quad \quad    + \langle \lambda^0_{jkl}, \langle \lambda^0_{ijl}, \hat{v}^0_l \rangle \rangle^\low - \langle \lambda^0_{ijk}, \langle \lambda^0_{ikl}, \hat{v}^0_l \rangle \rangle^\low 
                \\
                &  \quad \quad \quad \quad \quad \quad \quad    + \langle v^0_{jk}, \langle v^0_{ij}, \hat{v}^0_{kl} \rangle \rangle^\low + \langle \lambda^0_{ijk}, \langle v^0_{ik}, \hat{v}^0_{kl} \rangle \rangle^\low ] \Big\}~.
            \end{aligned}
        \end{equation}
        
        To show $T_2$-duality, we first need to construct an isomorphism $\psi:p^\ast R_\bullet \rightarrow \hat p^\ast\hat R_\bullet$. This isomorphism sets
        \begin{align}
            \begin{split}
                [w^2_{ijk}] &\coloneqq  [v^2_{ijk} + \langle v^0_i, v^1_{ijk} \rangle - \langle v^0_i - v^0_j, w^1_{jk} \rangle - \langle v^0_{ij}, \langle v^0_{jk}, \hat{v}^0_k \rangle \rangle^\low - \langle v^0_{ij}, \langle v^0_{j}, \hat{v}^0_j - \hat{v}^0_k - \hat{v}^0_{jk} \rangle \rangle^\low 
                \\
                & \quad \: - \langle v^0_{i}, \langle v^0_i - v^0_{j} - v^0_{ij}, \hat{v}^0_j - \hat{v}^0_k - \hat{v}^0_{jk} \rangle \rangle^\low + \langle v^0_{i}, \langle v^0_{ij}, \hat{v}^0_{jk} \rangle \rangle + \langle \lambda^0_{ijk}, \langle v^0_i, \hat{v}^0_k \rangle \rangle 
                \\
                & \quad \: - \langle \lambda^0_{ijk}, \langle v^0_{ik}, \hat{v}^0_k \rangle \rangle^\low + \langle v^0_i, \langle v^0_i, \hat{\lambda}^0_{ijk} \rangle \rangle^\low]~.
            \end{split}
        \end{align}
        To verify the fiberwise Poincar\'e condition, fix $x \in M$. A trivialization of $\scP_\rmC|_{x}$ induces trivializations of $\scP|_{P_x}$ and $\hat\scP|_{\hat Q_x}$, which means we may restrict to the case in which all the cocycle data of $\scP_\rmC|_{x}$ vanishes. In that case the isomorphism is
        \begin{equation}\label{eq:isosimplified}
            [w^2_{ijk} - v^2_{ijk}] = [\langle v^0_j - v^0_i,w^1_{jk} \rangle - \langle v^0_i, \langle v^0_j-v^0_i, \hat{v}^0_k-\hat{v}^0_j \rangle \rangle^\low]~.
        \end{equation}
        The right-hand side of~\eqref{eq:isosimplified} is cocycle data for a gerbe over $P_x \times \hat Q_x = \mathbb T^d \times \mathcal{M}(\mathbb T^d)$, while the cocycle data for the Poincar\'e gerbe~\eqref{eq:poincaredef} is, in the notation here,
        \begin{equation}\label{eq:poincarecoycle}
            [\langle v^0_j - v^0_i,w^1_{jk} \rangle - \langle v^0_j-v^0_i, \langle v^0_k-v^0_j, \hat{v}^0_k \rangle \rangle^\low]~.
        \end{equation}
        An isomorphism between~\eqref{eq:isosimplified} and~\eqref{eq:poincarecoycle} is given by $[\langle v^0_i, \langle v^0_j - v^0_i, \hat{v}^0_j \rangle \rangle^\low]$. 
    \end{proof}
    
    We then have the following immediate and important corollary.
    \begin{corollary}\label{cor:f2t2dual}
        A topological left $T_2$-duality background $(M,P,R)$ of type $F_2$ always admits a $T_2$-dual.
    \end{corollary}
    \begin{proof}
        The $T_2$-dual is constructed as follows. Write $(M,P,R)$ as a principal $\sfTBLFT_d$-bundle $\scP$ over $M$ via \ref{prop:leftbackgroundsf2}. By \ref{cor:any_bundle_from_T2D}, we can lift this bundle to a principal $\sfTDFT_d$-bundle. Finally, construct the associated principal $\sfTBRFT_d$-bundle by post-composing the Čech cocycles with the map $B\hat \phi$ induced by $\hat \phi$ from~\eqref{eq:span_3_groups}. The corresponding right $T_2$-duality background obtained via \ref{prop:rightbackgroundsf2} is then a $T_2$-dual of the original background by \ref{thm:spans_are_t2_duality}.
    \end{proof}
    
    Let us now come to the differential refinement. The required cocycle data are obtained by truncating\footnote{Concretely, set all components $g_i$, $g_{ij}$ to $\unit$.} the affine differential cocycle data of \ref{app:adjustment} to the non-affine case. A connection on a principal $\sfTDFT_d$-bundle $\scP_\rmC$ over $M$ with Čech cocycles~\eqref{eq:cocycledatat2d} then consists of the following data:
    \begin{subequations}\label{eq:cocycledatacont2d}     
        \begin{equation}
            \begin{aligned}
                (A_i,\hat{A}_i) &\in \Omega^1(Y,\mathbb R^d) \times \Omega^1(Y,\Lambda^2\mathbb R^{d,\ast})~,
                \\
                B_i &\in \Omega^2(Y,\mathbb R^{d,\ast})
                ~,~~~
                &\Lambda_{ij} &\in \Omega^1(Y^{[2]},\mathbb R^{d,\ast})~,
                \\
                C_i &\in \Omega^3(Y,\mathbb R)
                ~,~~~
                &\Sigma_{ij} &\in \Omega^2(Y^{[2]},\mathbb R)
                ~,~~~
                &\Xi_{ijk} &\in \Omega^1(Y^{[3]},\mathbb R)~,
            \end{aligned}
        \end{equation}
        satisfying the relations
        \begin{align}
            \begin{split}
                A_j &= A_i+\rmd v^0_{ij} ~,
                \\
                \hat{A}_j &= \hat{A}_i+\rmd\hat{v}^0_{ij} ~,
                \\
                B_j &= B_i +\rmd\Lambda_{ij} + \langle A_i , \rmd\hat{v}^0_{ij} \rangle_0~,
                \\
                C_{j} &=C_i + \rmd\Sigma_{ij} + \langle \rmd v^0_{ij} , B_j \rangle_1 + \langle A_i ,  \langle \rmd v^0_{ij}, \hat A_i \rangle_0 \rangle_1^\low~,
                \\
                \Lambda_{ij} - \Lambda_{ik} + \Lambda_{jk} &= \rmd v^1_{ijk} + \langle \rmd v^0_{ij}, \hat{v}^0_{jk} \rangle_0~,
                \\
                \Sigma_{ij} - \Sigma_{ik} + \Sigma_{jk} &= \rmd\Xi_{ijk} - \langle \rmd v^0_{ij} , \Lambda_{jk} \rangle_1 + \langle \rmd v^0_{ij} , \langle \rmd v^0_{ij}, \hat{v}^0_{jk} \rangle_0 \rangle_1^\low~,
                \\
                \Xi_{ijk} - \Xi_{ijl} + \Xi_{ikl} - \Xi_{jkl} &= \rmd v^2_{ijkl} - \langle \rmd v^0_{ij}, v^1_{jkl} \rangle_1 + \langle v^0_{ij}, \langle \rmd v^0_{ij}, \hat{\lambda}^0_{jkl} \rangle_0 \rangle_1^\low
            \end{split}
        \end{align}
        for all $ij\in Y^{[2]}$, $ijk\in Y^{[3]}$, and $ijkl\in Y^{[4]}$.
    \end{subequations}
    The curvature decomposes into the forms
    \begin{equation}
        \begin{aligned}
            (F,\hat F) &\coloneqq  (\rmd A_i,\rmd \hat A_i)~~~\in~~\Omega^2(M,\mathbb R^d\times \Lambda^2\mathbb R^{d,\ast})~,
            \\
            H_0 &\coloneqq  \rmd B_i - \langle F, \hat{A}_i \rangle_0~~~\in~~\Omega^3(M,\mathbb R^{d,\ast})~,
            \\
            G_0 &\coloneqq  \rmd C_i + \langle A_i , H_0 \rangle_1+\langle A_i,\langle A_i,\hat F\rangle_0\rangle_1^\low~~~\in~~\Omega^4(M,\mathbb R)~,
        \end{aligned}
    \end{equation}
    and the Bianchi identities read as
    \begin{align}
        \begin{split}
            \rmd F = \rmd \hat{F} &= 0~,\\
            \rmd H_0 + \langle F , \hat{F} \rangle_0 &= 0~,\\
            \rmd G_0 - \langle F , H_0 \rangle_1 &= 0~.
        \end{split}
    \end{align}
    
    We then have the following.
    \begin{proposition}\label{prop:connection_leftbackgroundsf2}
        Let $\scP\rightarrow M$ be a principal $\sfTBLFT_d$-bundle lifted to a principal $\sfTDFT_d$-bundle $\scP_\rmC \rightarrow M$ as in \ref{prop:lift_left}. Let $(P,R)$ be the left $T_2$-duality background described by $\scP$ by \ref{prop:leftbackgroundsf2} and let $Q$ be the $\IT^{d,\ast}$-gerbe determined by the lift as in \ref{prop:lift_left}.
        A connection on $\scP_\rmC$ of curvature $(F,\hat{F},H_0,G_0)$ induces the following connection data:
        \begin{myenumerate}
            \item a connection $A$ on $P \rightarrow M$ of curvature $F$;
            \item a connection $(B,\Lambda)$ on $Q \rightarrow P$ of curvature $H_0 + \langle A , \hat{F} \rangle_0$;
            \item a connection $(C,\Sigma,\Xi)$ on $R \rightarrow P$ of curvature $G_0 - \langle A , H_0 \rangle_1 - \langle A , \langle A , \hat{F} \rangle_0 \rangle_1^\low$.
        \end{myenumerate}
    \end{proposition}
    \begin{proof}
        Let $\scP_\rmC$ be given by the cocycle data~\eqref{eq:cocycledatacont2d}. Then the connection forms read as 
        \begin{align}\label{eq:conn_1}
            \begin{split}
                A &\coloneqq  \rmd v^0_i + A_i~,
                \\
                B &\coloneqq  B_i - \langle A_i+\rmd v^0_i , \hat{A}_i \rangle~,
                \\
                C &\coloneqq  C_i + \langle \rmd v^0_i +A_i, B_i \rangle - \langle \rmd v^0_i + A_i , \langle \rmd v^0_i + A_i , \hat{A}_i \rangle \rangle^\low~,
                \\
                \Lambda &\coloneqq  \rmd v^1_{ij} - \Lambda_{ij} - \langle \rmd v^0_i, \hat{v}^0_{ij} \rangle~,
                \\
                \Sigma &\coloneqq  -\Sigma_{ij} + \langle \rmd v^0_i , \Lambda_{ij} \rangle + \langle \rmd v^0_i , \langle \rmd v^0_i , \hat{v}^0_{ij} \rangle \rangle^\low~,
                \\
                \Xi &\coloneqq  \rmd v^2_{ijk} - \Xi_{ijk} - \langle \rmd v^0_i, v^1_{ijk} \rangle - \langle v^0_i, \langle \rmd v^0_i, \hat{\lambda}^0_{ijk} \rangle \rangle^\low~.                
            \end{split}
        \end{align}
    \end{proof}
    
    The corresponding statement for right $T_2$-duality backgrounds is the following.
    \begin{proposition}\label{prop:connection_rightbackgroundsf2}
        Let $\hat \scP \rightarrow M$ be a principal $\sfTBRFT_d$-bundle determining a right $T_2$-duality background $(\hat P,\hat Q,\hat R)$ as in \ref{prop:rightbackgroundsf2}. A lift of $\hat \scP$ to a principal $\sfTDFT_d$-bundle $\scP_\rmC\rightarrow M$ along~$\hat \phi$ with a connection of curvature $(F,\hat{F},H_0,G_0)$ determines the following data:
        \begin{myenumerate}
            \item a connection $\hat{A}$ on $\hat{P} \rightarrow M$ of curvature $\hat{F}$;
            \item a connection $(\hat{B},\hat{\Lambda})$ on $\hat Q \rightarrow \hat{P}$ of curvature $H_0 + \langle F , \hat{A} \rangle_0$;
            \item a connection $(\hat{C},\hat{\Sigma},\hat{\Xi})$ on $\hat R\rightarrow \hat Q$ of curvature $(\langle F , \hat{\Lambda} \rangle ,G_0 - \langle F , \hat{B} \rangle)$.
        \end{myenumerate}
    \end{proposition}
    \begin{proof}
        Let $\hat R_\bullet$ be the simplicial manifold~\eqref{eq:hatR-n}. Then the connections are described as follows.
        \begin{align}\label{eq:conn_2}
            \begin{split}
                \hat{A} &\coloneqq  \rmd \hat v^0_i + \hat{A}_i~,
                \\
                \hat{B} &\coloneqq  B_i + \langle A_i , \rmd \hat{v}^0_i \rangle~,
                \\
                \hat{C} &\coloneqq  C_i - \langle A_i , B_i \rangle~,
                \\
                \hat{\Lambda} &\coloneqq  \rmd w^1_{ij} - \Lambda_{ij} - \langle v^0_{ij},\rmd \hat{v}^0_{j} \rangle~,
                \\
                \hat{\Sigma} &\coloneqq  -\Sigma_{ij} + \langle A_i , \rmd w^1_{ij}- \Lambda_{ij} - 
                \langle v^0_{ij},\rmd \hat{v}^0_j \rangle \rangle + \langle v^0_{ij}, \langle \rmd v^0_{ij} ,  \rmd \hat{v}^0_{j} \rangle \rangle^\low~,
                \\
                \hat{\Xi} &\coloneqq  \rmd w^2_{ijk} - \Xi_{ijk} - \langle \rmd v^0_{ij}, w^1_{jk} \rangle - \langle v^0_{ij}, \langle \rmd v^0_{ij},\hat{v}^0_{jk} \rangle \rangle^\low \\
                &\quad - \langle v^0_{ij}, \langle v^0_{jk}, \rmd \hat{v}^0_{k} \rangle \rangle^\low - \langle \lambda^0_{ijk}, \langle v^0_{ik}, \rmd \hat{v}^0_{k} \rangle \rangle^\low~.
            \end{split}
        \end{align}
    \end{proof}
    
    We can now establish differential $T_2$-duality.
    \begin{theorem}\label{thm:spans_are_differential_t2_duality}
        Let $\scP_\rmC\rightarrow M$ be a principal $\sfTDFT_d$-bundle defining a $T_2$-dual pair of topological left and right $T_2$-duality backgrounds by \ref{thm:spans_are_t2_duality}. Choose a connection on $\scP_\rmC$. Then the connections induced on the left and right $T_2$-duality backgrounds by \ref{prop:connection_leftbackgroundsf2} and \ref{prop:connection_rightbackgroundsf2} are $T_2$-dual.
    \end{theorem}
    \begin{proof}
        Let the connections determined by the connection on $\scP_\rmC$ be 
        \begin{equation}
            (A,(B,\Lambda),(C,\Sigma,\Xi))
            \eand
            (\hat{A},(\hat{B},\hat{\Lambda}),(\hat{C},\hat{\Sigma},\hat{\Xi}))~,
        \end{equation}
        as before. We can then equip the isomorphism $\psi:p^\ast R_\bullet \rightarrow \hat p^\ast\hat R_\bullet$ from \ref{thm:spans_are_t2_duality} with a connection given by the following differential forms.
        \begin{align}
            \begin{split}
                B^{\psi}_i &\coloneqq  - \langle \rmd v^0_i , \langle \rmd v^0_i, \hat{v}^0_i \rangle \rangle^\low~,
                \\
                \Lambda^{\psi}_{ij} &\coloneqq  \langle v^0_i, \rmd w^1_{ij} \rangle - \langle v^0_{ij}, \langle \rmd v^0_j, \hat{v}^0_j \rangle \rangle^\low + \langle \rmd v^0_i, \langle v^0_{ij}, \hat{v}^0_j \rangle \rangle^\low + \langle \rmd v^0_i, \langle v^0_{i}, \hat{v}^0_i - \hat{v}^0_j - \hat{v}^0_{ij} \rangle \rangle^\low~.
            \end{split}
        \end{align}
        These satisfy the required identities
        \begin{equation}
            \begin{aligned}
                (\delta \Lambda^\psi)_{ijk}&=
                \Lambda_{ij}^{\psi} - \Lambda_{ik}^{\psi} + \Lambda_{jk}^{\psi} = \Xi_{ijk} - \psi^{\ast}\hat{\Xi}_{ijk} ~,
                \\ 
                (\delta B^\psi)_{ij}-\rmd \Lambda^\psi_{ij}&=B_i^{\psi} - B_j^{\psi} -\rmd \Lambda_{ij}^{\psi}=\hat{\Sigma}_{ij}-\Sigma_{ij} -\langle A_i , \hat{\Lambda}_{ij} \rangle~,
                \\
                \rmd B_i^\psi&=- C_i + \hat{C}_i+
                \langle A_i , \hat{B}_i \rangle~.
            \end{aligned}
        \end{equation}
        for all $i\in Y$, $ij\in Y^{[2]}$, and $ijk\in Y^{[3]}$.
    \end{proof}
    
    Finally, we show that differential $T_2$-duals of left $T_2$-duality backgrounds exist.
    \begin{corollary}\label{cor:difff2t2dual}
        Every topological left $T_2$-duality background $(M,P,R)$ of type $F_2$ admits a choice of connection for which it has a differential $T_2$-dual. 
    \end{corollary}
    \begin{proof}
        The topological $T_2$-duality background $(M,P,R)$ admits a $T_2$-dual by \ref{cor:f2t2dual}. It also admits connection data $(A,C,\Sigma,\Xi)$ such that a differential $T_2$-dual exists. These connections evidently comprise those obtained from a choice of connection on the principal $\sfTDFT_d$-bundle.
    \end{proof}

    \subsection{Example: \texorpdfstring{$T_2$}{T2}-duality of type \texorpdfstring{$F_2$}{F2} and \texorpdfstring{$d=2$}{d=2} over \texorpdfstring{$S^2$}{S2}}\label{ssec:exs2}
    
    We illustrate the preceding constructions for $d=2$ over $M=S^2$. Write $R_{c_1,c_2,c_3} \rightarrow P_{c_1,c_2}$ for a $2$-gerbe with Dixmier--Douady class $c_3 \in \IZ \cong H^4(P_{c_1,c_2},\IZ)$ over a principal $\IT^2$-bundle with Chern class $(c_1,c_2) \in \IZ^2 \cong H^2(S^2,\IZ^2)$. Because the base space has dimension two, $c_3$ has to lie in the $F_2$ step of the filtration of 
    $H^4(P_{c_1,c_2},\IZ)$ associated to the fibration $P_{c_1,c_2} \rightarrow S^2$. Therefore, this background admits a $T_2$-dual by \ref{cor:f2t2dual}.
    
    We first claim that any $T_2$-dual of $(P_{c_1,c_2},R_{c_1,c_2,c_3})$ must be of the form 
    \begin{equation}
        (\hat{P}_{c_3},\hat Q_{(-c_2,c_1)},\hat R)~,
    \end{equation}
    where $\hat{P}_{c_3} \rightarrow S^2$ is a principal $\Lambda^2 \IT^{2,\ast}$-bundle with Chern class $c_3 \in \IZ \cong \Lambda^2 \IZ^{2,\ast} = H^2(S^2,\Lambda^2 \IZ^{2,\ast})$ and $\hat Q_{(-c_2,c_1)} \rightarrow \hat{P}_{c_3}$ is a 1-gerbe with band $\IT^{2,\ast}$ and Dixmier--Douady class $(-c_2,c_1) \in \IZ^2 \cong H^3(\hat{P}_{c_3},\IZ^{2,\ast})$. 
    In these expressions, we are writing $(e_1,e_2)$ with $e_1, \,e_2 \in \IZ$ for points in $\IZ^2$, and identifying $\IZ^{2,\ast}$ with $\IZ^2$ using the standard inner product, and $\IZ \cong \Lambda^2 \IZ^{2,\ast}$ as $x \mapsto x \rmd e_1 \wedge \rmd e_2$. Then the claim follows from \eqref{eq:momentumexchange}, writing the map $\delta$ in these coordinates.
    
    To describe $\hat R$, we give cocycle data for a principal $\sfTDFT_d$-bundle with connection over $S^2$ whose corresponding left background is $(P_{c_1,c_2},R_{c_1,c_2,c_3})$. We present $S^2$ as a family of circles parameterized by a height coordinate:
    \begin{equation}
        S^2 = (\mathbb R/\IZ \times [0,1])/\sim~,
    \end{equation}
    where the circles at $t=0$ and $t=1$ are collapsed to the two poles, 
    \begin{equation}
        ([u_1],0) \sim ([u_2],0)
        \eand
        ([u_1],1) \sim ([u_2],1)
    \end{equation}
    for all $[u_1], \,[u_2] \in \mathbb R/\IZ$. We define an open cover of $S^2$ by the following three sets:
    \begin{align}
        U_1 &\coloneqq  \{[[u],t] \in S^2 \:|\: t \neq 1\}~,
        \\
        U_{0+} &\coloneqq  \{[[u],t] \in S^2 \:|\:[u] \neq [0], \: t \neq 0\}~,
        \\
        U_{0-} &\coloneqq  \{[[u],t] \in S^2 \:|\:[u] \neq [\tfrac12], \: t \neq 0\}~.
    \end{align}
    Note that there are functions $u_+:U_{0+} \cap U_{1} \rightarrow (0,1)$ and $u_{-}:U_{0-} \cap U_{1} \rightarrow (-\tfrac12,\tfrac12)$ such that
    \begin{align}
        u_+([u],t) &= [u] \quad \text{ mod } \IZ~,
        \\
        u_-([u],t) &= [u] \quad \text{ mod } \IZ~,
        \\
        u_+([u],t)-u_-([u],t) &= \begin{cases} 1 & \text{ if } [u] \in (0,\tfrac12)~, \\
            0 & \text{ if } [u] \in (\tfrac12,1)~.
        \end{cases}
    \end{align}
    The principal $\sfTDFT_d$-bundle with connection is described by the following cocycle data. The Čech data reads as
    \begin{equation}
        \begin{aligned}
            v^0_{0+1}([[u],t]) &\coloneqq  (c_1u_+,c_2u_+)~, \quad \quad \hat{v}^0_{0+1} \coloneqq  c_3u_+~,
            \\
            v^0_{0-1}([[u],t]) &\coloneqq  (c_1u_-,c_2u_-)~, \quad \quad \hat{v}^0_{0-1} \coloneqq  c_3u_-~,
            \\
            v^0_{0+0-}([[u],t]) &\coloneqq  (0,0)~, \quad \quad \quad \quad \:\:\: 
            \hat{v}^0_{0+0-} \coloneqq  0~,
            \\
            \lambda^0_{0+0-1}([[u],t]) &\coloneqq  \begin{cases} 
                (c_1,c_2) & u_+ \in (0,\tfrac12)~,
                \\
                (0,0)& u_+ \in (\tfrac12,1)~,
            \end{cases}
            \\
            \hat{\lambda}^0_{0+0-1}([[u],t]) &\coloneqq  \begin{cases} 
                c_3 & u_+ \in (0,\tfrac12)~,
                \\
                0 & u_+ \in (\tfrac12,1)~,
            \end{cases}
            \\
            v^1_{0+0-1} &\coloneqq  0~.
        \end{aligned}
    \end{equation}
    The differential refinement is given by 
    \begin{equation}
        \begin{aligned}
            A_{0+}([[u],t]) &\coloneqq  (c_1(t-1)\rmd u,c_2(t-1)\rmd u)~, \quad \quad \hat{A}_{0+}([[u],t]) \coloneqq  c_3(t-1)\rmd u~,
            \\
            A_{0-}([[u],t]) &\coloneqq  (c_1(t-1)\rmd u,c_2(t-1)\rmd u)~, \quad \quad \hat{A}_{0-}([[u],t]) \coloneqq  c_3(t-1)\rmd u~,
            \\
            A_{1}([[u],t]) &\coloneqq  (c_1t\rmd u,c_2t\rmd u), \quad \quad \quad \quad\quad \:\:\: \quad \quad \hat{A}_{1}([[u],t]) \coloneqq  c_3t\rmd u~,
        \end{aligned}
    \end{equation}
    with all other connection forms equal to $0$. The proofs of \ref{prop:connection_leftbackgroundsf2} and \ref{prop:connection_rightbackgroundsf2} then yield explicit formulas for the remaining differential cocycle data.
    
    \section{Affine \texorpdfstring{$T_2$}{T2}-duality}\label{sec:affine}
    
    We now extend $T_2$-duality to backgrounds involving affine $\IT^d$-bundles, which are principal bundles with structure group $\sfGL(d,\IZ)\ltimes \IT^d$, cf.~\cite{Baraglia:1105.0290}. Correspondingly, we have to generalize the toroidal 3-groups introduced in \ref{sec:toroidal-3-group-examples} by incorporating the extra $\sfGL(d,\IZ)$-factor.
    
    \subsection{Toroidal 3-groups for affine \texorpdfstring{$T_2$}{T2}-duality}
    
    We start by defining the relevant toroidal 3-groups. We note that the standard action on $\mathbb R^d$ induces the following action on $\Lambda^j \mathbb R^{d,\ast}$, $j\in \IN$:
    \begin{equation}
        (g \omega)(-,\ldots,-) \coloneqq  \omega(g^{-1}-,\ldots,g^{-1}-)
    \end{equation}
    for all $g\in \sfGL(d,\IZ)$ and $\omega\in \Lambda^j \mathbb R^{d,\ast}$. In particular, the pairings $\langle -,-\rangle_0$ and $\langle -,-\rangle_1$ from~\eqref{eq:pairingst2dnf2} are equivariant and invariant, respectively:
    \begin{equation}
        g \langle -, -\rangle_0= \langle g(-), g(-) \rangle_0
        \eand
        \langle -,- \rangle_1= \langle g(-), g(-)\rangle_1 
    \end{equation}
    for all $g \in \sfGL(d,\IZ)$. We stress that these properties do not directly translate to the combination $\langle -,\langle-,-\rangle_0\rangle_1^\low$.
    
    Affine left $T_2$-duality backgrounds will be encoded by principal 3-bundles with the following structure 3-group.
    \begin{definition}[Toroidal 3-group $\sfTBLFTA_d$]
        The toroidal 3-group $\sfTBLFTA_d$ has underlying complex
        \begin{equation}
            \begin{tikzcd}[column sep=0.8cm,row sep=0.2cm]
                \IZ^{d,\ast} \times \mathbb R/\IZ \arrow[r,"\sft"] 
                & \IZ^d \times \Lambda^2 \IZ^{d,\ast} \arrow[r,"\sft"] 
                & \sfGL(d,\IZ) \ltimes \mathbb R^d 
                \\
                (\lambda^1,[v^2]) \arrow[r,mapsto] & (0,0) ~,
                \\
                &(\lambda^0,\hat \lambda^0) \arrow[r,mapsto] & (\unit,\lambda^0)~,
            \end{tikzcd}
        \end{equation}
        and the structure maps
        \begin{equation}
            \begin{aligned}
                (g,v^0) \triangleright (\lambda^0,\hat \lambda^0) &= (g\lambda^0,g\hat \lambda^0)~,
                \\
                (g,v^0) \triangleright (\lambda^1,[v^2]) &= (g\lambda^1, [v^2 + \iota_v^0(g\lambda^1)])~,
                \\
                (\lambda^0,\hat \lambda^0) \triangleright (\lambda^1,[v^2]) &= (\lambda^1,[v^2])~,
                \\
                \{(\lambda^0_1,\hat \lambda^0_1),(\lambda^0_2,\hat \lambda^0_2)\} &= (-\iota_{\lambda^0_1}\hat \lambda^0_2,[0])~,
                \\
                \alpha((g_1,v^0_1),(g_2,v^0_2),(\lambda^0,\hat \lambda^0)) &= (0,[-\hat \lambda^{0,\low}(g_2^{-1}v^0_2,g_2^{-1}g_1^{-1}v^0_1)])~,
                \\
                \beta((g,v^0),(\lambda^0_1,\hat \lambda^0_1),(\lambda^0_2,\hat \lambda^0_2)) &= (0,[0])~.
            \end{aligned}
        \end{equation}
    \end{definition}
    A principal $\sfTBLFTA_d$-bundle over $M$ subordinate to $\sigma:Y\rightarrow M$ has Čech cocycle data
    \begin{subequations}\label{eq:cocycledatat2baff}
        \begin{align}
            \begin{split}
                (g_{ij},v^0_{ij}): Y^{[2]} &\rightarrow \sfGL(d,\IZ) \times \mathbb R^d~,
                \\
                (\lambda^0_{ijk},\hat{\lambda}^0_{ijk}): Y^{[3]}&\rightarrow \IZ^d \times \Lambda^2\IZ^{d,\ast}~,
                \\
                (\lambda^1_{ijkl},[v^2_{ijkl}]): Y^{[4]}&\rightarrow \IZ^{d,\ast} \times \mathbb R/\IZ
            \end{split}
        \end{align}
        satisfying
        \begin{equation}
            \begin{aligned}
                g_{ij}g_{jk} &= g_{ik}~,
                \\
                v^0_{ij} - v^0_{ik} + g_{ij}v^0_{jk} &= -\lambda^0_{ijk}~,
                \\
                \lambda^0_{ijk}-\lambda^0_{ijl}+\lambda^0_{ikl}-g_{ij}\lambda^0_{jkl}&=0~,
                \\
                \hat{\lambda}^0_{ijk} - \hat{\lambda}^0_{ijl} + \hat{\lambda}^0_{ikl} - g_{ij}\hat{\lambda}^0_{jkl} &= 0~,
                \\
                \lambda^1_{ijkl} - \lambda^1_{ijkm} + \lambda^1_{ijlm} - \lambda^1_{iklm} + g_{ij}\lambda^1_{jklm} &= - \langle \lambda^0_{ijk}, g_{ik}\hat{\lambda}^{0}_{klm} \rangle_0~,
                \\
                [v^2_{ijkl} - v^2_{ijkm} + v^2_{ijlm} - v^2_{iklm} + v^2_{jklm}] &
                \\
                =[-\langle v^0_{ij}, g_{ij}\lambda^1_{jklm} \rangle_1 + \langle \lambda^0_{ijk}, \langle v^0_{ik}, &g_{ik}\hat{\lambda}^0_{klm} \rangle_0 \rangle_1^\low + \langle v^0_{ij}, \langle g_{ij}v^0_{jk}, g_{ik}\hat{\lambda}^0_{klm} \rangle_0 \rangle_1^\low]~.
            \end{aligned}
        \end{equation}
    \end{subequations}    
    
    Affine right $T_2$-duality backgrounds will be encoded in principal 3-bundles with the following structure 3-group.
    \begin{definition}[Toroidal 3-group $\sfTBRFTA_d$]
        The toroidal 3-group $\sfTBRFTA_d$ has underlying complex
        \begin{equation}
            \begin{tikzcd}[column sep=0.8cm,row sep=0.2cm]
                \mathbb R/\IZ \arrow[r,"\sft"] 
                & \IZ^d \times \Lambda^2 \IZ^{d,\ast} \times \IT^{d,\ast}  \arrow[r,"\sft"] 
                & \sfGL(d,\IZ) \ltimes  \Lambda^2\mathbb R^{d,\ast}
                \\
                {[w^2]} \arrow[r,mapsto] & (0,0,[0]) ~,
                \\
                &(\lambda^0,\hat \lambda^0,w^1) \arrow[r,mapsto] & (\unit,\hat \lambda^0)~,
            \end{tikzcd}
        \end{equation}
        and the structure maps
        \begin{equation}
            \begin{aligned}
                (g,\hat v^0) \triangleright (\lambda^0,\hat \lambda^0,[w^1]) &= (g\lambda^0,g\hat \lambda^0,[gw^1+\iota_{g\lambda^0}\hat v^0])~,
                \\
                (g,\hat v^0) \triangleright [w^2] &= [w^2]~,
                \\
                (\lambda^0,\hat \lambda^0,[w^1]) \triangleright [w^2] &= [w^2]~,
                \\
                \{(\lambda^0_1,\hat \lambda^0_1,[w^1_1]),(\lambda^0_2,\hat \lambda^0_2,[w^1_2])\} &= [\iota_{\lambda^0_1}w^1_2]~,
                \\
                \alpha((g_1,\hat v^0_1),(g_2,\hat v^0_2),(\lambda^0,\hat \lambda^0,[w^1])) &= [0]~,
                \\
                \beta((g,\hat v^0),(\lambda^0_1,\hat \lambda^0_1,[w^1_1]),(\lambda^0_2,\hat \lambda^0_2,[w^1_2])) &= [\hat v^{0,\low}(g\lambda^0_1,g\lambda^0_2)]~.
            \end{aligned}
        \end{equation}
    \end{definition}
    
    A principal $\sfTBRFTA_d$-bundle, again over a manifold $M$ and subordinate to $\sigma:Y\rightarrow M$, has Čech cocycle data 
    \begin{subequations}\label{eq:cocycledatat2rbaff}
        \begin{align}
            \begin{split}
                (g_{ij},\hat{v}^0_{ij}): Y^{[2]} &\rightarrow \sfGL(d,\IZ) \times \Lambda^2 \mathbb R^{d,\ast}~,
                \\
                (\lambda^0_{ijk},\hat{\lambda}^0_{ijk},[w^1_{ijk}]): Y^{[3]}&\rightarrow \IZ^d \times \Lambda^2\IZ^{d,\ast} \times \mathbb T^{d,\ast},\\
                [w^2_{ijkl}]: Y^{[4]} &\rightarrow  \mathbb R/\IZ
            \end{split}
        \end{align}
        satisfying 
        \begin{equation}
            \begin{aligned}
                g_{ij}g_{jk} &= g_{ik}~,
                \\
                \hat{v}^0_{ij} - \hat{v}^0_{ik} + g_{ij}\hat{v}^0_{jk} &= -\hat \lambda^0_{ijk}~,
                \\
                \lambda^0_{ijk} - \lambda^0_{ijl} + \lambda^0_{ikl} - g_{ij}\lambda^0_{jkl} &= 0~,
                \\
                \hat\lambda^0_{ijk}-\hat\lambda^0_{ijl}+\hat\lambda^0_{ikl}-g_{ij}\hat\lambda^0_{jkl}&=0~,
                \\
                [w^1_{ijk} - w^1_{ijl} + w^1_{ikl} - g_{ij}w^1_{jkl}] &= [\langle g_{ij}\lambda^0_{jkl}, \hat{v}^0_{ij} \rangle]~,
                \\
                [w^2_{ijkl} - w^2_{ijkm} + w^2_{ijlm} - w^2_{iklm} + w^2_{jklm}] & 
                \\
                =[\langle g_{ik}\lambda^0_{klm}, \langle g_{ij} \lambda^0_{jkm}, \hat{v}^0_{ij} \rangle \rangle^\low &- \langle g_{ij}\lambda^0_{jkl}, \langle g_{ij}\lambda^0_{jlm}, \hat{v}^0_{ij} \rangle\rangle^\low 
                \\
                &+ \langle \lambda^0_{ijk}, g_{ik}w^1_{klm} + \langle g_{ik}\lambda^0_{klm}, \hat{v}^0_{ik} \rangle \rangle ]~.
            \end{aligned}
        \end{equation}
    \end{subequations}
    
    The expected span of toroidal 3-groups is completed by the following 3-group.
    \begin{definition}[Toroidal 3-group $\sfTDFTA_d$]
        The toroidal 3-group $\sfTDFTA_d$ has underlying complex
        \begin{equation}
            \begin{tikzcd}[column sep=0.8cm,row sep=0.2cm]
                \IZ^{d,\ast} \times \mathbb R/\IZ \arrow[r,"\sft"] 
                & \IZ^d \times \Lambda^2 \IZ^{d,\ast} \times \mathbb R^{d,\ast}  \arrow[r,"\sft"] 
                & \sfGL(d,\IZ) \ltimes (\mathbb R^d \times \Lambda^2\mathbb R^{d,\ast}) 
                \\
                (\lambda^1,[v^2]) \arrow[r,mapsto] & (0,0,\lambda^1) ~,
                \\
                &(\lambda^0,\hat \lambda^0,v^1) \arrow[r,mapsto] & (\unit,\lambda^0,\hat \lambda^0)~,
            \end{tikzcd}
        \end{equation}
        and structure maps
        \begin{equation}
            \begin{aligned}
                (g,v^0,\hat v^0) \triangleright (\lambda^0,\hat \lambda^0,v^1) &= (g\lambda^0,g\hat \lambda^0,gv^1 + \iota_v^0(g\hat \lambda^0))~,
                \\
                (g,v^0,\hat v^0) \triangleright (\lambda^1,[v^2]) &= (g\lambda^1, [v^2 + \iota_v^0(g\lambda^1)])~,
                \\
                (\lambda^0,\hat \lambda^0,v^1) \triangleright (\lambda^1,[v^2]) &= (\lambda^1,[v^2])~,
                \\
                \{(\lambda^0_1,\hat \lambda^0_1,v^1_1),(\lambda^0_2,\hat \lambda^0_2,v^1_2)\} &= (-\iota_{\lambda^0_1}\hat \lambda^0_2,[0])~,
                \\
                \alpha((g_1,v^0_1,\hat v^0_1),(g_2,v^0_2,\hat v^0_2),(\lambda^0,\hat \lambda^0,v^1)) &= (0,[-\hat \lambda^{0,\low}(g_2^{-1}v^0_2,g_2^{-1}g_1^{-1}v^0_1)])~,
                \\
                \beta((g,v^0,\hat v^0),(\lambda^0_1,\hat \lambda^0_1,v^1_1),(\lambda^0_2,\hat \lambda^0_2,v^1_2)) &= (0,[0])~.
            \end{aligned}
        \end{equation}
    \end{definition}
    The full differential cohomology describing principal $\sfTDFTA_2$-bundles with connections is found in \ref{app:adjustment}.
    
    The three toroidal 3-groups defined for the affine setting can again be arranged into a span.
    \begin{proposition}\label{prop:affine}
        There is a commutative diagram of homomorphisms of toroidal 3-groups
        \begin{equation}\label{eq:TTgroups_diagram}
            \begin{tikzcd}
                \sfTBLFT_d \ar[d] & \ar[l] \sfTDFT_d \ar[r]  \ar[d] &  \sfTBRFT_d  \ar[d] \\
                \sfTBLFTA_d & \ar[l] \sfTDFTA_d \ar[r] &  \sfTBRFTA_d
            \end{tikzcd}
        \end{equation}
        such that the horizontal arrows are weak homotopy equivalences.
    \end{proposition}
    \begin{proof}
        The homomorphism $\sfTDFTA_d\rightarrow \sfTBLFTA_d $ is the strict projection that forgets the $\Lambda^2\IR^{d,\ast}$- and $\IR^{d,\ast}$-components. The homomorphism $\sfTDFTA_d \rightarrow \sfTBRFTA_d$ is a weak homomorphism of toroidal 3-groups, which is an extension of that from \ref{prop:span_3_groups} by a suitable inclusion of the group $\sfGL(d,\IZ)$:
        \begin{align}
            \begin{split}\label{eq:homraff}
                [w^1_{ijk}] &\coloneqq  [v^1_{ijk} + \langle v^0_{ij},g_{ij}\hat{v}^0_{jk} \rangle + \langle \lambda^0_{ijk},v^0_{ik} \rangle]~,
                \\
                [w^2_{ijkl}] &\coloneqq  [v^2_{ijkl} - \langle v^0_{ij}, g_{ij}v^1_{jkl} \rangle + \langle g_{ij} v^0_{jk}, \langle v^0_{ij},g_{ik} \hat{v}^0_{kl} \rangle \rangle^\low + \langle \lambda^0_{ijk}, \langle v^0_{ik}, g_{ik} \hat{v}^0_{kl} \rangle \rangle^\low  \\
                &\hspace{1cm} + \langle \lambda^0_{ijk}, \langle \lambda^0_{ikl}, \hat{v}^0_{il} \rangle \rangle^\low -  \langle g_{ij}\lambda^0_{jkl}, \langle \lambda^0_{ijl}, \hat{v}^0_{il} \rangle \rangle^\low ]~.
            \end{split}
        \end{align}
        The embedding maps are the evident ones, setting the $\sfGL(d,\IZ)$-factor to $\unit$. Commutativity follows from putting all $g_{\ldots}=\unit$.
    \end{proof}
    
    \subsection{Affine \texorpdfstring{$T_2$}{T2}-duality as spans}
    
    For a geometric interpretation of toroidal principal 3-bundles with the structure groups defined above, we need to slightly extend our definition of $p$-gerbes.
    \begin{definition}[Affine gerbe]
        Let $\sfG=\sfGL(d,\IZ)$ act on $\sfH=\IT^{d,\ast}$, by the standard dual action. Together with the trivial boundary homomorphism $\sft:\sfH\rightarrow \sfG$, this defines a crossed module of Lie groups 
        \begin{equation}
            \caG=(\sfH \xrightarrow{~\sft~}\sfG,\acton)~.
        \end{equation}
        An \emph{affine 1-gerbe with band $\IT^{d,\ast}$} is then a principal $\caG$-bundle, see~\cite{Rist:2022hci} for the cocycles in our conventions.
    \end{definition}
    
    We then have the following statements, linking toroidal principal 3-bundles to affine $T_2$-duality backgrounds.
    \begin{proposition}\label{prop:left-affine}
        A principal $\sfTBLFTA_d$-bundle over a manifold $M$ determines the following data:
        \begin{myenumerate}
            \item an affine $\mathbb T^d$-bundle $P \rightarrow M$;
            \item a $2$-gerbe $R \rightarrow P$ with band $\sfU(1)$.
        \end{myenumerate}
    \end{proposition}
    \begin{proof}
        Consider cocycle data~\eqref{eq:cocycledatat2baff} for a principal $\sfTBLFTA_d$-bundle. The corresponding affine left $T_2$-duality background $(P,R)$ is given by the simplicial manifold $R_\bullet$ with
        \begin{align}
            \begin{split}
                R_n \coloneqq  \Big\{ &((v^0_{i}),([v^2_{ijk}])) \in Y^{[n+1]}\times (\mathbb R^d)^{n+1} \times (\mathbb R/\IZ)^{\binom{n+1}{3}} 
                \\
                &\Big|\: \: [v^0_{i} - g_{ij}v^0_{j}] = [v^0_{ij}]~,
                \\
                &\phantom{\Big|}\:\: [v^2_{ijk} - v^2_{ijl} + v^2_{ikl} - v^2_{jkl}]
                \\
                &\hspace{1cm}= [v^2_{ijkl} - \langle v^0_i,\lambda^1_{ijkl} \rangle + \langle v^0_i, \langle v^0_{ij}, g_{ij} \hat{\lambda}^0_{jkl} \rangle \rangle^\low 
                \\
                &\hspace{2cm}+ \langle v^0_i - g_{ij}v^0_j - v^0_{ij}, \langle g_{ij}v^0_j,g_{ij}\hat{\lambda}^0_{jkl} \rangle \rangle^\low ] \Big\}~.
            \end{split}
        \end{align}
        This is precisely a 2-gerbe $R$ on the resolved total space of an affine $\IT^d$-bundle $P$.
    \end{proof}
    
    \begin{proposition}\label{prop:right-affine}
        A principal $\sfTBRFTA_d$-bundle over a manifold $M$ determines the following data:
        \begin{myenumerate}
            \item an affine $\Lambda^2 \mathbb T^{d,\ast}$-bundle $\hat{P} \rightarrow M$;
            \item an affine 1-gerbe $\hat Q \rightarrow \hat{P}$ with band $\mathbb T^{d,\ast}$;
            \item a $2$-gerbe $\hat R\rightarrow \hat Q$ with band $\sfU(1)$.
        \end{myenumerate}
    \end{proposition}
    \begin{proof}
        Consider cocycle data~\eqref{eq:cocycledatat2rbaff} for a principal $\sfTBRFTA_d$-bundle. The corresponding affine right $T_2$-duality background $(\hat P,\hat Q,\hat R)$ is given by the simplicial manifold $\hat R_\bullet$ with 
        \begin{align}
            \begin{split}
                \hat R_n \coloneqq  \Big\{ &((\hat{v}^0_{i}),(w^1_{ij}),([w^2_{ijk}])) \in Y^{[n+1]}\times (\Lambda^2 \mathbb R^{d,\ast})^{n+1} \times (\mathbb R^{d,\ast})^{\binom{n+1}{2}} \times (\mathbb R/\IZ)^{\binom{n+1}{3}} 
                \\
                &\Big| \: \: [\hat{v}^0_{i} - g_{ij}\hat{v}^0_{j}] = [\hat{v}^0_{ij}]~,
                \\
                &\phantom{\Big|}\: \:[w^1_{ij} - w^1_{ik} + g_{ij} w^1_{jk}] = [w^1_{ijk} - \langle \lambda^0_{ijk}, \hat{v}^0_i \rangle]~,
                \\
                &\phantom{\Big|}\:\: [w^2_{ijk} - w^2_{ijl} + w^2_{ikl} - w^2_{jkl}]
                \\
                &\hspace{1cm}= [w^2_{ijkl} + \langle v^0_{ij},g_{ij}(w^1_{jk} - w^1_{jl} + g_{jk} w^1_{kl}) \rangle 
                \\
                &\hspace{2cm}+ \langle g_{ij}\lambda^0_{jkl}, \langle \lambda^0_{ijl}, \hat{v}^0_i \rangle \rangle^\low - \langle \lambda^0_{ijk}, \langle \lambda^0_{ikl}, \hat{v}^0_i \rangle \rangle^\low ] \Big\}~.
            \end{split}
        \end{align}
        This is indeed a 2-gerbe $\hat R$ over the resolved total space of an affine gerbe with band $\IT^{d,\ast}$ which, in turn, sits over an affine $\Lambda^2\IT^{d,\ast}$-bundle.
    \end{proof}
    
    We now use the differential cocycles for principal $\sfTDFTA_d$-bundles with connections described in \ref{app:adjustment}. 
    \begin{proposition}\label{prop:affine_connections}
        Let $\scP\rightarrow M$ be a principal $\sfTBLFTA_d$-bundle determining a left $T_2$-duality background $(P,R)$ by \ref{prop:left-affine}. A lift to a principal $\sfTDFTA_d$-bundle $\scP_\rmC\rightarrow M$ determines an affine 1-gerbe $Q \rightarrow P$ with band $\mathbb T^{d,\ast}$. 
        
        Let $(\hat P,\hat Q,\hat R)$ be the right $T_2$-duality background corresponding to the induced principal $\sfTBRFTA_d$-bundle, again by \ref{prop:right-affine}. Then a connection on $\scP_\rmC\rightarrow M$ of curvature $(F,\hat{F},H_0,G_0)$ determines the following data:
        \begin{myenumerate}
            \item a connection $A$ on $P \rightarrow M$ of curvature $F$;
            \item a connection $(B,\Lambda)$ on $Q \rightarrow P$ of curvature $H_0 + \langle A ,\hat{F} \rangle$;
            \item a connection $(C,\Sigma,\Xi)$ on $R \rightarrow P$ of curvature $G_0 - \langle A ,H_0 \rangle - \langle A ,\langle A ,\hat{F} \rangle \rangle^\low $;
            \item a connection $\hat{A}$ on $\hat{P} \rightarrow M$ of curvature $\hat{F}$;
            \item a connection $(\hat{B},\hat{\Lambda})$ on $\hat Q\rightarrow \hat{P}$ of curvature $H_0 + \langle F ,\hat{A} \rangle$;
            \item a connection $(\hat{C},\hat{\Sigma},\hat{\Xi})$ on $\hat R\rightarrow \hat Q$ of curvature $(\langle F ,\hat{\Lambda} \rangle, G_0 - \langle F ,\hat{B} \rangle)$.
        \end{myenumerate}
    \end{proposition}
    \begin{proof}
        A connection on the principal $\sfTDFTA_d$-bundle, given by cocycle data~\eqref{eq:cocycledatacont2daff}, determines the connection forms on the left background and $Q$ through the following formulas:
        \begin{align}
            \begin{split}
                A &\coloneqq  \rmd v^0_i + A_i~,
                \\
                B &\coloneqq  B_i - \langle A_i+\rmd v^0_i , \hat{A}_i \rangle~,
                \\
                C &\coloneqq  C_i + \langle \rmd v^0_i +A_i, B_i \rangle - \langle \rmd v^0_i + A_i , \langle \rmd v^0_i + A_i , \hat{A}_i \rangle \rangle^\low~,
                \\
                \Lambda &\coloneqq  \rmd v^1_{ij} - \Lambda_{ij} - \langle \rmd v^0_i, \hat{v}^0_{ij} \rangle~,
                \\
                \Sigma &\coloneqq  -\Sigma_{ij} + \langle \rmd v^0_i , \Lambda_{ij} \rangle + \langle \rmd v^0_i , \langle \rmd v^0_i , \hat{v}^0_{ij} \rangle \rangle^\low~,
                \\
                \Xi &\coloneqq  \rmd v^2_{ijk} - \Xi_{ijk} - \langle \rmd v^0_i, v^1_{ijk} \rangle - \langle v^0_i, \langle \rmd v^0_i, \hat{\lambda}^0_{ijk} \rangle \rangle^\low~.
            \end{split}
        \end{align}
        The connections on the right background are described as follows:
        \begin{align}
            \begin{split}
                \hat{A} &\coloneqq  \rmd \hat v^0_i + \hat{A}_i~,
                \\
                \hat{B} &\coloneqq  B_i + \langle A_i , \rmd\hat{v}^0_i \rangle~,
                \\
                \hat{C} &\coloneqq  C_i - \langle A_i , B_i \rangle~,
                \\
                \hat{\Lambda} &\coloneqq  \rmd w^1_{ij} - \Lambda_{ij} - \langle v^0_{ij},g_{ij}\rmd \hat{v}^0_{j} \rangle~,
                \\
                \hat{\Sigma} &\coloneqq  -\Sigma_{ij} + \langle A_i , \rmd w^1_{ij}- \Lambda_{ij} - 
                \langle v^0_{ij},g_{ij}\rmd\hat{v}^0_j \rangle \rangle + \langle v^0_{ij}, \langle \rmd v^0_{ij} , g_{ij} \rmd\hat{v}^0_{j} \rangle \rangle^\low~,
                \\
                \hat{\Xi} &\coloneqq  \rmd w^2_{ijk} - \Xi_{ijk} - \langle \rmd v^0_{ij}, g_{ij}w^1_{jk} \rangle - \langle v^0_{ij}, \langle \rmd v^0_{ij}, g_{ij}\hat{v}^0_{jk} \rangle \rangle^\low \\
                &\quad - \langle v^0_{ij}, \langle g_{ij}v^0_{jk}, g_{ik}\rmd\hat{v}^0_{k} \rangle \rangle^\low - \langle \lambda^0_{ijk}, \langle v^0_{ik}, g_{ik}\rmd\hat{v}^0_{k} \rangle \rangle^\low~.
            \end{split}
        \end{align}
    \end{proof}
    
    \subsection{Example: Affine \texorpdfstring{$T_2$}{T2}-duality of type \texorpdfstring{$F_2$}{F2} and \texorpdfstring{$d=2$}{d=2} over \texorpdfstring{$\IT^2$}{T2}}\label{ex:t2f2s2}
    
    We consider an affine $\IT^2$-bundle $P$ over $M = \mathbb T^2$ with monodromies around the two cycles of the base $\IT^2$ given by commuting matrices $m_1, \,m_2 \in \sfSL(2,\IZ)$. Let furthermore $(c_1,c_2)\in \IZ^2$ describe the first Chern class of the bundle in a certain trivialization of the local $\IZ^2$-system determined by $m_1, \,m_2$, and let $c_3\in \IZ$ be the Dixmier--Douady class of a 2-gerbe with band $\sfU(1)$ over $P$.
    
    To construct an affine $T_2$-duality correspondence, we give cocycle data~\eqref{eq:cocycledatat2daff} for a principal $\sfTDFTA_d$-bundle with connection over the base $M=\IT^2$, which we cover as follows:
    \begin{align}
        U_1 &:= \{[u_1,u_2] \in \mathbb T^2 \::\: u_1 \neq [0] \}~, 
        \\ 
        U_{2+} &:= \{[u_1,u_2] \in \mathbb T^2 \::\: u_1 \neq [\tfrac12], \, u_2 \neq [0] \}~, 
        \\
        U_{2-} &:= \{[u_1,u_2] \in \mathbb T^2 \::\: u_1 \neq [\tfrac12], \, u_2 \neq [\tfrac12] \}~.
    \end{align}
    We consider functions $u_{1+}:U_{1} \rightarrow \mathbb R$, $u_{1-}:U_{2+}\cup U_{2-} \rightarrow \mathbb R$ and $u_{2\pm}:U_{2\pm} \rightarrow \mathbb R$ similar to those in \ref{ssec:exs2}. The non-zero cocycle data is the following, where we write $(e_1,e_2)$ with $e_1, \,e_2 \in \mathbb R$ for points in $\mathbb R^2$ and we identify $\mathbb R \cong \Lambda^2 (\mathbb R^2)^{\ast}$ as $x \mapsto x \rmd e_1 \wedge \rmd e_2$.
    \begin{equation}
        \begin{aligned}
            g_{1,2+} &= m_1~, 
            \\
            g_{1,2-} &= m_1m_2~, 
            \\
            g_{2+,2-} &= m_2~,
        \end{aligned}
    \end{equation}
    \begin{align}
        v^0_{1,2\pm}([u_1,u_2]) &= \begin{cases} (0,0) & u_{1+} \in (0,\tfrac12)~, \\ (c_1,c_2)u_{2\pm} & u_{1+} \in (\tfrac12,1)~, \end{cases} 
        \\
        \hat{v}^0_{1,2\pm}([u_1,u_2]) &= \begin{cases} 0 & u_{1+} \in (0,\tfrac12)~, \\ c_3u_{2\pm} & u_{1+} \in (\tfrac12,1)~, \end{cases} 
        \\
        \lambda^0_{1,2+,2-} &= \begin{cases} -(c_1,c_2) & u_{1+}, \,u_{2+} \in (\tfrac12,1)~, \\ (0,0) & \text{otherwise}~, \end{cases} 
        \\
        \hat{\lambda}^0_{1,2+,2-} &= \begin{cases} -c_3 & u_{1+}, \,u_{2+} \in (\tfrac12,1)~, \\ 0 & \text{otherwise}~. \end{cases}
    \end{align}
    
    \begin{align}
        A_{1}  &= -(c_1,c_2)u_{1+}\rmd u_2~,              
        & \hat{A}_{1} &= -c_3u_{1+}\rmd u_2~, 
        \\
        A_{2+}  &= -m_1^{-1}(c_1,c_2)u_{1-}\rmd u_2~,                 
        & \hat{A}_{2+} &= -m_1^{-1}c_3u_{1-}\rmd u_2~,
        \\
        A_{2-}  &= -m_2^{-1}m_1^{-1}(c_1,c_2)u_{1-}\rmd u_2~,      
        & \hat{A}_{2-} &= -m_2^{-1}m_1^{-1}c_3u_{1-}\rmd u_2~.
    \end{align}
    The proof of \ref{prop:affine_connections} then yields explicit formulas for the connections on the affine left and right $T_2$-duality backgrounds.
    
    \subsection{Example: Affine \texorpdfstring{$T_2$}{T2}-duality of type \texorpdfstring{$F_2$}{F2} and \texorpdfstring{$d=2$}{d=2} over \texorpdfstring{$\IT^4$}{T4}}
    
    We consider an affine $\IT^2$-bundle $P$ over $M = \mathbb T^4$ with monodromy around one of the cycles of the base $\IT^4$ given by a matrix $m \in \sfSL(2,\IZ)$. Let furthermore $c_i \in \IZ$ for $i=1,\,...,\,6$ describe a collection of Chern and Dixmier--Douady classes. We construct an affine $T_2$-duality correspondence associated to this data, but this does not realize all isomorphism classes of affine $T_2$-duality correspondence spaces over $\mathbb T^4$. We give cocycle data~\eqref{eq:cocycledatat2daff} for a principal $\sfTDFTA_2$-bundle with connection over the base $M=\IT^4$, using the following ordered cover:
    \begin{align}
        U_{1} &:= \{[u_1,u_2,u_3,u_4] \in \mathbb T^4 \: : \: u_1 \neq 0\}~,
        \\
        U_{2+3+} &:= \{[u_1,u_2,u_3,u_4] \in \mathbb T^4 \::\: u_1 \neq \tfrac12, \, u_2 \neq 0, \, u_3 \neq 0\}~, 
        \\
        U_{2-3+} &:= \{[u_1,u_2,u_3,u_4] \in \mathbb T^4 \::\: u_1 \neq \tfrac12, \, u_2 \neq \tfrac12, \, u_3 \neq 0\}~, 
        \\
        U_{2+3-} &:= \{[u_1,u_2,u_3,u_4] \in \mathbb T^4 \::\: u_1 \neq \tfrac12, \, u_2 \neq 0, \, u_3 \neq \tfrac12\}~, 
        \\   
        U_{2-3-} &:= \{[u_1,u_2,u_3,u_4] \in \mathbb T^4 \::\: u_1 \neq \tfrac12, \, u_2 \neq \tfrac12, \, u_3 \neq \tfrac12\}~.      
    \end{align}
    Again, we have $\mathbb R$-valued functions $u_{i\pm}$ for $i=1,\,2,\,3$, similar to those in \ref{ssec:exs2}. We write $(e_1,e_2)$ with $e_1, \,e_2 \in \mathbb R$ for points in $\mathbb R^2$. We identify $\mathbb R^2 \cong \mathbb R^{2,\ast}$ through the standard inner product, and $\mathbb R \cong \Lambda^2 (\mathbb R^2)^{\ast}$ as $x \mapsto x(e_1,e_2) = x(e_1e_2-e_2e_1)$, which also determines an action of $SL(2,\IZ)$ on $\mathbb R$. The non-zero ordered cocycle data is the following (the subscripts $23$, $2'3'$, $2''3''$ are to be replaced by any of $2+3+$, $2-3+$, $2+3-$, $2-3-$, as long as the resulting subindex respects the order of the cover).
    \begin{align}
        g_{1,23} &= m~,
        \\
        v^0_{1,23} &= \begin{cases} (0,0) & u_{1+} \in (0,\tfrac12)~,
            \\ (c_1,c_2)u_2 & u_{1+} \in (\tfrac12,1)~, \end{cases} 
        \\
        \hat{v}^0_{1,23} &= \begin{cases} 0 & u_{1+} \in (0,\tfrac12)~,
            \\ c_3u_3 & u_{1+} \in (\tfrac12,1), \end{cases} 
        \\
        \lambda^0_{1,23,2'3'} &= \begin{cases} (0,0) & u_{1+} \in (0,\tfrac12)~,
            \\ (u_2-u_{2'})(c_1,c_2) & u_{1+} \in (\tfrac12,1)~, \end{cases} 
        \\
        \hat{\lambda}^0_{1,23,2'3'} &= \begin{cases} 0 & u_{1+} \in (0,\tfrac12)~,
            \\ (u_3-u_{3'})c_3 & u_{1+} \in (\tfrac12,1), \end{cases} 
        \\
        v^1_{1,23,2'3'} &= \begin{cases} (0,0) & u_{1+} \in (0,\tfrac12)~,
            \\ (u_2-u_{2'})(c_4,c_5)u_3 & u_{1+} \in (\tfrac12,1)~, \end{cases} 
        \\
        \lambda^1_{1,23,2'3',2''3''} &= \begin{cases} (0,0) & u_{1+} \in (0,\tfrac12)~,
            \\ (u_{2'}-u_{2''})(u_3-u_{3'})(c_4,c_5) & u_{1+} \in (\tfrac12,1)~, \end{cases} 
        \\
        [v^2]_{1,23,2'3',2''3''} &= \begin{cases} [0] & u_{1+} \in (0,\tfrac12)~,
            \\ [(u_{2'}-u_{2''})(u_3-u_{3'})c_6u_4] & u_{1+} \in (\tfrac12,1)~. \end{cases}
    \end{align}
    
    \begin{align}
        A_{1}  &= -(c_1,c_2)u_{1+}\rmd u_2~,
        \\
        A_{23}  &= -m^{-1}(c_1,c_2)u_{1-}\rmd u_2~,
        \\ 
        \hat{A}_{1} &= -c_3u_{1+}\rmd u_3~,
        \\
        \hat{A}_{23} &= -m^{-1}c_3u_{1-}\rmd u_3~,
        \\ 
        \Lambda_{1,23} &= (c_4,c_5)(u_{1+}-u_{1-})u_2\rmd u_3~,
        \\ 
        B_1 &= ((-c_4,-c_5)u_{1+} + (-c_2c_3,c_1c_3)(u_{1+}^2+u_{1+})/2)\rmd u_2 \wedge \rmd u_3~,
        \\ 
        B_{23} &= m^{-1}((-c_4,-c_5)u_{1-} + (-c_2c_3,c_1c_3)(u_{1-}^2+u_{1-})/2)\rmd u_2 \wedge \rmd u_3~,
        \\ 
        \Xi_{1,23,2'3'} &= c_6(u_{1+}-u_{1-})(u_{2'}-u_{2})u_3\rmd u_4~,
        \\ 
        \Sigma_{1,23} &= -c_6(u_{1+}-u_{1-})u_2\rmd u_3 \wedge \rmd u_4~,
        \\ 
        C_1 &= c_6 u_{1+} \rmd u_2 \wedge \rmd u_3 \wedge \rmd u_4~,
        \\ 
        C_{23} &= m^{-1}c_6 u_{1-} \rmd u_2 \wedge \rmd u_3 \wedge \rmd u_4~.
    \end{align}
    Note the curvature of this connection, written over $U_1$, is
    \begin{align}
        F &= (-c_1,-c_2)\rmd u_1 \wedge \rmd u_2~,
        \\ 
        \hat{F} &= -c_3\rmd u_1 \wedge \rmd u_3~,
        \\
        H_0 &= (-c_4-c_2c_3/2,-c_5+c_1c_3/2)\rmd u_1 \wedge \rmd u_2 \wedge \rmd u_3~,
        \\ 
        G_0 &= c_6 \rmd u_1 \wedge \rmd u_2 \wedge \rmd u_3 \wedge \rmd u_4~.
    \end{align}
    The proof of \ref{prop:affine_connections} then yields explicit formulas for the connections on the affine left and right $T_2$-duality backgrounds.

    \section{Reduced \texorpdfstring{$T_2$}{T2}-duality}
    
    Right $T_2$-duality backgrounds involve 2-gerbes on top of 1-gerbes, which may be unfamiliar from string and M-theory. However, if one considers restricted families of left $T_2$-duality backgrounds, then it turns out that their corresponding $T_2$-duals can be reduced to simpler spaces, and previously known constructions, particularly geometric T-duality, arise in this way. 
    
    \subsection{Reduced toroidal 3-groups}
    
    The reduction of $T_2$-duality will be fully implemented by a reduction of the toroidal 3-groups in the spans of toroidal principal 3-bundles. We therefore start with these.
    \begin{definition}[Reduction and reduced toroidal 3-groups]
        A \emph{reduction} is a quadruple of free $\IZ$-modules
        \begin{equation}
            N \subset \IZ^d
            ~,~~~
            \hat N_1 \subset \Lambda^2 \IZ^{d,\ast}
            ~,~~~
            \hat N_2 \subset \IZ^{d,\ast}
            ~,~~~
            \hat N_3 \subset \IZ
        \end{equation}
        such that 
        \begin{equation}
            \iota_v\hat v \in \hat N_2 
            \eand
            \iota_v\lambda \in \hat N_3
        \end{equation}
        for all $v \in N$, $\hat v \in \hat N_1$, and $\lambda \in \hat N_2$.
        
        Restricting the discrete lattices $\Lambda_{0,1,2}$ and the vector spaces $V_{0,1,2}$ in the definition of $\sfTDFT_d$ to 
        \begin{equation}
            \Lambda_0=N\oplus \hat N_1
            ~,~~~
            \Lambda_1=\hat N_2
            ~,~~~
            \Lambda_2=\hat N_3~,~~~V_i=\IR\otimes \Lambda_i
        \end{equation}
        defines the reduced toroidal 3-group $\sfTDFTR_d$, and, via the images of $\phi$ and $\hat \phi$ from~\eqref{eq:span_3_groups}, also reduced toroidal 3-groups $\sfTBLFTR_d$ and $\sfTBRFTR_d$, which fit into the diagram
        \begin{equation}\label{eq:red_diagram}
            \begin{tikzcd}
                \sfTBLFTR_d \ar[d] & \ar[l] \sfTDFTR_d \ar[r] \ar[d]&  \sfTBRFTR_d \ar[d]
                \\
                \sfTBLFT_d & \ar[l] \sfTDFT_d \ar[r] &  \sfTBRFT_d
            \end{tikzcd}
        \end{equation}
    \end{definition}    
    \begin{proposition}\label{prop:reduced}
        Consider a reduction $(N,\hat N_{1,2,3})$. Then the following statements hold:
        \begin{myenumerate}
            \item the horizontal maps in~\eqref{eq:red_diagram} are weak homotopy equivalences;
            \item a principal $\sfTBLFTR_d$-bundle determines a principal $N \otimes \mathbb R / N$-bundle $P \rightarrow M$ and a 2-gerbe $R \rightarrow P$ with band $\hat N_3 \otimes \mathbb R/\hat N_3$;
            \item a principal $\sfTBRFTR_d$-bundle determines a principal $\hat N_1 \otimes \mathbb R / \hat N_1$-bundle $\hat{P} \rightarrow M$, a 1-gerbe $\hat Q\rightarrow \hat P$ with band $\hat N_2 \otimes \mathbb R/\hat N_2$, and a 2-gerbe $\hat R\rightarrow \hat Q$ with band $\hat N_3 \otimes \mathbb R/\hat N_3$. 
        \end{myenumerate} 
    \end{proposition}
    \begin{proof}
        This follows from \ref{cor:f2t2dual}, \ref{prop:left-affine} and \ref{prop:right-affine} together with the straightforward check that the structure maps for the Lie almost 2-crossed modules restrict to the evident subgroups determined by $N$, $\hat N_1$, $\hat N_2$ and $\hat N_3$.
    \end{proof}
    
    The next subsections present concrete examples of reductions. First, however, let us define an affine reduction.
    \begin{definition}[Affine reduction and 3-groups]
        An \emph{affine reduction} is analogously a reduction $(N,\hat N_{1,2,3})$ together with a group
        \begin{equation}
            \sfGL(d,\IZ)^{\rm red} \coloneqq  \{g \in \sfGL(d,\IZ) \:|\: gN \subset N, \,g\hat N_1 \subset \hat N_1, \,g \hat N_2 \subset \hat N_2, \,g\hat N_3 \subset \hat N_3 \}~.
        \end{equation} 
        Substituting the reduction into the group $\sfTDFTA_d$ as above, we obtain reduced 3-groups fitting into the diagram
        \begin{equation}
            \begin{tikzcd}
                \sfTBLFTRA_d \ar[d] & \ar[l] \sfTDFTRA_d \ar[r] \ar[d]&  \sfTBRFTRA_d \ar[d]
                \\
                \sfTBLFTA_d & \ar[l] \sfTDFTA_d \ar[r] &  \sfTBRFTA_d
            \end{tikzcd},
        \end{equation}
    \end{definition}
    \noindent We then evidently have the following.
    \begin{proposition}
        The statements of \ref{prop:reduced} naturally extend to affine reductions.
    \end{proposition}

    \subsection{Embedding of T-duality into \texorpdfstring{$T_2$}{T2}-duality}\label{ssec:embedding_T-duality}
    
    As a first application of reduced $T_2$-duality, let us reproduce geometric T-duality. Consider a splitting $\IZ^d = \IZ^{d-1} \oplus \IZ$ and use it to construct the reduction
    \begin{equation}
        N = \IZ^{d-1}
        ~,~~~
        \hat N_1 = \IZ^{d-1,\ast} \otimes \IZ^{\ast}
        ~,~~~
        \hat N_2 = \IZ^{\ast}
        ~,~~~
        \hat N_3 = 0~.
    \end{equation}
    Then the reduced group $\sfTDFTR_d$ has the following underlying complex:
    \begin{equation}
        \begin{tikzcd}[column sep=0.8cm,row sep=0.2cm]
            \IZ^\ast \arrow[r,"\sft"] 
            & \IZ^{d-1}\times \IZ^{d-1,\ast}\times \IR^\ast \arrow[r,"\sft"] 
            & \IR^{d-1}\times \IR^{d-1,\ast} 
            \\
            \lambda_1\arrow[r,mapsto] & (0,0,\lambda_1)~,
            \\
            &(\lambda_0,\lambda_1,u_1) \arrow[r,mapsto] & (\lambda_0,\lambda_1)~.
        \end{tikzcd}
    \end{equation}
    This group is equivalent to the Lie 2-group $\sfTD^{F_2}_{d-1}$ with complex 
    \begin{equation}
        \begin{tikzcd}[column sep=0.8cm,row sep=0.2cm]
            \IZ^{d-1}\times \IZ^{d-1,\ast}\times \sfU(1) \arrow[r,"\sft"] 
            & \IR^{d-1}\times \IR^{d-1,\ast}~,
            \\
            (\lambda_0,\lambda_1,\varphi) \arrow[r,mapsto] & (\lambda_0,\lambda_1)~,
        \end{tikzcd}
    \end{equation}
    which is the structure 2-group of the principal 2-bundle on top of the span controlling geometric T-duality in~\cite{Nikolaus:2018qop} and~\cite{Waldorf:2022tib}. Together with \ref{prop:reduced}, this implies that geometric T-duality can be embedded as reduced $T_2$-duality as follows.
    \begin{theorem}
        Every geometric T-duality along torus fibers of dimension $d-1$ embeds into a $T_2$-duality along torus fibers of dimension $d$. 
        
        This extends to affine geometric T-duality in the sense of~\cite{Kim:2022opr}, which embeds into affine geometric $T_2$-duality, where 
        \begin{equation}
            \sfGL(d,\IZ)^{\rm red} = \sfGL(d-1,\IZ)~.
        \end{equation}
    \end{theorem}
    \begin{proof}
        Encode the T-duality into a span of principal 2-bundles as in~\cite{Nikolaus:2018qop}, or, in differentially refined form, as in~\cite{Kim:2022opr}. The corresponding principal $\sfTD_{d-1}$-bundle (with connection) on top of the span can trivially be lifted to a principal $\sfTDFT_d$-bundle, which encodes the desired $T_2$-duality. 
        
        The stated reduction of the monodromy group in the affine case follows directly from the chosen splitting, and the remaining argument follows analogously.
    \end{proof}
    
    Let us be slightly more explicit in the construction. A principal $\sfTDFTR_d$-bundle with connection yields curvature forms 
    \begin{equation}
        (F,\hat{F}) \in \Omega^2(M,\mathbb R^{d-1})\oplus\Omega^2(M,\mathbb R^{d-1,\ast} )
        \eand
        H_0 \in \Omega^3(M,\mathbb R^{\ast})~,
    \end{equation}
    satisfying
    \begin{equation}
        \rmd F = 0~,~~~
        \rmd\hat{F}=0~,~~~
        \rmd H_0 + \langle F , \hat{F} \rangle = 0~.
    \end{equation}
    Through the span, it determines the following data:
    \begin{myenumerate}
        \item a principal $\mathbb T^{d-1}$-bundle  $P \rightarrow M$ with connection $A$ of curvature $F$;
        \item a 1-gerbe $Q\rightarrow P$ with band $\sfU(1)$ and connection of curvature $H_0 + \langle A , \hat{F} \rangle$;
        \item a $\mathbb T^{d-1,\ast}$-bundle with connection $\hat P\rightarrow M$ with connection $\hat A$ of curvature $\hat{F}$;
        \item a 1-gerbe $\hat Q\rightarrow \hat P$ with band $\sfU(1)$ and connection of curvature $H_0 + \langle F , \hat{A} \rangle$.\\[-0.5cm]
    \end{myenumerate}
    As indicated in \ref{rk:explainQ}, the $1$-gerbe $Q$ is to be interpreted as a fiber integration, or dimensional reduction, of the 2-gerbe $R$ appearing in the corresponding $T_2$-duality. This is exactly the mechanism by which the Kalb-Ramond field in string theory is expected to arise from a C-field in 11 dimensions, as we elaborate in \ref{sec:rrfields}.
    
    \begin{example}
        As an example of such an embedding, let us specialize the example of \ref{ssec:exs2} by setting $c_1=0$. 
        
        At the level of $T_2$-duality, the left background is a principal $\IT^2$-bundle $P_{(0,c_2)} \rightarrow S^2$ with first Chern class $(0,c_2)\in \IZ^2$, carrying the 2-gerbe $R_{(0,c_2,c_3)} \rightarrow P_{(0,c_2)}$ with Dixmier--Douady class $c_3$. Its $T_2$-dual is given by the $\mathbb T^{\ast}$-bundle $\hat{P}_{c_3} \rightarrow S^2$ with first Chern class $c_3 \in H^2(S^2,\IZ)$, the gerbe $\hat Q_{(-c_2,0)}$ with band $\IT^{2,\ast}$ and Dixmier--Douady class $(-c_2,0) \in \IZ = H^3(\hat{P}_{c_3},\IZ)$, and a $2$-gerbe $\hat{R} \rightarrow \hat{Q}_{(-c_2,0)}$ on top of it.
        
        Now the 2-gerbe $R_{(0,c_2,c_3)} \rightarrow P_{(0,c_2)}$ reduces by fiber integration along the trivial circle fiber to a 1-gerbe $Q_{c_3} \rightarrow P_{c_2}$ with Dixmier--Douady class $c_3$, where $P_{c_2} \rightarrow S^2$ is the $\IT^2$-bundle with Chern class $c_2\in \IZ^2$. By examining the Chern and Dixmier--Douady classes, we see that the data of $\hat{P}_{c_3}$, together with the projection of $\hat{Q}_{(-c_2,0)}$ onto its first $\IT^{2,\ast}$ band is the T-dual of $(P_{c_2},Q_{c_3})$. Thus, indeed we see that this $T_2$-duality determines an underlying $T$-duality.
    \end{example}
    
    \subsection{\texorpdfstring{$T_2$}{T2}-duality of type \texorpdfstring{$F_3$}{F3}}\label{ssec:F3-T2-duality}
    
    As seen in \ref{ssec:characterization}, a left $T_2$-duality background $(M,P,R)$ admitting a $T_2$-dual must satisfy that $c(R) \in F_2H^4(P,\IZ)$. One may wonder what happens when, in fact, $c(R) \in F_3H^4(P,\IZ)$. The result is the following.
    
    \begin{proposition}
        Let $(M,P,R)$ be a left $T_2$-duality background with $c(R) \in F_3H^4(P,\IZ)$. Then, it admits a $T_2$-dual $(M,\hat{P},\hat{Q},\hat{R})$ in which the following is satisfied.
        \begin{enumerate}
            \item $\hat{P} \rightarrow M$ is trivial.
            \item $\hat{Q} \rightarrow \hat{P}$ descends to a 1-gerbe $\hat{Q}^{\rm red} \rightarrow M$ with band $\mathbb T^{d,\ast}$, and $\hat{R} \rightarrow \hat{Q}$ descends to a 2-gerbe $\hat{R}^{\rm red} \rightarrow \hat{Q}^{\rm red}$ with $c(\hat{R}^{\rm red}) \in  F_2H^4(\hat{Q}^{\rm red},\IZ)$.
            \item $\sigma_3(c(R)) = [c(\hat{Q}^{\rm red})]$ and $\sigma_2(c(\hat{R}^{\rm red})) = [c(P)]$.
        \end{enumerate}
    \end{proposition}
    \begin{proof}
        The first two statements follow directly from \eqref{eq:momentumexchange}. The third one follows from reasoning as in the proof of \ref{th:existencet2duals}, and we omit it for brevity.
    \end{proof}
    
    In terms of Lie 3-groups, this phenomenon  can be modeled as follows. Take the reduction 
    \begin{equation}
        N = 0
        ~,~~~
        \hat N_1 = \IZ^d
        ~,~~~
        \hat N_2 = \IZ^{d,\ast}
        ~,~~~
        \hat N_3 = \IZ~.
    \end{equation}
    In this case, a principal $\sfTDFTR_d$-bundle with connection determines curvature forms 
    \begin{equation}
        F \in \Omega^2(M,\mathbb R^d)
        ~,~~~H_0 \in \Omega^3(M,\mathbb R^{d,\ast})
        ~,~~~G_0 \in \Omega^4(M,\mathbb R)
    \end{equation}
    with
    \begin{equation}
        \rmd F = 0
        ~,~~~
        \rmd H_0 =0
        ~,~~~
        \rmd G_0 - \langle F , H_0 \rangle = 0~.
    \end{equation}
    Using the span, it determines the following data:
    \begin{myenumerate}
        \item a principal $\mathbb T^d$-bundle $P\rightarrow M$ with connection $A$ of curvature $F$;
        \item a $2$-gerbe $R\rightarrow P$ with connection of curvature $G_0 - \langle A , H_0 \rangle$;
        \item a 1-gerbe $\hat Q\rightarrow M$ with band $\mathbb T^{d,\ast}$ and connection $(\hat{\Lambda},\hat{B})$ of curvature $H_0$;
        \item a 2-gerbe $\hat R\rightarrow \hat Q$ with connection of curvature $(\langle F , \hat{\Lambda} \rangle, G_0 - \langle F , \hat{B} \rangle)$.
    \end{myenumerate}
    In fact, one can show, similarly as in the proof of \ref{prop:leftbackgroundsf2}, that principal $\sfTDFTR_d$-bundles are in one-to-one correspondence with topological left $T_2$-duality backgrounds of type $F_3$. For this reason, we write 
    \begin{equation}
        \sfTBLFTR_d=\sfTBLFTT_d
        ~,~~~
        \sfTDFTR_d = \sfTDFTT_d
        ~,~~~\sfTBRFTR_d=\sfTBRFTT_d
    \end{equation}
    in this case. For the generalization to the affine setting, we have again $\sfGL(n,\IZ)^{\rm red} = \sfGL(n,\IZ)$.
    
    Let us also give a concrete example of such an embedding.
    \begin{example}
        Let $P = S^3 \times \mathbb T^1\rightarrow S^3$ be the trivial circle bundle\footnote{which is the only topological circle bundle over $S^3$} over $S^3$ and let $R \rightarrow P$ be a $2$-gerbe with Dixmier--Douady class $c \in \IZ \cong H^4(P,\IZ)$. This class lies in the $F_3$ step of the filtration associated to the fibration $P \rightarrow M$. Thus, by the above discussion, we can construct a reduced $T_2$-dual of the form $(\hat Q,\hat R)$, where $\hat Q$ is a 1-gerbe with Dixmier--Douady class $c \in \IZ \cong H^3(S^3,\IZ)$ and $\hat R\rightarrow \hat Q$ is a $2$-gerbe. 
        
        For simplicity, we restrict to the case $c=1$. In the case $d=1$, we have $T_2\mathbb D_1^{F_3} = T_2\mathbb D_1^{F_2}$, so that we may use the cocycle data from~\eqref{eq:cocycledatat2d} and~\eqref{eq:cocycledatacont2d} to describe a principal $T_2\mathbb D_1^{F_3}$-bundle with connection over $S^3$.
        
        Write $S^3 = (\mathbb CP^1 \times [0,1]) / \sim$, where the equivalence relation is $([z_0:z_1],0) \sim ([w_0:w_1],0)$ and $([z_0:z_1],1) \sim ([w_0:w_1],1)$ for $[z_0:z_1], \,[w_0:w_1] \in \mathbb CP^1$. We define an open cover of $S^3$ by the following three sets:
        \begin{align}
            U_1 &\coloneqq  \{[[z_0:z_1],t] \in S^3 \:|\: t \neq 1\}~,
            \\
            U_{0+} &\coloneqq  \{[[z_0:z_1],t] \in S^3 \:|\:z_0 \neq 0, \: t \neq 0\}~,
            \\
            U_{0-} &\coloneqq  \{[[z_0:z_1],t] \in S^3 \:|\:z_1 \neq 0, \: t \neq 0\}~.
        \end{align}
        The non-trivial cocycle data are just the data of the gerbe with connection $\widehat{\mathcal{P}} \rightarrow S^3$:
        \begin{equation}
            \begin{aligned}
                v^1_{0+0-1} &= t \cdot \text{Im} \log(z_1/z_0)/2\pi = t \cdot \log((z_1/z_0)/|z_1/z_0|)/2\pi i~, 
                \\ 
                \Lambda_{0+1} &= t \cdot \frac{-1}{1 + |z_1/z_0|^2} \cdot \rmd\log((z_1/z_0)/|z_1/z_0|)/2\pi i~,
                \\ 
                \Lambda_{0-1} &= t \cdot \frac{-1}{1 + |z_0/z_1|^2} \cdot \rmd\log((z_0/z_1)/|z_0/z_1|)/2\pi i~, 
                \\
                \Lambda_{0+0-} &= \rmd t \cdot \text{Im} \log(z_1/z_0)/2\pi~, 
                \\ 
                B_{0+} &= \rmd t \wedge \frac{1}{1 + |z_1/z_0|^2} \cdot \rmd\log((z_1/z_0)/|z_1/z_0|)/2\pi i~, 
                \\ 
                B_{0-} &= \rmd t \wedge \frac{1}{1 + |z_0/z_1|^2} \cdot \rmd\log((z_0/z_1)/|z_0/z_1|)/2\pi i~, 
                \\ 
                B_1 &= t \cdot \frac{(z_0\rmd z_1-z_1\rmd z_0) \wedge (\overline{z}_0\rmd\overline{z}_1-\overline{z}_1\rmd\overline{z}_0)}{2\pi i(|z_0|^2+|z_1|^2)^2}~.
            \end{aligned}
        \end{equation}
        The proofs of \ref{prop:connection_leftbackgroundsf2} and \ref{prop:connection_rightbackgroundsf2} yield then explicit formulas for the connection forms on the left and right backgrounds.
    \end{example}
    
    
    \subsection{More examples of reductions}
    The next example of reduction is a variation of the case of T-duality, in which the $2$-gerbe winds the additional torus direction but not the T-duality directions.
    \begin{example}\label{ex:dbranetdual}
        Choose any splitting $\IZ^d = \IZ^{d-1} \oplus \IZ$ and use it to define the reduction
        \begin{equation}
            N = \IZ^{d}
            ~,~~~
            \hat N_1 = 0
            ~,~~~
            \hat N_2 = \IZ^{\ast}
            ~,~~~
            \hat N_3 = \IZ~.
        \end{equation}
        A principal $\sfTDFTR_d$-bundle with connection has curvature forms 
        \begin{equation}
            F_{d-1} \in \Omega^2(M,\mathbb R^{d-1})
            ~,~~~
            F_0 \in \Omega^2(M,\mathbb R)
            ~,~~~
            H_0 \in \Omega^3(M,\mathbb R^{\ast})
            ~,~~~
            G_0 \in \Omega^4(M,\mathbb R)
        \end{equation}
        satisfying
        \begin{equation}
            \rmd F_{d-1} = 0
            ~,~~~
            \rmd F_0 =0
            ~,~~~
            \rmd H_0 = 0
            ~,~~~
            \rmd G_0 - \langle F_0 , H_0 \rangle = 0~.
        \end{equation}
        Through the span, it determines the following data:
        \begin{myenumerate}
            \item a principal $\mathbb T^{d-1}$-bundle $P_{d-1}\rightarrow M$ with connection $A_{d-1}$  of curvature $F_{d-1}$;
            \item a principal $\mathbb T^1$-bundle $P_0\rightarrow M$ with connection $A_0$ of curvature $F_{0}$;
            \item a $2$-gerbe $R\rightarrow P_0$ with connection of curvature $G_0 - \langle A_0 , H_0 \rangle$;
            \item a 1-gerbe $\hat Q\rightarrow M$ (with band $\mathbb T^{1,\ast}$) with connection $(\hat{\Lambda},\hat{B})$ of curvature $H_0$;
            \item a 2-gerbe $\hat R\rightarrow \hat Q$ with connection of curvature $(\langle F_0 , \hat{\Lambda} \rangle, G_0 - \langle F , \hat{B} \rangle)$.
        \end{myenumerate}
        An analogous affine case with $\sfGL(d,\IZ)^{\rm red} = \sfGL(d-1,\IZ)$ can also be considered.
    \end{example}

    The following example shows that reduced affine $T_2$-duality constructs a dual space for a pair given by a Siegel bundle (in the sense of~\cite{Lazaroiu:2021vmb}) and an affine $2$-gerbe on top of it.
    \begin{example}
        Let $d = 2n_v$ be even, and let $\omega \in \Lambda^2 \IZ^{d,\ast}$ be such that its $\mathbb R$-linear extension is symplectic.
        We use this to define the reduction
        \begin{equation}
            N = \IZ^{2n_v}
            ~,~~~
            \hat N_1 = \{k\omega \:|\: k \in \IZ\} \cong \IZ
            ~,~~~
            \hat N_2 = \IZ^{2n_v,\ast}
            ~,~~~
            \hat N_3 = \IZ
            ~.
        \end{equation}
        Then a principal $\sfTDFTR_d$-bundle with connection has curvatures
        \begin{equation}
            F \in \Omega^2(M,\mathbb R^{2n_v})
            ~,~~~
            \hat{F} \in \Omega^2(M,\mathbb R)
            ~,~~~
            H_0 \in \Omega^3(M,\mathbb R^{2n_v,\ast})
            ~,~~~
            G_0 \in \Omega^4(M,\mathbb R)
        \end{equation}
        satisfying
        \begin{equation}
            \rmd F = 0
            ~,~~~
            \rmd\hat{F} =0
            ~,~~~
            \rmd H_0 + \langle F , \hat{F} \rangle = 0
            ~,~~~
            \rmd G_0 - \langle F , H_0 \rangle = 0~.
        \end{equation}
        Through the span, it determines the following data:
        \begin{myenumerate}
            \item a principal $\mathbb T^{2n_v}$-bundle $P\rightarrow M$ with connection $A$ of curvature $F$;
            \item a 1-gerbe $Q\rightarrow P$ with band $\mathbb T^{2n_v,\ast}$ and connection of curvature $H_{0} + \langle A , \hat{F} \rangle$;
            \item a $2$-gerbe $R\rightarrow P$ with connection of curvature $G_0 - \langle A , H_0 \rangle$;
            \item a principal $\mathbb T^1$-bundle $\hat P\rightarrow M$ with connection $\hat A$ of curvature $\hat{F}$;
            \item a 1-gerbe $\hat Q\rightarrow \hat P$ with band $\mathbb T^{2n_v,\ast}$ and connection $(\hat{\Lambda},\hat{B}) \rightarrow P$ of curvature $H_0 + \langle F , \hat{A} \rangle$;
            \item a 2-gerbe with connection of curvature $(\langle F , \hat{\Lambda} \rangle, G_0 - \langle F , \hat{B} \rangle)$.
        \end{myenumerate}
        For the affine case, we have here $\sfGL(d,\IZ)^{\rm red} = \sfSp(d,\IZ)$.
    \end{example}
    
    In the following example, left and right $T_2$-duality backgrounds become somewhat more symmetric.
    \begin{example}
        Let $d= 3$ or $d= 7$ and consider the map $\IZ^d \rightarrow \Lambda^2 \IZ^{d,\ast}$ given by $u \mapsto \langle u, -\times -\rangle$ for $\times$ the cross product and $\langle -,-\rangle$ the standard inner product of $\mathbb R^d$. We use this to identify $\IZ^d\subset \Lambda^2 \IZ^{d,\ast}$ and define a reduction
        \begin{equation}
            N = \IZ^d
            ~,~~~
            \hat N_1=\IZ^d\subset \Lambda^2 \IZ^{d,\ast}
            ~,~~~
            \hat N_2=\IZ^{d,\ast}
            ~,~~~
            \hat N_3=\IZ~.
        \end{equation}
        In this case the curvature of a connection on a $T_2\mathbb D_d^{F_2,red}$-bundle can be identified with 
        \begin{equation}
            F \in \Omega^2(M,\mathbb R^d)
            ~,~~~
            \hat{F} \in \Omega^2(M,\mathbb R^d)
            ~,~~~
            H_0 \in \Omega^3(M,\mathbb R^d)
            ~,~~~
            G_0 \in \Omega^4(M,\mathbb R)
        \end{equation}
        with
        \begin{equation}
            \rmd F = 0
            ~,~~~
            \rmd\hat{F} = 0
            ~,~~~
            \rmd H_0 - F \wedge_{\times} \hat{F} = 0
            ~,~~~
            \rmd G_0 - \langle F , H_0 \rangle = 0~,
        \end{equation}
        where $\wedge_{\times}$ is defined in terms of $\times$ and wedge product. In this case left and right $T_2$-duality backgrounds determined by the span are more (but not fully) symmetric:
        \begin{myenumerate}
            \item a principal $\mathbb T^{d}$-bundle $P\rightarrow M$ with connection $A$ of curvature $F$;
            \item a 1-gerbe $Q\rightarrow P$ with band $\mathbb T^{d,\ast}$ and connection of curvature $H_{0} + \langle A , \hat{F} \rangle$;
            \item a $2$-gerbe $R\rightarrow P$ with connection of curvature $G_0 - \langle A , H_0 \rangle$;
            \item a principal $\mathbb T^d$-bundle $\hat P\rightarrow M$ with connection of curvature $\hat{F}$;
            \item a 1-gerbe $\hat Q\rightarrow \hat P$ with band $\mathbb T^{d,\ast}$ and connection $(\hat{\Lambda},\hat{B}) \rightarrow P$ of curvature $H_0 + \langle F , \hat{A} \rangle$;
            \item a 2-gerbe $\hat R\rightarrow \hat Q$ with connection of curvature $(\langle F , \hat{\Lambda} \rangle, G_0 - \langle F , \hat{B} \rangle)$.
        \end{myenumerate}
        For the affine case, we have $\sfGL(3,\IZ)^{\rm red} = \sfSO(3,\IZ)$ and $\sfGL(7,\IZ)^{\rm red} = G_2(\IZ)$.
    \end{example}
    
    \section{Making contact with physics}\label{sec:physics}
    
    After developing the mathematical details of $T_2$-duality, let us now connect our discussion to physical expectations in two ways. First, we explore the parallels between geometric $T_2$-duality for $d=2$ and the dualities linking the bosonic sectors of 11d supergravity (SUGRA) and type IIB SUGRA. We then discuss the generalization of the sigma model derivation of the Buscher rules for $T$-duality to $T_2$-duality. That is, we establish semi-classical physical equivalence between sigma models defined in terms of $T_2$-dual information. In both cases, we restrict ourselves to topologically trivial torus fibrations and $p$-gerbes; the topologically non-trivial case is much more involved and therefore beyond the scope of this paper.
    
    \subsection{The case \texorpdfstring{$d=2$}{d=2}: from 11d to type IIB SUGRA}\label{sec:rrfields}
    
    When restricted to $d=2$, the toroidal 3-group $\sfTDFT_d$ has a torus fiber of the form 
    \begin{equation}
        (\IR^2/\IZ^2)\times (\Lambda^2\IR^{2,\ast})/(\Lambda^2\IZ^{2,\ast}) \cong \IT^2\times \IT^\ast~,
    \end{equation}
    and it is tempting to identify $\IT^2$ with the product of a circular M-theory direction and a T-duality direction in 11d supergravity and $\IT^\ast$ with a dual T-duality direction in type IIB supergravity. 
    
    We first consider the left and right $T_2$-duality backgrounds arising from trivial principal $\sfTDFT_2$-bundles.
    \begin{lemma}\label{lem:trivial_TDF2-bundle}
        A connection on a topologically trivial principal $\sfTDFT_2$-bundle over a $9$-dimensional manifold $M$ is given by forms
        \begin{equation}
            a \in \Omega^1(M,\mathbb R^2)
            ~,~~~
            c^1 \in \Omega^1(M, \Lambda^2 \mathbb R^{2,\ast})
            ~,~~~
            c^2 \in \Omega^2(M,\mathbb R^{2,\ast})
            ~,~~~
            c^3 \in \Omega^3(M,\mathbb R)~.
        \end{equation}
        The corresponding left $T_2$-duality background, together with the additional 1-gerbe arising from the lifting, consists of a trivial 1-gerbe with band $\IT^{2,\ast}$ and a trivial 2-gerbe with band $\Lambda^2 \IT^2\cong \sfU(1)$, both over the trivial $\IT^2$-bundle:
        \begin{equation}
            P=M\times \IT^2
            ~,~~~
            Q=P\times B\IT^{2,\ast}
            ~,~~~
            R=P\times B^2\IT^1~.
        \end{equation}
        These gerbes carry connections which in the notation of~\eqref{eq:conn_1} read as
        \begin{equation}
            \begin{aligned}
                A&\coloneqq  \rmd v^0 + a
                ~,\\
                (B,\Lambda)&\coloneqq  (c^2 - \langle A , c^1\rangle,\rmd v^1)
                ~,\\
                (C,\Sigma,\Xi)&\coloneqq  (c^3 + \langle A , c^2 \rangle - \tfrac{1}{2}\langle A ,\langle A ,c^1 \rangle \rangle , 0 ,\rmd v^2)~,
            \end{aligned}
        \end{equation}
        respectively.
        
        The corresponding right $T_2$-duality background consists of a trivial principal $\Lambda^2\mathbb T^{2,\ast}$-bundle which carries a trivial 1-gerbe with band $\IT^{2,\ast}$ which in turn carries a trivial 2-gerbe with band $\IT^1$:
        \begin{subequations}\label{eq:right_background_85}
            \begin{equation}
                \hat P=M\times \Lambda^2\mathbb T^{2,\ast}
                ~,~~~
                \hat Q=\hat P\times B\IT^{2,\ast}
                ~,~~~
                \hat R=\hat Q\times B^2\IT^1~.
            \end{equation}
            These carry connections which in the notation of~\eqref{eq:conn_2} read as
            \begin{equation}
                \begin{aligned}
                    \hat{A} &\coloneqq  \rmd\hat{v}^0 + c^1~,
                    \\
                    (\hat B,\hat{\Lambda}) &\coloneqq  ( c^2 + \langle a ,\rmd\hat{v}^0 \rangle,\rmd w^1)~,
                    \\
                    (\hat{C},\hat{\Sigma},\hat{\Xi}) &\coloneqq  (c^3 - \tfrac{1}{2}\langle a ,\langle a ,\rmd\hat{v}^0 \rangle \rangle,\langle a ,\rmd w^1 \rangle , \rmd w^2)~. 				
                \end{aligned}
            \end{equation}
        \end{subequations}
    \end{lemma}
    \begin{proof}
        The explicit form of the connections follows directly from the proofs of \ref{prop:connection_leftbackgroundsf2} and \ref{prop:connection_rightbackgroundsf2}. 
    \end{proof}
    
    Let us now compare this connection information with the data obtained from dimensional reduction of a bosonic 11d SUGRA background, following the analysis of~\cite{Bergshoeff:1995as}. The 11d SUGRA background itself consists of an $11$-dimensional manifold $P$, regarded as a trivial principal $\IT^2$-bundle over a manifold $M$ with connection~$A$,
    \begin{subequations}\label{eq:11d_background}
        \begin{equation}
            P=M\times \IT^2~,~~~A = \rmd v + a~.
        \end{equation}
        The manifold $M$ comes with a metric $g^9$, and the total space $P$ is endowed with the usual Kaluza--Klein metric on $\pi_{11d}:M\times \mathbb T^2 \rightarrow M$,
        \begin{equation}
            g = \pi^\ast_{11d}g^9 + h(A \odot A)~,
        \end{equation}
        where $h\in C^{\infty}(M,\odot^2\mathbb R^{2,\ast})$ is related to the dilaton. There is also a 3-form potential $C\in \Omega^3(M\times \IT^2,\IR)$ of the form 
        \begin{equation}
            C = c^3 + \langle A , c^2 \rangle - \tfrac{1}{2}\langle A ,\langle A , c^1 \rangle \rangle
        \end{equation}
        with
        \begin{equation}
            c^1 \in \Omega^1(M, \Lambda^2 \mathbb R^{2,\ast})
            ~,~~~
            c^2 \in \Omega^2(M,\mathbb R^{2,\ast})
            ~,~~~
            c^3 \in \Omega^3(M,\mathbb R)~.
        \end{equation}
    \end{subequations}
    
    We then identify M-theory and type IIA directions and assume that $h$ is diagonal,
    \begin{equation}
        \mathbb T^2 = \mathbb T_{\rm M} \times \mathbb T_{\rm IIA}
        \eand
        h = h_M + h_{\rm IIA}~.
    \end{equation}
    This induces a splitting\footnote{Identifying the M-theory and type IIA directions also splits $\IR^{2,\ast}$ into $\IR^\ast_{\rm IIA}\times \IR^\ast_{\rm M}$.}
    \begin{equation}
        A = A_M+A_{\rm IIA}
        \eand 
        c^2 = c^2_M+c^2_{\rm IIA}~,
    \end{equation}
    and we obtain a reduced type IIA background $\pi_{\rm IIA}: M\times \mathbb T_{\rm IIA}\rightarrow M$ with the following fields:
    \begin{myenumerate}
        \item a Kaluza--Klein metric $g^{l} \coloneqq  \pi^\ast_{\rm IIA} g^9 + h_{\rm IIA}(A_{\rm IIA} \odot A_{\rm IIA})$;
        \item a dilaton $\phi^l$ with $e^{4\phi^l/3} = h_M$;
        \item a Ramond--Ramond (RR) $1$-form $c^{l}_1 \coloneqq  a_M \in \Omega^1(M \times \mathbb T_{\rm IIA},\mathbb R_M^{\ast})$;
        \item a Kalb--Ramond (KR) $2$-form $b^{l}_2 \coloneqq  c^2_M - \langle A_{\rm IIA} ,c^1 \rangle \in \Omega^2(M \times \mathbb T_{\rm IIA}, \mathbb R_M^{\ast})$;
        \item an RR $3$-form $c_3^{l} \coloneqq  c^3 + \langle A_{\rm IIA} , c^2_{\rm IIA} \rangle \in \Omega^3(M \times \mathbb T_{\rm IIA},\mathbb R)$.
    \end{myenumerate}
    
    T-dualizing these fields, we obtain a reduced type IIB background $\pi_{\rm IIB}:M \times \mathbb T_{\rm IIB}\cong M \times \Lambda^2\mathbb T^{2,\ast}$ with the following fields:
    \begin{myenumerate}
        \item a 1-form connection on the bundle $\hat{A} \coloneqq  \rmd\hat{v} + \hat{a}$;
        \item a metric $g^{r} \coloneqq  \pi_{\rm IIB}^\ast g^9+ h_{\rm IIA}^{-1}(\hat{A} \odot \hat{A})$;
        \item a dilaton $\phi^r$ with $e^{2\phi^{r}} = h_{\rm IIA}^{-1}e^{2\phi^l}$;
        \item a KR $2$-form\footnote{This choice differs from the common choice $c^2_M + \langle a_{\rm IIA} , \hat{A} \rangle$, but corresponds to the definition of differential T-duality. Indeed, $b^r_2-b^l_2 = \langle A , \hat{A} \rangle - \langle \rmd v , \rmd\hat{v} \rangle$ and therefore the Poincar\'e line bundle with connection, whose curvature is $\langle \rmd v , \rmd\hat{v} \rangle$, can be understood as an isomorphism between the trivial gerbes with connections $b^l_2$, $b^r_2$, which as such has the desired curvature $\langle A , \hat{A} \rangle$.} $b^{r}_2 \coloneqq  c^2_M + \langle a_{\rm IIA} , \rmd\hat{v} \rangle \in \Omega^2(M \times \mathbb T_{\rm M}^{\ast} \wedge \mathbb T_{\rm IIA}^{\ast}, \mathbb R_M^{\ast})$;
        \item a RR $2$-form $c^{r}_2 \coloneqq  c^2_{\rm IIA} - \langle a_{\rm M} , \rmd\hat{v} \rangle \in \Omega^2(M \times \mathbb T_{\rm M}^{\ast} \wedge \mathbb T_{\rm IIA}^{\ast}, \mathbb R_{\rm IIA}^{\ast})$;
        \item a RR $4$-form $c_4^{r} \coloneqq  c^3 \wedge \rmd\hat{v} \in \Omega^4(M \times \mathbb T_{\rm M}^{\ast} \wedge \mathbb T_{\rm IIA}^{\ast}, \mathbb R_{\rm M}^{\ast} \otimes \mathbb R_{\rm IIA}^{\ast})$.
    \end{myenumerate}
    We note that the $\sfSL(2,\IZ)$-action on $\mathbb T^2$ induces by this construction the S-duality symmetry of the type IIB~field content~\cite{Bergshoeff:1995as,Hull:1994ys}.
    
    Let us now match the field content found in~\ref{lem:trivial_TDF2-bundle} with the one above. 
    \begin{proposition}\label{prop:t2d2trivialtopology}
        Consider the connections and 2-gerbes introduced in \ref{lem:trivial_TDF2-bundle}. Upon splitting $\mathbb T^2 = \mathbb T_M \times \mathbb T_{\rm IIA}$, we observe the following.
        \begin{myenumerate}
            \item The connection form $C$ on $R$ has the form of a supergravity 3-form potential that can be dimensionally reduced to a type IIA background and T-dualized to a type IIB background.
            \item The 1-gerbe $Q$ and the connection it carries descend to $M \times \mathbb T_{\rm IIA}$, where $B_M$ coincides with the IIA KR 2-form.
            \item The torus bundle $P$ is a $\mathbb T^1$-bundle over a type IIA background whose connection form coincides with the RR $1$-form of that background.
            \item The connection $\hat{A}$ on $\hat{P}$ coincides with the 1-form connection $\hat A$ appearing in the type IIB background.
            \item The $\mathbb R_\rmM^{\ast}$-component of the connection form $\hat{B}$ coincides with the type IIB KR 2-form.
            \item The $\mathbb R_{\rm IIA}^{\ast}$-component of the connection form $\hat{B}$ coincides with the type IIB RR 2-form.
            \item The $\sfSL(2,\IZ)$-action defined in \ref{ex:t2f2s2} corresponds to the S-duality action on type~IIB fields under this identification.
        \end{myenumerate}
    \end{proposition}
    \begin{proof}
        The relation with the IIA and IIB field content follows by inspection. The only subtlety  is that, in order to decompose $B$ into its two components, we use isomorphisms $\mathbb R_{\rm IIA}^{\ast} \rightarrow \Lambda^2 \mathbb R^{2,\ast}$, $\rmd v_{\rm IIA} \mapsto \rmd v_{\rm IIA} \wedge \rmd v_\rmM$ and $\mathbb R_{\rm M}^{\ast} \rightarrow \Lambda^2 \mathbb R^{2,\ast}$, $\rmd v_{\rm M} \mapsto \rmd v_{\rm M} \wedge \rmd v_{\rm IIA} = - \rmd v_{\rm IIA} \wedge \rmd v_{\rm M}$, yielding the characteristic minus sign for the RR 2-form.
    \end{proof}
    
    There are two limitations to the above results, which we hasten to add:
    \begin{remark}
        Our above discussion only involves principal bundles, 1-gerbes, and 2-gerbes. As such, it can only capture D$p$-branes with $p=0,\ldots,2$. One can, in principle, also partially capture D$3$-branes, but a full analysis requires working in a ``democratic formulation'' also involving the 5-gerbe with connection 6-form coupling to an M5-brane (as well as the dual graviton)~\cite{Cremmer:1997ct,Godazgar:2013dma,Bandos:1997gd,Bandos:2003et,Bekaert:2002uh}.
    \end{remark}
    \begin{remark}
        Our analysis does not directly translate to the topologically non-trivial sector. For example, Witten's $G_4$-flux quantization condition~\cite{Witten:1996md} shows that it is not the curvature of the 2-gerbe $R$ that should be quantized, i.e.~take values in integral cohomology, but it is its difference from $\tfrac14 p_1$ where $p_1$ is the first Pontryagin class of the frame bundle. This implies that in order to capture general topologically non-trivial configurations, we would also have to include (at least) the group $\sfSO(10,1)$ into the toroidal 3-group $\sfTDFT_2$ to access the above shift by the frame bundle. Our constructions as presented above can therefore only describe configurations with trivial frame bundle.
        
        More generally, however, the full set of conditions to be imposed on an 11d SUGRA background remains under discussion, with a strong contender given by Hypothesis H, see e.g.~\cite{Fiorenza:2020dcp,Sati:2020nob}.
    \end{remark} 
    
    Fixing either of the above restrictions is beyond the scope of this paper and has to be left to future work.
    
    \subsection{Physical equivalence of 3d sigma models with \texorpdfstring{$T_2$}{T2}-dual data}\label{ssec:sigmamodel}
    
    Consider again a topologically trivial bosonic 11d SUGRA background $(M,P,A,g,C)$ as defined in the previous section in~\eqref{eq:11d_background}.\footnote{The arguments below, however, extend to arbitrary $d$.} In addition, consider a 3d Riemannian manifold $\Sigma_3$ with metric $\eta$. Then the bosonic part of the M2-brane sigma model~\cite{Bergshoeff:1987cm,Howe:1977ut} is the action functional\footnote{We have suppressed the cosmological constant term, as it will not participate in our discussion.} 
    \begin{equation}\label{eq:left3dsigmamodel}
        S \coloneqq  \int_{\Sigma_3} \Big( \tfrac12 g(X)(\rmd X,\star_\eta \rmd X)+\rmi X^\ast C\Big)
    \end{equation}
    on maps $X:\Sigma_3\rightarrow P$, where we have suppressed the brane tension and  $\rmd X \in \Gamma(T^{\ast}\Sigma_3 \otimes X^{\ast}TP)$ is the differential of $X$.
    
    We proceed to show how $T_2$-duality relates~\eqref{eq:left3dsigmamodel} to a different but physically equivalent 3d sigma model, following the familiar derivation of the Buscher rules of T-duality in 2d sigma models~\cite{Buscher:1987qj,Fridling:1983ha,Hull:1989jk,Belov:2007qj}. Since the procedure has some interesting new features, such as the appearance of an additional gerbe, we will present the computations in detail.
    
    Concretely, we construct a span of field theories
    \begin{equation}
        \begin{tikzcd}
            & S_\rmP=\int_{\Sigma_3} \caL_\rmP \arrow[dl]\arrow[dr]
            \\
            S=\int_{\Sigma_3} \caL 
            & & \hat S=\int_{\Sigma_3} \hat \caL
        \end{tikzcd}
    \end{equation}
    where $\caL$, $\caL_\rmP$, $\hat \caL$ are the Lagrangians of the physically equivalent actions $S,S_\rmP,\hat S$.\footnote{The arrows describe a homotopy transfer of the underlying $L_\infty$-algebras inducing a semi-classical equivalence between the field theories, see~\cite{JalaliFarahani:2023sfq} for a general discussion. Including also the transformation of the determinant would yield full physical equivalence.}
    
    To be as concrete and transparent as possible, we adopt ``physics notation'' and compute with respect to coordinates $y^i$, $i=1,2,3$, on $\Sigma_3$ and $x^M=(x^\mu,\mathring{x}^a)$ with $\mu=1,\ldots,\rmdim(M)$, and $a=1,2$, on $P$, where $\mathring{x}^a$ parameterize the directions of the $\IT^2$-fibers of $P$. We work with a general metric 
    \begin{equation}
        \tfrac12 g_{MN}\rmd x^M\odot\rmd x^N=\tfrac12 g_{\mu\nu}\rmd x^\mu\odot\rmd x^\nu+g_{\mu a}\rmd x^\mu\odot \rmd \mathring{x}^a+\tfrac12 g_{ab}\rmd \mathring{x}^a\odot \rmd \mathring{x}^b~,
    \end{equation}
    and a general $C$-field
    \begin{equation}
        \begin{aligned}
            \tfrac1{3!} C_{MNK}\rmd X^M\wedge \rmd X^N\wedge \rmd X^K&=
            \tfrac1{3!} C_{\mu\nu\kappa}\rmd x^\mu\wedge \rmd x^\nu\wedge \rmd x^\kappa+\tfrac1{2!} C_{\mu\nu a}\rmd x^\mu\wedge \rmd x^\nu\wedge \rmd \mathring{x}^a
            \\
            &\hspace{1cm}+\tfrac1{2!} C_{\mu ab}\rmd x^\mu\wedge \rmd \mathring{x}^a\wedge \rmd \mathring{x}^b~.
        \end{aligned}
    \end{equation}
    On $\Sigma_3$, we have a metric
    \begin{equation}
        \eta=\tfrac12 \eta_{ij}\rmd y^i\odot \rmd y^j~,
    \end{equation}
    and we use the shorthands 
    \begin{equation}
        \dpar_i\coloneqq \frac{\dpar}{\dpar y^i}~,~~~\dpar^i\coloneqq \eta^{ij}\dpar_j\coloneqq(\eta^{-1})^{ij}\dpar_j~,
    \end{equation}
    with the usual Einstein summation convention applied. Finally, we decompose the sigma model map $X:\Sigma_3\rightarrow P$ into its components $(X^\mu,\mathring{X}^a)$.     
    
    The Lagrangian of the action~\eqref{eq:left3dsigmamodel} then reads
    \begin{equation}
        \begin{aligned}
            \caL&=\tfrac12 g_{\mu\nu}\dpar_i X^\mu \dpar^i X^\nu+g_{\mu a}\dpar_i X^\mu \dpar^i \mathring{X}^a+\tfrac12 g_{ab} \dpar_i \mathring{X}^a \dpar^i \mathring{X}^b+\tfrac\rmi{3!}\eps^{ijk}\dpar_i X^\mu \dpar_j X^\nu \dpar_k X^\kappa C_{\mu\nu\kappa}
            \\
            &\hspace{1cm}+\tfrac\rmi2\eps^{ijk} \dpar_i X^\mu \dpar_j X^\nu \dpar_k \mathring{X}^a C_{\mu\nu a}+\tfrac\rmi2\eps^{ijk}\dpar_i X^\mu \dpar_j \mathring{X}^a \dpar_k \mathring{X}^b C_{\mu ab}~.
        \end{aligned}
    \end{equation}
    The next step in the Buscher procedure is to gauge the symmetry in the fiber directions of $P$. That is, we introduce a connection 1-form $\alpha\in \Omega^1(\Sigma_3,\IR^2)$, with components $\alpha_i^a$, and we replace\footnote{We note that if we gauged only one ``M-theory''-direction, we would reproduce the derivation of the type IIA D2-brane action of~\cite{Townsend:1995af}.}
    \begin{equation}
        \dpar_i~~~\rightarrow~~~D_i
        ~~~\mbox{with}~~~
        D_i \mathring{X}^a\coloneqq\dpar_i \mathring{X}^a+\alpha_i^a
    \end{equation}
    everywhere.\footnote{This procedure works, since we assume that $g$ and $C$ do not depend on the fiber directions; see~\cite{Hull:1989jk} for a full discussion of the required conditions.} This procedure leads to a physically equivalent theory if the curvature of $\alpha$ vanishes,
    \begin{equation}
        f_{ij}^a\coloneqq \dpar_i \alpha_j^a-\dpar_j \alpha_i^a=0~,
    \end{equation}
    and this condition is enforced by adding a Lagrange multiplier $\hat\beta\in \Omega^1(\Sigma_3,\IR^{2,\ast})$. In contrast to the case of the 2d sigma model, we need a 1-form Lagrange multiplier, which can be seen as a coordinate on a gerbe over $M$. 
    
    Altogether, the parent action is then
    \begin{equation}
        \begin{aligned}
            \caL_\rmP&=\tfrac12 g_{\mu\nu}\dpar_i X^\mu \dpar^i X^\nu+g_{\mu a}\dpar_i X^\mu D^i \mathring{X}^a+\tfrac12 g_{ab} D_i \mathring{X}^a D^i \mathring{X}^b+\tfrac\rmi{3!}\eps^{ijk}\dpar_i X^\mu \dpar_j X^\nu \dpar_k X^\kappa C_{\mu\nu\kappa}
            \\
            &\hspace{1cm}+\tfrac\rmi2\eps^{ijk} \dpar_i X^\mu \dpar_j X^\nu D_k \mathring{X}^a C_{\mu\nu a}+\tfrac\rmi2\eps^{ijk}\dpar_i X^\mu D_j \mathring{X}^a D_k \mathring{X}^b C_{\mu ab}+\tfrac12 \eps^{ijk}g_{ab}\hat f_{ij}^a D_k \mathring{X}^b~,
        \end{aligned}
    \end{equation}
    where we have used the identity\footnote{As usual, we assume that standard boundary conditions hold, so that we can integrate by parts.}
    \begin{equation}
        \int_{\Sigma_3} \eps^{ijk}g_{ab}\hat \beta_i^a \dpar_j \alpha_k^b
        =-\int_{\Sigma_3} \eps^{ijk}g_{ab}(\dpar_j \hat \beta_i^a)\alpha_k^b
        \eqcolon\int_{\Sigma_3} \tfrac12 \eps^{ijk}g_{ab}\hat f_{ij}^a \alpha_k^b~.
    \end{equation}
    We have also used the fact that $\tfrac12 \eps^{ijk}g_{ab} \hat f_{ij}^a \dpar_k \mathring{X}^b$ is a total derivative, so that we could further identify
    \begin{equation}
        \int_{\Sigma_3} \tfrac12 \eps^{ijk}g_{ab}\hat f_{ij}^a \alpha_k^b=\int_{\Sigma_3} \tfrac12 \eps^{ijk}g_{ab}\hat f_{ij}^a D_k \mathring{X}^b~.
    \end{equation}
    
    The action $S_\rmP$ is indeed a gauge theory, invariant under the gauge symmetry $\alpha\rightarrow \tilde \alpha=\alpha+\rmd \lambda$ (and $\mathring{X}\rightarrow \tilde{\mathring{X}}=\mathring{X}-\lambda$) for any $\lambda\in C^\infty(\Sigma_3,\IT^2)$. Because the connection $\alpha$ is flat, we can assume it to be pure gauge and use the gauge symmetry to set it to zero. This process yields a physically equivalent field theory, cf.~\cite{JalaliFarahani:2023sfq}, and reduces $S_\rmP$ back to $S$.
    
    The reduction from $S_\rmP$ to $\hat S$ is performed by integrating out the gauge potential $\alpha$. Since this potential appears at most quadratically, we can perform the corresponding Gaussian integral and obtain the semi-classically equivalent field theory $\hat S$. Explicitly, this amounts to solving the equation of motion for $\alpha$ and substituting the solution back into the action. This is also the well-known 3d abelian duality between a compact scalar and a $\sfU(1)$ gauge field.
    
    We note that $\caL_\rmP$ is of the form
    \begin{equation}
        \caL_\rmP=\tfrac12 g_{\mu\nu}\dpar_i X^\mu \dpar^i X^\nu+\tfrac\rmi{3!}\eps^{ijk}\dpar_i X^\mu \dpar_j X^\nu \dpar_k X^\kappa C_{\mu\nu\kappa}+\tfrac12 G_{ia,jb} D^i \mathring{X}^a D^j \mathring{X}^b+D^i\mathring{X}^a R_{ia}~,
    \end{equation}
    where 
    \begin{equation}
        R_{ia}=g_{\mu a}\dpar_i X^\mu+\tfrac\rmi2\eps_{ijk} \dpar^j X^\mu \dpar^k X^\nu C_{\mu\nu a}+\tfrac12 \eps_{ijk}g_{ab}\hat f^{jk b}~.
    \end{equation}
    Varying $S_\rmP$ with respect to $\alpha_i^a$ and assuming that $G_{ia,jb}$ has an inverse $G^{ia,jb}$, we obtain the equation of motion
    \begin{equation}\label{eq:eom}
        D^i \mathring{X}^a=-G^{ia,jb} R_{jb}~.
    \end{equation}
    Substituting this back into $\caL_\rmP$, we obtain the dual Lagrangian 
    \begin{equation}
        \hat \caL=\tfrac12 g_{\mu\nu}\dpar_i X^\mu \dpar^i X^\nu+\tfrac\rmi{3!}\eps^{ijk}\dpar_i X^\mu \dpar_j X^\nu \dpar_k X^\kappa C_{\mu\nu\kappa}-\tfrac12 G^{ia,jb} R_{ia}R_{jb}~.
    \end{equation}
    
    Let us further split
    \begin{equation}
        G_{ia,jb}=\eta_{ij}g_{ab}+r_{ia,jb}~,~~~r_{ia,jb}=\rmi\eps_{ijk}\dpar^k X^\mu C_{\mu ab}~.
    \end{equation}
    Writing $g^{ab}$ for the inverse of $g_{ab}$ and assuming that $r_{ia,jb}$ is small, we then have from the expansion as a geometric series
    \begin{equation}
        G^{ia,jb}=\eta^{ij}g^{ab}-\eta^{ik}g^{ac}r_{kc,ld}\eta^{lj}g^{db}+\ldots~.
    \end{equation}
    This allows us to expand the term involving this inverse metric in powers of $r_{ia,jb}$ or, equivalently, in orders of $C_{\mu ab}$:
    \begin{equation}
        \hat \caL=\tfrac12 g_{\mu\nu}\dpar_i X^\mu \dpar^i X^\nu+\tfrac\rmi{3!}\eps^{ijk}\dpar_i X^\mu \dpar_j X^\nu \dpar_k X^\kappa C_{\mu\nu\kappa}+\sum_{n=0}^\infty\caL_R^{(n)}~.
    \end{equation}
    To lowest order, we have
    \begin{equation}
        \begin{aligned}
            \caL_R^{(0)}&=-\tfrac12 g_{\mu a}g^{ab}g_{\nu b} \dpar_i X^\mu \dpar^i X^\nu+\tfrac\rmi2\eps_{ijk}g_{\mu a}g^{ab}\dpar^i X^\mu \dpar^j X^\nu \dpar^k X^\kappa C_{b\nu\kappa}+\tfrac12\eps_{ijk}g_{\mu a}g^{ab}\dpar^i X^\mu \hat f^{jk}_b
            \\
            &\hspace{0.5cm}
            +\tfrac12 g^{ab}\dpar^j X^\mu \dpar^k X^\nu\dpar_j X^\kappa \dpar_k X^\lambda C_{\mu\nu a}C_{\kappa\lambda b}
            -\rmi\dpar^j X^\nu \dpar^k X^\kappa C_{\nu\kappa a}\hat f_{jk}^a
            -\tfrac12 g_{ab} \hat f_{ij}^a \hat f^{ijb}~.
        \end{aligned}
    \end{equation}
    
    To make contact with our discussion of $T_2$-duality, we now specialize to the left background~\eqref{eq:11d_background}, where $\rmdim(M)=9$. For the metric, we have 
    \begin{equation}
        g_{\mu\nu}=g^9_{\mu\nu}~,~~~g_{\mu a}=h a_{\mu a}~,~~~g_{ab}=h \delta_{ab}~,
    \end{equation}
    and the $C$-field has the components
    \begin{equation}
        C_{\mu\nu\kappa}=c^3_{\mu\nu\kappa}~,~~~C_{\mu\nu a}=c^2_{\mu\nu a}~,~~~C_{\mu ab}=c^1_{\mu ab}~.
    \end{equation}
    We also assumed that $h$ is constant in our derivation. Clearly, the action $S$ is the expected sigma model action, with this left background as a target space. 
    
    The action with Lagrangian $\caL_R^{(0)}$ also shares many of the expected features. First of all, we note that the expected target space of the sigma model is $\hat P\times B \IT^{2,\ast}$, and we would like to match the field content. Recall that a generalized map from a manifold $\Sigma_3$ into $B\IT^{2,\ast}$ is a map from a suitable Čech groupoid resolution of $\Sigma_3$  into $B\IT^{2,\ast}$ and defines a principal $\IT^{2,\ast}$-bundle. After differential refinement, we therefore expect a $\sfLie(\IT^{2,\ast})$-valued connection 1-form, and the fields in the Lagrangian $\caL_R^{(0)}$ indeed include the gauge connection $\hat\beta\in \Omega^1(\Sigma_3,\IR^{2,\ast})$, which appears through its field strength $\hat f_{ij a}$. Moreover, we expect the connection 2-form $\hat B=c^2$ on the gerbe $\hat P\times B\IT^{2,\ast}$ to enter the action as a shift of the field strength $\hat f^a_{ij}$, and this is explicitly visible in the Lagrangian $\caL_R^{(0)}$. It remains to integrate in further, auxiliary fields that account for the map $\mathring{\hat{X}}$ from $\Sigma_3$ into the fibers of $\hat P=M\times \Lambda^2 \IT^{2,\ast}\cong M\times S^1$.
    
    To this end, we first observe that we expect these directions to be non-dynamical (e.g.~from a simple degree count of scalar fields, dualizing connections to scalars). We then consider the term containing $C_{\mu ab}$ in the parent action. Introducing additional fields $\mathring{\hat{X}}_{ab}\in \Omega^0(\Sigma_3,\Lambda^2\IR^{2,\ast})$, an additional gauge connection $\hat \alpha_{ab}\in \Omega^1(\Sigma_3,\Lambda^2\IR^{2,\ast})$, and a metric component $\hat g_{\mu ab}\in \Omega^0(\Sigma_3,\IR^d\otimes \Lambda^2\IR^{2,\ast})$, we can introduce the covariant derivative
    \begin{equation}
        \hat D_i \mathring{\hat{X}}_{ab}\coloneqq\dpar_i \mathring{\hat{X}}_{ab}+\hat \alpha_{i ab}
    \end{equation}
    and consider the Lagrangian 
    \begin{equation}
        \caL_{\rm alt}=\tfrac12 \hat D_i \mathring{\hat{X}}_{ab} \hat D^i \mathring{\hat{X}}^{ab}+\hat g_{\mu ab} \dpar_i X^\mu \hat D^i \mathring{\hat{X}}^{ab}+\hat D_i \mathring{\hat{X}}_{ab} T^{i ab}+\tfrac12 T^{i ab} T_{i ab}+\tfrac12\hat g_{\mu ab} \hat g_\nu^{ab} \dpar_i X^\mu \dpar^i X^\nu~.
    \end{equation}
    The connection $\hat \alpha_{ab}$ appears quadratically and yields the equation of motion
    \begin{equation}
        \hat D_i \mathring{\hat{X}}_{ab}=-\hat g_{\mu ab} \dpar_i X^\mu -T_{i ab}~,
    \end{equation}
    so we can integrate it out, obtaining
    \begin{equation}
        \tilde \caL_{\rm alt}=-T^{iab}\hat g_{\mu ab} \dpar_i X^\mu ~.
    \end{equation}
    For 
    \begin{equation}
        T^{i ab}=\tfrac{\rmi}{2}\eps^{ijk}D_j \mathring{X}^a D_k \mathring{X}^b\eand \hat g_{\mu ab}=C_{\mu ab}~,
    \end{equation}
    this reproduces the term in $\caL_\rmP$ containing $C_{\mu ab}$. Therefore, we can replace the latter term with $\caL_{\rm alt}$ and obtain a natural sigma model action whose target space is the correspondence space $M\times \IT^2\times \Lambda^2\IT^{2,\ast}\times B\IT^{2,\ast}$.\footnote{A shortcoming of this replacement, however, is that the gauge potential $\alpha$, which appears quadratically in $T_{i ab}$ now appears cubically, and can no longer be integrated out.} All fields appear in the expected form; in particular,
    \begin{equation}
        \hat g_{\mu ab}=C_{\mu ab}=c^1_{\mu ab}
    \end{equation}
    now describes the Kaluza--Klein connection coupling the components of $X$ in the base manifold with those in the fibers. We also note that the correction to the metric, 
    \begin{equation}
        g_{\mu\nu}\rightarrow g_{\mu\nu}+\hat g_{\mu ab} \hat g_\nu^{ab}~,
    \end{equation}
    has the expected form for generalized metrics.
    
    We conclude that the topological coupling to $C$ in the original action becomes a coupling to $\hat{C}$ as in \eqref{eq:right_background_85} in the dual action, and that the key to making the computation work is the relation between $C$ and $\hat{C}$ that is axiomatized in the topologically non-trivial setting by \ref{def:t2dualitydiff}. 
    
    Finally, we note that while the above discussion was carried out in the physically most relevant case of two-dimensional torus fibrations in the left $T_2$-duality background, it readily extends to arbitrary torus fibrations. Concretely, had we started from a sigma model for a left $T_2$-duality background containing a principal $\IT^d$-bundle, the gauge potential, and correspondingly its curvature, would have taken values in $\IR^d$, the Lie algebra of $\IT^d$. The Lagrange multiplier would have taken values in $\IR^{d,\ast}$, being the local coordinate of a 1-gerbe with band $\IT^{d,\ast}$, as found in the right $T_2$-duality background.
    
    \appendices
    
    \section{Notation}\label{app:notation}
    
    We collect here several notational conventions used throughout the paper.
    
    \begin{notation}\label{not:torus}
        We write $\mathbb T^d \coloneqq  \mathbb R^d/\IZ^d$ for the standard $d$-torus and $\mathbb T^{d,\ast} \coloneqq  \mathbb R^{d,\ast}/\IZ^{d,\ast}$ for its dual torus. By a slight abuse of notation, we write $\Lambda^k\mathbb T^{d,\ast}$ to denote $\Lambda^k\mathbb R^{d,\ast}/\Lambda^k\IZ^{d,\ast}$.
    \end{notation}
    
    \begin{notation}\label{not:low}
        We identify $\Lambda^2 \IR^{d,\ast}$ with the vector space of skew-symmetric bilinear maps $B: \IR^d\times \IR^d\rightarrow \IR$. For $B \in \Lambda^2 \mathbb R^{d,\ast}$, we write
        \begin{equation}
            B^\low: \IR^d \times \IR^d \rightarrow \IR
        \end{equation}
        for the bilinear map defined on the standard basis $\{e_1,\ldots,e_d\}$ of $\mathbb R^d$ by
        \begin{equation}  
            B^\low(e_i,e_j) \coloneqq  \begin{cases}
                B(e_i,e_j)      & \quad \text{if } \quad i<j  \\
                0 & \quad \text{otherwise },
            \end{cases} 
        \end{equation}
        and $\low$ refers to the lower corner of a matrix. (This is evidently a basis-dependent construction.)
        
        We have the following relations.
        \begin{align}
            B(k_1,k_2) &= B^\low(k_1,k_2) - B^\low(k_2,k_1)~,~~~k_{1,2}\in \IR^d~,
            \\
            (B_1+B_2)^\low &= B_1^\low + B_2^\low~.
        \end{align}
        
        This construction restricts to $\Lambda^2 \IZ^{d,\ast}$, and we also use the notation $\langle -, \langle -, -\rangle \rangle^\low$ for modifying the map 
        \begin{equation}
            \langle -, \langle -,-\rangle \rangle: \mathbb R^d \otimes \mathbb R^d \otimes \Lambda^2 \mathbb R^{d,\ast} \rightarrow \mathbb R~.
        \end{equation}
    \end{notation}
    
    We also stress that we always use additive notation for abelian groups throughout the paper.
    
    \section{Background on simplicial manifolds}\label{app:simplicial}
    
    We briefly review crucial facts about simplicial manifolds used in the main text. Let $[n] \coloneqq  \{0, \,\ldots, \,n\}$ for $n \in \mathbb N$.
    \begin{definition}[Simplicial manifolds and their morphisms, cf.~\cite{BOTT197643}]\label{def:simpman}
        A \emph{simplicial manifold} $M_{\bullet}$ consists of the following data:
        \begin{myenumerate}
            \item for all $n \in \mathbb N$, a manifold of $n$-simplices $M_n$;
            \item for all $n_1,\,n_2 \in \mathbb N$ and every non-decreasing function $f:[n_1] \rightarrow [n_2]$, a smooth map $f^{\ast}:M_{n_2} \rightarrow M_{n_1}$ such that $(f \circ g)^{\ast} = g^{\ast} \circ f^{\ast}$;
        \end{myenumerate}
        A \emph{morphism of simplicial manifolds} $\psi_{\bullet}: M_{\bullet} \rightarrow N_{\bullet}$ is a collection of smooth maps $\psi_n:M_n \rightarrow N_n$ such that for any non-decreasing function $f:[n_1] \rightarrow [n_2]$ we have a commutative diagram
        \begin{equation}
            \begin{tikzcd}
                M_{n_2} \rar["{f^{\ast}_M}"] \dar["{\psi_{n_2}}"] & M_{n_1} \dar["{\psi_{n_1}}"] \\
                N_{n_2} \rar["{f^{\ast}_N}"] & N_{n_1}
            \end{tikzcd}.
        \end{equation}
    \end{definition}
    
    Fiber products and pullbacks are defined degree-wise. Let $\psi^i_{\bullet}: M^i_{\bullet} \rightarrow N_{\bullet}$, $i=1,\,2$ be morphisms with $\psi_n^1$ and $\psi_n^2$ transverse for all $n\in \mathbb N$. The fiber product $(M^1 \times_N M^2)_{\bullet}$ is the simplicial manifold with $(M^1 \times_N M^2)_{n} \coloneqq  M^1_n \times_{N_n} M^2_n$ and pullback maps defined from the pullback maps of $M^1_{\bullet}$ and $M^2_{\bullet}$ in the obvious way.
    
    Given a simplicial manifold $M_{\bullet}$ and $j, \,n \in \mathbb N$ with $0 \leq j \leq n$, we write $d^n_j:M_n \rightarrow M_{n-1}$ and $s^n_j:M_n \rightarrow M_{n+1}$ for $d^n_j = (\delta^n_j)^{\ast}$ and $s^n_j = (\sigma^n_j)^{\ast}$, where $\delta^n_j: [n-1] \rightarrow [n]$ and $\sigma^n_j:[n+1] \rightarrow [n]$ are defined by
    \begin{align}
        \delta^n_j(i) \coloneqq  \begin{cases}
            i   & \mbox{for}~ i < j~,
            \\
            i+1 & \mbox{for}~i \geq j~,
        \end{cases}
        \quad \quad 
        \sigma^n_j(i) \coloneqq  \begin{cases}
            i   & \mbox{for}~ i \leq j~,\\
            i-1 & \mbox{for}~ i > j~.
        \end{cases}
    \end{align}
    For $n \in \mathbb N$, write 
    \begin{align}
        \begin{split}
            |\Delta^n| &\coloneqq  \{(x_0,\ldots,x_n) \in [0,1]^{n+1} \:|\: \textstyle \sum_i x_i = 1\}
        \end{split}
    \end{align}
    for the \emph{geometric $n$-simplex}. For $0 \leq j \leq n$, write 
    \begin{equation}
        d_j^{\Delta}: |\Delta^{n-1}| \rightarrow |\Delta^n|
        \eand
        s_j^{\Delta}: |\Delta^{n+1}| \rightarrow |\Delta^n|
    \end{equation}
    for the continuous maps
    \begin{equation}
        \begin{aligned}
            d_j^{\Delta}(x_0,\ldots,x_{n-1}) &\coloneqq (x_0,\ldots,x_{j-1},0,x_j,\ldots,x_{n-1})~,
            \\
            s_j^{\Delta}(x_0,\ldots,x_{n+1}) &\coloneqq (x_0,\ldots,x_j+x_{j+1},\ldots,x_{n+1})~.
        \end{aligned}
    \end{equation}
    
    \begin{definition}[Geometric realization]
        The \emph{geometric realization} of a simplicial manifold $M_{\bullet}$ is the topological space $|M_{\bullet}|$  defined by
        \begin{equation}
            |M_{\bullet}| \coloneqq  \bigsqcup_{n \in \mathbb N} M_n \times |\Delta^n| \: / \sim~,
        \end{equation}
        where the equivalence relation is
        \begin{align}
            (p,d_j^{\Delta}(x)) &\sim (d_j(p),x)~, \quad \quad p \in M_n, \: x \in |\Delta^{n-1}|, \: \:j = 0,\,\ldots,\,n~, \label{eq:fatgeomreal}
            \\
            (p,s_j^{\Delta}(x)) &\sim (s_j(p),x)~, \quad \quad p \in M_n, \: x \in |\Delta^{n+1}|, \: \:j = 0,\,\ldots,\,n~.
        \end{align}
        The \emph{fat geometric realization} $||M_{\bullet}||$ of $M_{\bullet}$ is similarly defined, but quotienting only by relation~\eqref{eq:fatgeomreal}.
    \end{definition}
    
    A smooth manifold $M$ may trivially be regarded as a constant simplicial manifold $M_{\bullet}$ with $M_n = M$ for all $n \in \mathbb N$ and $f^{\ast} = \sfid$ for all non-decreasing $f:[n_1] \rightarrow [n_2]$. The geometric realization of $M_{\bullet}$ is homeomorphic to $M$ itself. There are also topological spaces $X$ which do not admit a smooth manifold structure, but which are homeomorphic to the geometric realization of some simplicial manifold $X_{\bullet}$. In such cases, topological properties of $X$ can be studied using smooth objects on $X_{\bullet}$. To illustrate this, we introduce the following definition.
    
    \begin{definition}[Simplicial de~Rham double complex]
        Let $M_{\bullet}$ be a simplicial manifold and let $V$ be a vector space. The \emph{simplicial de~Rham double complex} on $M_{\bullet}$ is the $\mathbb N$-bigraded complex
        \begin{equation}
            (\{\Omega^{p}(M_{q},V)\}_{p,q \in \mathbb N},\rmd,\delta)~,
        \end{equation}
        where $\Omega^{p}(M_{q},V)$ denotes the space of $p$-forms on $M_q$, $\rmd: \Omega^p(M_q,V) \rightarrow \Omega^{p+1}(M_q,V)$ is the exterior derivative on $M_q$, and $\delta: \Omega^p(M_q,V) \rightarrow \Omega^{p}(M_{q+1},V)$ is defined by
        \begin{equation}
            \delta\coloneqq  \sum_{j=0}^{q+1} (-1)^j d_j^{q+1,\ast}~.
        \end{equation}
        The \emph{de~Rham cohomology} of $M_{\bullet}$ with coefficients in $V$ is the cohomology $H^{\bullet}_{\rm dR}(M_{\bullet},V)$ of the total complex associated to the double complex of differential forms with differential
        \begin{equation}
            D\omega \coloneqq \delta \omega +(-1)^{p+q+1} \rmd\omega
        \end{equation}
        for $\omega\in \Omega^p(M_q,V)$.
    \end{definition}
    
    \begin{theorem}[\cite{BOTT197643}]\label{th:simpderham}
        Let $V$ be a real vector space and let $M_{\bullet}$ be a (paracompact) simplicial manifold. Then 
        \begin{equation}\label{eq:simpderham}
            H^{\bullet}(||M_{\bullet}||,V) = H^{\bullet}_{\rm dR}(M_{\bullet},V)~,
        \end{equation}
        where the left-hand side of~\eqref{eq:simpderham} denotes singular cohomology of $||M_{\bullet}||$.
    \end{theorem}
    
    \begin{remark}\label{rk:simpderham}
        Versions of \ref{th:simpderham} with coefficients in an arbitrary abelian Lie group can be formulated using sheaf cohomology on simplicial manifolds, see~\cite{BOTT197643,GrzegorzAndrzejczak1989}.
    \end{remark}
    
    \begin{example}\label{ex:bg}
        Let $\sfG$ be a Lie group. We define simplicial manifolds $E\sfG_{\bullet}$ and $B\sfG_{\bullet}$ degree-wise by
        \begin{equation}\label{eq:ordinary_first_delooping}
            \begin{aligned}
                (B\sfG)_n &= \big\{ (g_{ij})_{i < j \in [n]} \in \sfG^{\binom{n+1}{2}} \:| \: \forall i<j<k \in [n], \, g_{ij}g_{jk} = g_{ik}\big\}\cong \sfG^n~,
                \\
                (E\sfG)_n &= \big\{ ((g_i)_{i \in [n]},(g_{ij})_{i < j \in [n]}) \in \sfG^{n+1} \times \sfG^{\binom{n+1}{2}} \:| \: \forall i<j \in [n], \, g_ig_{ij} = g_{j}\big\} \cong \sfG^{n+1}~.
            \end{aligned}
        \end{equation}
        Then, for a non-decreasing function $f:[n_1] \rightarrow [n_2]$, we define
        \begin{equation}\label{eq:simpmapspullback}
            \begin{aligned}
                f^{\ast}:(E\sfG)_{n_2} &\rightarrow (E\sfG)_{n_1}~,
                \\
                ((g_i)_{i \in [n_2]},(g_{ij})_{i < j \in [n_2]}) &\mapsto ((f^{\ast}g_i)_{i \in [n_1]},(f^{\ast}g_{ij})_{i < j \in [n_1]})~,
                \\
                f^{\ast}g_i \coloneqq  g_{f(i)}~, &\quad f^{\ast}g_{ij} \coloneqq  g_{f(i)f(j)}~,
            \end{aligned}
        \end{equation}
        (and similarly for $B\sfG_{\bullet}$), where $g_{f(i)f(j)} = \unit$ if $f(i) = f(j)$. 
        
        We note that there are morphisms of simplicial manifolds
        \begin{align}
            \begin{split}
                \sfG \xrightarrow{~\psi_{\bullet}~} (E\sfG)_{\bullet} \xrightarrow{~\chi_{\bullet}~}(B\sfG)_{\bullet}~,
            \end{split}
        \end{align}
        where $\psi_n(g) = ((g_i)_i,(g_{ij})_{i,j})$ with $g_i = g$ and $g_{ij} = \unit$, while $\chi_n((g_i)_i,(g_{ij})_{i,j}) = (g_{ij})_{i,j}$. The geometric realization of this sequence is the universal bundle of $\sfG$, as constructed for example in~\cite{Hatcher:2002bk}. Together with \ref{th:simpderham}, this model is used in~\cite{BOTT197643} to describe the cohomology of $B\sfG$ in terms of invariant polynomials on the Lie algebra $\mathfrak{g}$ of $\sfG$. Note that the pullback maps of $E\sfG_{\bullet}$ and $B\sfG_{\bullet}$, as well as $\psi_{\bullet}$ and $\chi_{\bullet}$, are less transparent under the identification $(B\sfG)_n \cong \sfG^n$ and $(E\sfG)_n \cong \sfG^{n+1}$.
    \end{example}
    
    \begin{example}[\cite{Li:2309.00901}]\label{ex:bltsimpman}
        Let $\sfT$ be an abelian Lie group and let $p\geq 1$. We define simplicial manifolds $(B^p\sfT)_{\bullet}$ and $(EB^{p-1}\sfT)_{\bullet}$ by
        \begin{subequations}
            \begin{equation}
                \begin{aligned}
                    (B^p\sfT)_n &\coloneqq  \Big\{ (t_{i_0,\ldots,i_p })_{i_0 < i_1 < \ldots< i_p \in [n]} \in \sfT^{\binom{n+1}{p+1}} 
                    \\
                    &\hspace{1cm}\:\Big|\:\forall i_0 < \ldots < i_{p+1} \in [n], \: \textstyle \sum_{j} (-1)^j t_{i_0,\ldots,\hat{i}_j,\ldots,i_{p+1}} = 0 \Big\} \cong \sfT^{\binom{n}{p}}~,
                \end{aligned}
            \end{equation}
            \begin{equation}\label{eq:defblt}
                \begin{aligned}
                    (EB^{p-1}\sfT)_n &\coloneqq  \Big\{ ((t_{i_0,\ldots,i_{p-1}})_{i_0 < i_1 < \ldots< i_{p-1} \in [n]},(t_{i_0,\ldots,i_{p}})_{i_0 < i_1 < \ldots < i_{p} \in [n]}) \in \sfT^{\binom{n+1}{p}} \times \sfT^{\binom{n+1}{p+1}} 
                    \\
                    &\hspace{1cm}\Big|\:\forall i_0 < \ldots < i_{p} \in [n], \: \textstyle \sum_{j} (-1)^j t_{i_0,\ldots,\hat{i}_j,\ldots,i_{p}} = t_{i_0\ldots i_p} \Big\} \cong \sfT^{\binom{n+1}{p}}~,
                \end{aligned}
            \end{equation}
        \end{subequations}
        where we are using additive notation on $\sfT$. The pullback map of $B^p\sfT_{\bullet}$ corresponding to a non-decreasing function $f:[n_1] \rightarrow [n_2]$ is given by
        \begin{align}
            \begin{split}
                (t_{i_0,\ldots,i_l})_{i_0 < i_1 <\ldots< i_l \in [n_2]} 
                &\mapsto (f^{\ast}t_{i_0,\ldots,i_l})_{i_0 < i_1 <\ldots< i_l \in [n_1]}~,
                \\
                f^{\ast}t_{i_0,\ldots,i_l} &\coloneqq  t_{f(i_0),\ldots,f(i_l)}
            \end{split}
        \end{align}
        (and similarly for $EB^{p-1}\sfT_{\bullet}$). We use the convention that a component with repeated indices is zero.
        
        There is a sequence of simplicial manifolds
        \begin{equation}
            B^{p-1}\sfT_{\bullet} \xrightarrow{~~~} EB^{p-1}\sfT_{\bullet} \xrightarrow{~~~} B^p\sfT_{\bullet}~,
        \end{equation}
        and its geometric realization is the universal bundle of a topological group with the same homotopy type as $B^{p-1}\sfT$.
        
        We define the \emph{Maurer--Cartan form} on $B^{p}\sfT_{\bullet}$ to be the differential form 
        \begin{equation}
            t_{i_0,\ldots,i_l}^{\ast}\theta^\sfT \in \Omega^1((B^p\sfT)_{p},\mathfrak{t})~,
        \end{equation}
        where $\mathfrak{t}$ is the Lie algebra of $T$ and $\theta^\sfT \in \Omega^1(\sfT,\mathfrak{t})$ is the Maurer--Cartan form on $\sfT$. Note that this form is $\rmd$-closed, while~\eqref{eq:defblt} implies it is $\delta$-closed; hence by \ref{th:simpderham}, it defines a class in $H^{p+1}_{\rm dR}(B^p\sfT_{\bullet},\mathfrak{t}) = H^{p+1}(B^p\sfT,\mathfrak{t})$. For example, if $\sfT = \mathbb T^d$ is the $d$-torus, then we recall that $H^{\bullet}(B^p\sfT,\mathbb R)$ is the free graded symmetric algebra on $d$ generators of degree $p+1$, and these generators can be represented in de~Rham cohomology by the $d$~components of the Maurer--Cartan form.
    \end{example}

    \section{Differential cohomology for principal \texorpdfstring{$\sfTDFTA_d$}{T2DF2A}-bundles}\label{app:adjustment}
    
    We begin with the BRST groupoid describing local connection forms and (higher) gauge transformations over a local patch $Y$. This is an action groupoid, describing composition of (higher) gauge transformations and their actions. The computations and derivations follow the steps presented in~\cite{Gagliardo:2025oio}, generalized to toroidal 3-groups and specialized to the case of $\sfTDFTA_d$. Some differences in the formulas arise from the redefinition of gauge parameters mentioned in \ref{ssec:toroidal-3-bundles}.
    
    We list only the transformation laws needed below. The remaining composition and whiskering maps can be derived as in~\cite{Gagliardo:2025oio}; as an illustration, we give the composition of 1-morphisms explicitly. A local connection is given by 
    \begin{subequations}
        \begin{equation}
            (A,\hat{A}) \in \Omega^1(Y,\mathbb R^d) \times \Omega^1(Y,\Lambda^2\mathbb R^{d,\ast})
            ~,~~~
            B \in \Omega^2(Y,\mathbb R^{d,\ast})
            ~,~~~
            C \in \Omega^3(Y,\mathbb R)~.
        \end{equation}
        Gauge transformations are parameterized by 
        \begin{equation}
            \begin{aligned}
                (g,v^0,\hat{v}^0)&\in \Omega^0(Y,\sfGL(d,\IZ) \times \mathbb R^d \times \Lambda^2\mathbb R^{d,\ast})~,
                \\
                \Lambda &\in \Omega^1(Y,\mathbb R^{d,\ast})~,
                \\
                \Sigma&\in \Omega^2(Y,\mathbb R)~,
            \end{aligned}
        \end{equation}
        and they act according to
        \begin{equation}
            \begin{aligned}
                \tilde A&=g^{-1}(A+\rmd v^0)~,
                \\
                \tilde{\hat{A}}&=g^{-1}(\hat A+\rmd \hat{v}^0)~,
                \\
                \tilde B&=g^{-1}(B+\rmd\Lambda + \langle A , \rmd\hat{v}^0 \rangle_0)~,
                \\
                \tilde C&= C+\rmd\Sigma + \langle \rmd v^0 , g \tilde B \rangle_1 + \langle A ,  \langle \rmd v^0 , \hat A\rangle_0 \rangle_1^\low~.
            \end{aligned}
        \end{equation}
        These 1-gauge transformations compose as follows:
        \begin{equation}
            (g_2,v^0_2,\hat v^0_2,\Lambda_2,\Sigma_2)\circ (g_1,v^0_1,\hat v^0_1,\Lambda_1,\Sigma_1)=(g_3,v^0_3,\hat v^0_3,\Lambda_3,\Sigma_3)
        \end{equation}
        with
        \begin{equation}
            \begin{aligned}
                (g_3,v^0_3,\hat v^0_3)&=(g_1g_2,v^0_1+g_1v^0_2,\hat v^0_1+g_1\hat v^0_2)~,
                \\
                \Lambda_3&=\Lambda_1+g_1\Lambda_2-\langle \rmd v^0_1,g_1\hat v^0_2\rangle_0~,
                \\
                \Sigma_3&=\Sigma_1+\Sigma_2+\langle \rmd v^0_1,g_1\Lambda_2\rangle_1
                -\langle \rmd v^0_1,\langle \rmd v^0_1,g_1\hat v^0_2\rangle_0\rangle_1^\low~.
            \end{aligned}
        \end{equation}
        
        Next, 2-gauge transformations are parameterized by 
        \begin{equation}
            \begin{aligned}
                (\lambda^0,\hat{\lambda}^0,v^1)&\in \Omega^0(Y,\IZ^d \times \Lambda^2\IZ^{d,\ast} \times \mathbb R^{d,\ast})~,
                \\
                \Xi &\in \Omega^1(Y,\mathbb R)~,
            \end{aligned}
        \end{equation}
        and they act according to
        \begin{equation}
            \begin{aligned}
                (\tilde g,\tilde v^0,\tilde{\hat{v}}^0)&=(g,\lambda^0+v^0,\hat \lambda^0+\hat v^0)~,
                \\
                \tilde \Lambda&=\Lambda+\rmd v^1~,
                \\
                \tilde \Sigma&=\Sigma+\rmd \Xi~.
            \end{aligned}
        \end{equation}
        Finally, the 3-gauge transformations are parameterized by 
        \begin{equation}
            (\lambda^1,[v^2])\in \Omega^0(Y,\IZ^{d,\ast} \times \mathbb R/\IZ)
        \end{equation}
        and act according to
        \begin{equation}
            \begin{aligned}
                (\tilde \lambda^0,\tilde{\hat{\lambda}}^0,\tilde v^1)&=(\lambda^0,\hat \lambda^0,v^1+\lambda^1)~,
                \\
                \tilde \Xi&=\Xi-\rmd v^2 ~.
            \end{aligned}
        \end{equation}
    \end{subequations}
    One can check that these actions and the implied compositions have the required behavior. 
    
    We now pass from local connection data to global differential cocycles by descent. That is, we glue together local connection forms on $Y$ using gauge transformations over $Y^{[2]}$. Over $Y^{[3]}$, we then use 2-gauge transformations to consistently glue together the 1-gauge transformations, etc. In this manner, we recover the Čech data for toroidal principal 3-bundles as presented in \ref{ssec:toroidal-3-bundles} for the toroidal 3-group $\sfTDFTA_d$. 
    
    Explicitly, the differential cocycle data for a principal $\sfTDFTA_d$-bundle over a manifold $M$ subordinate to $\sigma:Y\rightarrow M$ consists of the Čech data
    \begin{subequations}
        \begin{align}
            \begin{split}\label{eq:cocycledatat2daff}
                (g_{ij},v^0_{ij},\hat{v}^0_{ij})&\in\Omega^0(Y^{[2]},\sfGL(d,\IZ) \times \mathbb R^d \times \Lambda^2\mathbb R^{d,\ast})~,
                \\
                (\lambda^0_{ijk},\hat{\lambda}^0_{ijk},v^1_{ijk})&\in\Omega^0(Y^{[3]},\IZ^d \times \Lambda^2\IZ^{d,\ast} \times \mathbb R^{d,\ast})~,
                \\
                (\lambda^1_{ijkl},[v^2_{ijkl}])&\in \Omega^0(Y^{[4]},\IZ^{d,\ast} \times \mathbb R/\IZ)
            \end{split}
        \end{align}
        and the differential refinement
        \begin{equation}\label{eq:cocycledatacont2daff}
            \begin{aligned}
                (A_i,\hat{A}_i) &\in \Omega^1(Y,\mathbb R^d) \times \Omega^1(Y,\Lambda^2\mathbb R^{d,\ast})~,
                \\
                B_i &\in \Omega^2(Y,\mathbb R^{d,\ast})
                ~,~~~
                &\Lambda_{ij} &\in \Omega^1(Y^{[2]},\mathbb R^{d,\ast})~,
                \\
                C_i &\in \Omega^3(Y,\mathbb R)
                ~,~~~
                &\Sigma_{ij} &\in \Omega^2(Y^{[2]},\mathbb R)
                ~,~~~
                &\Xi_{ijk} &\in \Omega^1(Y^{[3]},\mathbb R)~.
            \end{aligned}
        \end{equation}
        These satisfy the equations
        \begin{equation}
            \begin{aligned}
                g_{ij}g_{jk} &= g_{ik}~,
                \\
                v^0_{ij} - v^0_{ik} + g_{ij}v^0_{jk} &= -\lambda^0_{ijk}~,
                \\
                \hat{v}^0_{ij} - \hat{v}^0_{ik} + g_{ij}\hat{v}^0_{jk} &= -\hat{\lambda}^0_{ijk}~,
                \\
                \lambda^0_{ijk}-\lambda^0_{ijl}+\lambda^0_{ikl}-g_{ij}\lambda^0_{jkl}&=0~,
                \\
                \hat\lambda^0_{ijk}-\hat\lambda^0_{ijl}+\hat\lambda^0_{ikl}-g_{ij}\hat\lambda^0_{jkl}&=0~,
                \\
                v^1_{ijk} - v^1_{ijl} + v^1_{ikl} - g_{ij}v^1_{jkl} &= \langle v^0_{ij}, g_{ij}\hat{\lambda}^{0}_{jkl}  \rangle_0 -\lambda^1_{ijkl}~,
                \\
                [v^2_{ijkl} - v^2_{ijkm} + v^2_{ijlm} - v^2_{iklm} + v^2_{jklm}] &
                \\
                =[-\langle v^0_{ij}, g_{ij}\lambda^1_{jklm} \rangle_1 + \langle \lambda^0_{ijk}, &\langle v^0_{ik}, g_{ik}\hat{\lambda}^0_{klm} \rangle_0 \rangle_1^\low + \langle v^0_{ij}, \langle g_{ij} v^0_{jk}, g_{ik}\hat{\lambda}^0_{klm} \rangle_0 \rangle_1^\low]
            \end{aligned}
        \end{equation}
        and 
        \begin{equation}
            \begin{aligned}
                A_j&= g_{ij}^{-1}(A_i+\rmd v^0_{ij})~,
                \\
                \hat A_j&= g_{ij}^{-1}(\hat A_i+\rmd \hat v^0_{ij})~,
                \\
                B_j &= g_{ij}^{-1}(B_i +\rmd\Lambda_{ij} + \langle A_i , \rmd\hat{v}^0_{ij} \rangle_0)~,
                \\
                C_{j} &=C_i + \rmd\Sigma_{ij} + \langle \rmd v^0_{ij} , g_{ij}B_j \rangle_1 + \langle A_i ,  \langle \rmd v^0_{ij}, \hat A_i \rangle_0 \rangle_1^\low~,
                \\
                \Lambda_{ij} - \Lambda_{ik} + g_{ij}\Lambda_{jk} &= \rmd v^1_{ijk} + \langle \rmd v^0_{ij}, g_{ij}\hat{v}^0_{jk} \rangle_0~,
                \\
                \Sigma_{ij} - \Sigma_{ik} + \Sigma_{jk} &= \rmd\Xi_{ijk} - \langle \rmd v^0_{ij} , g_{ij}\Lambda_{jk} \rangle_1 + \langle \rmd v^0_{ij} , \langle \rmd v^0_{ij}, g_{ij}\hat{v}^0_{jk} \rangle_0 \rangle_1^\low~,
                \\
                \Xi_{ijk} - \Xi_{ijl} + \Xi_{ikl} - \Xi_{jkl} &= \rmd v^2_{ijkl} - \langle \rmd v^0_{ij}, g_{ij}v^1_{jkl} \rangle_1 + \langle v^0_{ij}, \langle \rmd v^0_{ij}, g_{ij}\hat{\lambda}^0_{jkl} \rangle_0 \rangle_1^\low~.
            \end{aligned}
        \end{equation}
        
        Write $\mathbb R^{d,\ltimes}$, $\Lambda^2 (\mathbb R^{d,\ast})^{\ltimes}$, and $(\mathbb R^{d,\ast})^{\ltimes}$ for the flat vector bundles associated to the underlying $\sfGL(d,\IZ)$-bundle and the canonical representations on the corresponding vector spaces. The curvature of the connection on the $\sfTDFTA_d$-bundle then decomposes into the forms
        \begin{equation}
            \begin{aligned}
                (F,\hat F) &\coloneqq  (\rmd A_i,\rmd \hat A_i)~~~\in~~\Omega^2(M,\mathbb R^d\times \Lambda^2\mathbb R^{d,\ast})~,
                \\
                H_0 &\coloneqq  \rmd B_i - \langle F, \hat{A}_i \rangle_0~~~\in~~\Omega^3(M,\mathbb R^{d,\ast})~,
                \\
                G_0 &\coloneqq  \rmd C_i + \langle A_i , H_0 \rangle_1+\langle A_i,\langle A_i,\hat F\rangle_0\rangle_1^\low~~~\in~~\Omega^4(M,\mathbb R)~,
            \end{aligned}
        \end{equation}
        which satisfy the Bianchi identities
        \begin{align}
            \begin{split}
                (\rmd F,\rmd\hat{F})&= (0,0)~,
                \\
                \rmd H_0 + \langle F , \hat{F} \rangle_0 &= 0~,
                \\
                \rmd G_0 - \langle F , H_0 \rangle_1 &= 0~.
            \end{split}
        \end{align}
    \end{subequations}
    
    Bundle isomorphisms, or 1-gauge transformations, are parameterized by elements\footnote{To be concise, we switch to a condensed notation $\hat g_i=(g_i,v^0_i,\hat v^0_i)$, etc.}
    \begin{subequations}
        \begin{equation}
            \begin{aligned}
                \hat g_{i}&\in \Omega^0(Y,\sfGL(d,\IZ) \times \mathbb R^d \times \Lambda^2\mathbb R^{d,\ast})~,
                \\
                \hat h_{ij}&\in \Omega^0(Y^{[2]},\IZ^d \times \Lambda^2\IZ^{d,\ast} \times \mathbb R^{d,\ast})
                ~,~~~
                &\Lambda_i &\in \Omega^1(Y,\mathbb R^{d,\ast})~,
                \\
                \hat \ell_{ijk}&\in \Omega^0(Y^{[3]},\IZ^{d,\ast} \times \mathbb R/\IZ)
                ~,~~~
                &\Xi_{ij}&\in \Omega^1(Y^{[2]},\mathbb R)
                ~,~~~
                &\Sigma_{i}&\in \Omega^2(Y,\mathbb R)~,
            \end{aligned}
        \end{equation}
        and they act according to
        \begin{equation}
            \begin{aligned}
                \tilde{\hat{g}}_{ij}&=\hat g_i^{-1}\sft(\hat h_{ij})\hat g_{ij}\hat g_j~,
                \\
                \tilde{\hat{h}}_{ijk}&=\hat g_i^{-1}\acton (\hat h_{ik}+\hat h_{ijk}+\sft(\hat \ell_{ijk})-\hat g_{ij}\acton \hat h_{jk}-\hat h_{ij})~,
            \end{aligned}
        \end{equation}
        \begin{equation}
            \begin{aligned}
                \tilde{\hat{\ell}}_{ijkl}&=\hat g_{i}^{-1}\acton \Big((\hat h_{ij}\,\hat g_{ij}\acton \hat h_{jk})\acton \Big(\hat \ell^{-1}_{ijk}+\hat g_{ij}\acton \hat \ell_{jkl}
                \\
                &\hspace{2cm}\times((\hat g_{ij}\hat g_{jk})\acton \hat h_{kl})\acton \Big((\hat h_{ijk}^{-1}\acton \hat\ell^{-1}_{ikl})+\hat \ell_{ijkl}+\left(\hat g_{ij}\acton \hat h_{jkl}^{-1}\right)\acton \hat \ell_{ijl}\Big)\Big)
                \\
                &\hspace{2cm}-\left\{\hat h_{ij},\hat g_{ij}\acton \hat h_{jl}\right\}-\left\{\hat h_{ij},\hat g_{ij}\acton \hat h_{jkl}\right\}+\left\{\hat h_{ij},(\hat g_{ij}\hat g_{jk})\acton \hat h_{kl}\right\}
                \\
                &\hspace{2cm}\left.+\left\{\hat h_{ij},\hat g_{ij}\acton \hat h_{jk}\right\}\right)-\left\{\hat h_{ijk},(\hat g_{ij}\hat g_{jk})\acton \hat h_{kl}\right\}\Big)+R_{ijkl}~,
            \end{aligned}
        \end{equation}
        and
        \begin{equation}
            \begin{aligned}
                \tilde A_i&=g^{-1}_i(A_i+\rmd v^0_i)~,
                \\
                \tilde{\hat{A}}_i&=g^{-1}_i(\hat A_i+\rmd \hat{v}_i^0)~,
                \\
                \tilde B_i&=g^{-1}_i(B_i+\rmd\Lambda_i + \langle A_i , \rmd\hat{v}^0_i \rangle_0)~,
                \\
                \tilde C_i&= C_i+\rmd\Sigma_i + \langle \rmd v^0_i , g_i \tilde B_i \rangle_1 + \langle A_i ,  \langle \rmd v^0_i,\hat A_i  \rangle_0 \rangle_1^\low~,
                \\
                \tilde \Lambda_{ij}&=g_i^{-1}(g_{ij}\Lambda_j+\Lambda_{ij}-\Lambda_i)+\rmd v^1_{ij}+\langle g_i^{-1}\rmd v^0_{i},\hat v^0_{ij}\rangle_0-\tilde g_{ij}\langle (g_{ij}g_j)^{-1}\rmd v^0_{ij},g_j^{-1}\hat v^0_j\rangle_0~,
                \\
                \tilde \Sigma_{ij}&=\Sigma_j+\Sigma_{ij}-\Sigma_i+\rmd \Xi_{ij}-
                \langle \rmd v^0_{ij} , g_i\Lambda_{ij} \rangle_1 + \langle \rmd v^0_{ij}, \langle \rmd v^0_{ij}, g_i\hat{v}^0_{ij} \rangle_0 \rangle_1^\low~,
                \\
                \tilde \Xi_{ijk}&=\Xi_{ik}+\Xi_{ijk}-\Xi_{ij}-\Xi_{jk}-\rmd v^2_{ijk} +\langle \rmd v^0_{ij}, g_{ij} v^1_{jk} \rangle_1 -\langle v^0_{ij}, \langle \rmd v^0_{ij}, g_{ij}\hat{\lambda}^0_{jk} \rangle_0 \rangle_1^\low
            \end{aligned}
        \end{equation}
        for all $ij\in Y^{[2]}$, $ijk\in Y^{[3]}$, and $ijkl\in Y^{[4]}$, where 
        \begin{equation}
            \begin{aligned}
                R_{ijkl}&=\alpha(\sft(\hat h_{ij}),\hat g_{ij},\hat g_j\acton\tilde{\hat{h}}_{jkl})+\alpha(\hat g_i^{-1},\sft(\hat h_{ij}),\hat g_j\acton \tilde{\hat{h}}_{jkl})-\alpha(\tilde{\hat {g}}_{ij},\hat g_j^{-1},\hat g_j\acton \tilde{\hat{h}}_{jkl})
                \\
                &\hspace{1cm}+\alpha(\sft(\hat h_{ijk}),\hat g_{ij}\hat g_{jk},\hat h_{kl})+\alpha(\hat g_{ij},\hat g_{jk},\hat h_{kl})~.
            \end{aligned}
        \end{equation}
    \end{subequations}
    
    Bundle 2-isomorphisms, or 2-gauge transformations, are parameterized by elements 
    \begin{subequations}
        \begin{equation}
            \begin{aligned}
                \hat h_{i}&\in \Omega^0(Y,\IZ^d \times \Lambda^2\IZ^{d,\ast} \times \mathbb R^{d,\ast})~,
                \\
                \hat \ell_{ij}&\in \Omega^0(Y^{[2]},\IZ^{d,\ast} \times \mathbb R/\IZ)
                ~,~~~
                &\Xi_{i}&\in \Omega^1(Y,\mathbb R)~,
            \end{aligned}
        \end{equation}
        and they act according to 
        \begin{equation}
            \begin{aligned}
                \tilde{\hat{g}}_{i}&=\sft(\hat h_{i})\hat g_i~,
                \\
                \tilde{\hat{h}}_{ij}&=\hat h_i+\hat h_{ij}+\sft(\hat \ell_{ij})-\hat g_{ij}\acton \hat h_j~,
            \end{aligned}
        \end{equation}
        \begin{equation}
            \begin{aligned}
                \tilde{\hat{\ell}}_{ijk}&=((\hat g_{ij}\hat g_{jk})\acton \hat h_c)\acton \Big( \hat h_{ijk}^{-1}\acton
                \hat \ell_{ik}^{-1}+\hat \ell_{ijk}+(\hat g_{ij}\acton \hat h_{jk}^{-1})\acton\hat \ell_{ij}+\hat g_{ij}\acton \hat \ell_{jk}\,\Big)            
                \\
                &\hspace{0.5cm}-\{\hat h_{ijk},(\hat g_{ij}\hat g_{jk})\acton \hat h_k\}-\{\hat h_i,\hat h_{ik}\}
                -\{\hat h_i,\hat h_{ijk}\}+\{\hat h_i,\hat g_{ij}\acton\hat h_{jk}\}
                +\{\hat h_i,\hat h_{ij}\}+R_{ijk}~,
            \end{aligned}
        \end{equation}
        and
        \begin{equation}
            \begin{aligned}
                \tilde \Lambda_i&=\Lambda_i+\rmd v^1_i+\langle \rmd v^0_i,g_i\hat v^0_i\rangle_0~,
                \\
                \tilde \Sigma_i&=\Sigma_i+\rmd \Xi_i-
                \langle \rmd v^0_i , g_i\Lambda_i \rangle_1 + \langle \rmd v^0_i, \langle \rmd v^0_i, g_i\hat{v}^0_i \rangle_0 \rangle_1^\low~,
                \\
                \tilde \Xi_{ij}&=\Xi_{ij}+\Xi_i-\Xi_j-\rmd v^2_{ij} +\langle \rmd v^0_{ij}, g_{ij} v^1_j \rangle_1 -\langle v^0_{ij}, \langle \rmd v^0_{ij}, g_{ij}\hat{\lambda}^0_j \rangle_0 \rangle_1^\low
            \end{aligned}
        \end{equation}
        for all $i\in Y$, $ij\in Y^{[2]}$, and $ijk\in Y^{[3]}$, where 
        \begin{equation}
            R_{ijk}=\alpha(\hat g_i^{-1},\sft(-\hat h_i),-\hat g_i\acton \tilde{\hat{h}}_{ijk})+\alpha(\sft(\hat h_{ijk}),\hat g_{ij}\hat g_{jk},\hat h_k)+\alpha(\hat g_{ij},\hat g_{jk},\hat h_k)~.
        \end{equation}
    \end{subequations}
    
    Finally, bundle 3-isomorphisms or 3-gauge transformations are parameterized by elements 
    \begin{subequations}
        \begin{align}
            \begin{split}
                \hat \ell_{i}&\in \Omega^0(Y,\IZ^{d,\ast} \times \mathbb R/\IZ)~,
            \end{split}
        \end{align}
        and act according to 
        \begin{equation}
            \begin{aligned}
                \tilde{\hat{h}}_{i}&=\sft(\hat\ell_{i})+\hat h_i~,
                \\
                \tilde{\hat{\ell}}_{ij}&=-(-\hat h_i-\hat h_{ij})\acton \hat \ell_i+\hat \ell_{ij}+(-\hat g_{ij}\acton \hat h_j)\acton (\hat g_{ij}\acton \hat \ell_j)~,
                \\
                \tilde \Xi_i&=\Xi_i-\rmd v^2_i
            \end{aligned}
        \end{equation}
        for all $i\in Y$ and $ij\in Y^{[2]}$.
    \end{subequations}
    
    \section*{Acknowledgements} 
    
    GG has been supported by an STFC studentship ST/Y509206/1. Part of this work
    was carried out while RTD was affiliated with the Instituto de Ciencias Matemáticas and was supported by grant SEV-2015-0554-19-2 from Agencia Estatal de
    Investigación de España
    
    \section*{Declarations}
    
		\textbf{Conflict of interest.}
		The authors have no relevant financial or non-financial interests to disclose.\\[5pt]
		\textbf{Data statement.}
		No additional research data beyond the data presented and cited in this work are needed to validate the research findings in this work.\\[5pt]
		\textbf{Licence statement.}
		For the purpose of open access, the authors have applied a Creative Commons Attribution (CC-BY) license to any author-accepted manuscript version arising.
	
    \bibliographystyle{latexeu2}
    
    \bibliography{bigone}
    
\end{document}